\documentclass[english,11pt, letter]{article}
\usepackage[T1]{fontenc}
\usepackage[latin9]{inputenc}
\usepackage{geometry}
\usepackage{amsmath}
\usepackage{amsthm}
\usepackage{amssymb}

\makeatletter

\providecommand{\tabularnewline}{\\}

\numberwithin{equation}{section}
\numberwithin{figure}{section}
\theoremstyle{plain}
\newtheorem{thm}{\protect\theoremname}
\theoremstyle{definition}
\newtheorem{defn}[thm]{\protect\definitionname}
\theoremstyle{plain}
\newtheorem{lem}[thm]{\protect\lemmaname}
\theoremstyle{plain}
\newtheorem{cor}[thm]{\protect\corollaryname}

\pdfoutput=1

\pdfoutput=1

\usepackage{amsmath, amsthm, amssymb, amsfonts,bbold}
\usepackage[basic]{complexity}
\usepackage[pdftex,bookmarks,colorlinks]{hyperref}
\usepackage{enumitem}
\usepackage{color}

\DeclareMathOperator*{\argmaxTex}{arg\,max}
\DeclareMathOperator*{\argminTex}{arg\,min}
\DeclareMathOperator*{\signTex}{sign}
\DeclareMathOperator*{\rankTex}{rank}
\DeclareMathOperator*{\diagTex}{diag}
\DeclareMathOperator*{\imTex}{im}

\hypersetup{colorlinks=true,citecolor=red}

\renewcommand{\varepsilon}{\epsilon}

\usepackage{thmtools}
\usepackage{thm-restate}

\usepackage[lined,boxed,ruled,norelsize,algo2e]{algorithm2e}

\usepackage[small,compact]{titlesec}
\usepackage{setspace}
\usepackage{subfigure}

\makeatother

\usepackage{babel}
\providecommand{\corollaryname}{Corollary}
\providecommand{\definitionname}{Definition}
\providecommand{\lemmaname}{Lemma}
\providecommand{\theoremname}{Theorem}

\begin{document}

\global\long\def\R{\mathbb{R}}%

\global\long\def\Rn{\mathbb{R}^{n}}%

\global\long\def\Rm{\mathbb{R}^{m}}%

\global\long\def\Rmn{\mathbb{R}^{m \times n}}%

\global\long\def\Rnm{\mathbb{R}^{n \times m}}%

\global\long\def\Rmm{\mathbb{R}^{m \times m}}%

\global\long\def\Rnn{\mathbb{R}^{n \times n}}%

\global\long\def\Z{\mathbb{Z}}%

\global\long\def\Rp{\R_{> 0}}%

\global\long\def\dom{\mathrm{dom}}%

\global\long\def\dInterior{K}%

\global\long\def\Rpm{\R_{> 0}^{m}}%


\global\long\def\ellOne{\ell_{1}}%
 
\global\long\def\ellTwo{\ell_{2}}%
 
\global\long\def\ellInf{\ell_{\infty}}%
 
\global\long\def\ellP{\ell_{p}}%

\global\long\def\otilde{\widetilde{O}}%

\global\long\def\argmax{\argmaxTex}%

\global\long\def\argmin{\argminTex}%

\global\long\def\sign{\signTex}%

\global\long\def\rank{\rankTex}%

\global\long\def\diag{\diagTex}%

\global\long\def\im{\imTex}%

\global\long\def\enspace{\quad}%

\global\long\def\mvar#1{\mathbf{#1}}%

\global\long\def\vvar#1{#1}%



\global\long\def\defeq{\stackrel{\mathrm{{\scriptscriptstyle def}}}{=}}%

\global\long\def\diag{\mathrm{diag}}%

\global\long\def\mDiag{\mvar{Diag}}%
 
\global\long\def\ceil#1{\left\lceil #1 \right\rceil }%

\global\long\def\E{\mathbb{E}}%

\global\long\def\jacobian{\mvar J}%

\global\long\def\onesVec{1}%
 
\global\long\def\indicVec#1{1_{#1}}%
\global\long\def\cordVec#1{e_{#1}}%
\global\long\def\op{\mathrm{spe}}%

\global\long\def\va{\vvar a}%
 
\global\long\def\vb{\vvar b}%
 
\global\long\def\vc{\vvar c}%
 
\global\long\def\vd{\vvar d}%
 
\global\long\def\ve{\vvar e}%
 
\global\long\def\vf{\vvar f}%
 
\global\long\def\vg{\vvar g}%
 
\global\long\def\vh{\vvar h}%
 
\global\long\def\vl{\vvar l}%
 
\global\long\def\vm{\vvar m}%
 
\global\long\def\vn{\vvar n}%
 
\global\long\def\vo{\vvar o}%
 
\global\long\def\vp{\vvar p}%
 
\global\long\def\vs{\vvar s}%
 
\global\long\def\vu{\vvar u}%
 
\global\long\def\vv{\vvar v}%
 
\global\long\def\vx{\vvar x}%
 
\global\long\def\vy{\vvar y}%
 
\global\long\def\vz{\vvar z}%
 
\global\long\def\vxi{\vvar{\xi}}%
 
\global\long\def\valpha{\vvar{\alpha}}%
 
\global\long\def\veta{\vvar{\eta}}%
 
\global\long\def\vphi{\vvar{\phi}}%
\global\long\def\vpsi{\vvar{\psi}}%
 
\global\long\def\vsigma{\vvar{\sigma}}%
 
\global\long\def\vgamma{\vvar{\gamma}}%
 
\global\long\def\vphi{\vvar{\phi}}%
\global\long\def\vDelta{\vvar{\Delta}}%
\global\long\def\vzero{\vvar 0}%

\global\long\def\ma{\mvar A}%
 
\global\long\def\mb{\mvar B}%
 
\global\long\def\mc{\mvar C}%
 
\global\long\def\md{\mvar D}%
 
\global\long\def\mf{\mvar F}%
 
\global\long\def\mg{\mvar G}%
 
\global\long\def\mh{\mvar H}%
 
\global\long\def\mj{\mvar J}%
 
\global\long\def\mk{\mvar K}%
 
\global\long\def\mm{\mvar M}%
 
\global\long\def\mn{\mvar N}%
 
\global\long\def\mq{\mvar Q}%
 
\global\long\def\mr{\mvar R}%
 
\global\long\def\ms{\mvar S}%
 
\global\long\def\mt{\mvar T}%
 
\global\long\def\mU{\mvar U}%
 
\global\long\def\mv{\mvar V}%
 
\global\long\def\mx{\mvar X}%
 
\global\long\def\my{\mvar Y}%
 
\global\long\def\mz{\mvar Z}%
 
\global\long\def\mSigma{\mvar{\Sigma}}%
 
\global\long\def\mLambda{\mvar{\Lambda}}%
\global\long\def\mPhi{\mvar{\Phi}}%
 
\global\long\def\mZero{\mvar 0}%
 
\global\long\def\iMatrix{\mvar I}%
\global\long\def\mi{\mvar I}%
\global\long\def\mDelta{\mvar{\Delta}}%

\global\long\def\oracle{\mathcal{O}}%
 
\global\long\def\mw{\mvar W}%

\global\long\def\runtime{\mathcal{T}}%


\global\long\def\mProj{\mvar P}%

\global\long\def\vLever{\sigma}%
 
\global\long\def\mLever{\mSigma}%
 
\global\long\def\mLapProj{\mvar{\Lambda}}%

\global\long\def\penalizedObjective{f_{t}}%
 
\global\long\def\penalizedObjectiveWeight{{\color{red}f}}%

\global\long\def\fvWeight{\vg}%
 
\global\long\def\fmWeight{\mg}%

\global\long\def\vNewtonStep{\vh}%

\global\long\def\norm#1{\|#1\|}%
 
\global\long\def\normFull#1{\left\Vert #1\right\Vert }%
 
\global\long\def\normA#1{\norm{#1}_{\ma}}%
 
\global\long\def\normFullInf#1{\normFull{#1}_{\infty}}%
 
\global\long\def\normInf#1{\norm{#1}_{\infty}}%
 
\global\long\def\normOne#1{\norm{#1}_{1}}%
 
\global\long\def\normTwo#1{\norm{#1}_{2}}%
 
\global\long\def\normLeverage#1{\norm{#1}_{\mSigma}}%
 
\global\long\def\normWeight#1{\norm{#1}_{\fmWeight}}%

\global\long\def\cWeightSize{c_{1}}%
 
\global\long\def\cWeightStab{c_{\gamma}}%
 
\global\long\def\cWeightCons{{\color{red}c_{\delta}}}%

\global\long\def\TODO#1{{\color{red}TODO:\text{#1}}}%
\global\long\def\mixedNorm#1#2{\norm{#1}_{#2+\infty}}%

\global\long\def\mixedNormFull#1#2{\normFull{#1}_{#2+\infty}}%
\global\long\def\CNorm{C_{\mathrm{norm}}}%
\global\long\def\Pxw{\mProj_{\vx,\vWeight}}%
\global\long\def\vq{q}%
\global\long\def\cnorm{\CNorm}%

\global\long\def\next#1{#1^{\mathrm{(new)}}}%

\global\long\def\trInit{\vx^{(0)}}%
 
\global\long\def\trCurr{\vx^{(k)}}%
 
\global\long\def\trNext{\vx^{(k + 1)}}%
 
\global\long\def\trAdve{\vy^{(k)}}%
 
\global\long\def\trAfterAdve{\vy}%
 
\global\long\def\trMeas{\vz^{(k)}}%
 
\global\long\def\trAfterMeas{\vz}%
 
\global\long\def\trGradCurr{\grad\Phi_{\alpha}(\trCurr)}%
 
\global\long\def\trGradAdve{\grad\Phi_{\alpha}(\trAdve)}%
 
\global\long\def\trGradMeas{\grad\Phi_{\alpha}(\trMeas)}%
 
\global\long\def\trGradAfterAdve{\grad\Phi_{\alpha}(\trAfterAdve)}%
 
\global\long\def\trGradAfterMeas{\grad\Phi_{\alpha}(\trAfterMeas)}%
 
\global\long\def\trSetCurr{U^{(k)}}%
\global\long\def\vWeightError{\vvar{\Psi}}%
\global\long\def\code#1{\texttt{#1}}%

\global\long\def\nnz{\mathrm{nnz}}%
\global\long\def\tr{\mathrm{tr}}%
\global\long\def\vones{\vec{1}}%

\global\long\def\volPot{\mathcal{V}}%
\global\long\def\grad{\mathcal{\nabla}}%
\global\long\def\hess{\nabla^{2}}%
\global\long\def\hessian{\nabla^{2}}%

\global\long\def\shurProd{\circ}%
 
\global\long\def\shurSquared#1{{#1}^{(2)}}%
\global\long\def\trans{\top}%

\global\long\def\lpweight{w_{p}}%
\global\long\def\mlpweight{\mw_{p}}%
\global\long\def\lqweight{w_{q}}%
\global\long\def\mlqweight{\mw_{q}}%

\newcommand{\bracket}[1]{[#1]}
\global\long\def\interiorPrimal{\Omega^{\circ}}%
\global\long\def\interiorDual{\Omega^{\circ}}%

\pagenumbering{gobble}
\title{Solving Linear Programs with $\otilde(\sqrt{\rank})$ Linear System
Solves}
\author{Yin Tat Lee\\
University of Washington and Microsoft Research\\
\texttt{yintat@uw.edu}\and  Aaron Sidford\\
Stanford University\\
\texttt{sidford@stanford.edu}}
\date{}

\maketitle
\global\long\def\vol{\mathrm{vol}}%
\global\long\def\rPos{\R_{> 0}}%
\global\long\def\dWeight{\R_{>0}^{m}}%
\global\long\def\dWeights{\rPos^{m}}%

\global\long\def\specGeq{\succeq}%
 
\global\long\def\specLeq{\preceq}%
 
\global\long\def\specGt{\succ}%
 
\global\long\def\specLt{\prec}%
\global\long\def\gradient{\nabla}%
\global\long\def\weight{w}%
 
\global\long\def\vWeight{\vvar{\weight}}%
 
\global\long\def\mWeight{\mvar W}%
\global\long\def\mNormProjLap{\bar{\mLambda}}%
\global\long\def\volPot{\mathcal{V}}%
\global\long\def\dInterior{\Omega^{\circ}}%
\global\long\def\dFull{\{\dInterior\times\R_{>0}^{m}\}}%
\global\long\def\vWeight{\vvar w}%
 
\global\long\def\mWeight{\mvar W}%
\global\long\def\polytope{\Omega}%
\global\long\def\interior{\Omega^{\circ}}%

\begin{abstract}
We present an algorithm that given a linear program with $n$ variables,
$m$ constraints, and constraint matrix $\ma$, computes an $\epsilon$-approximate
solution in $\otilde(\sqrt{\rank(\ma)}\log(1/\epsilon))$ iterations
with high probability. Each iteration of our method consists of solving
$\otilde(1)$ linear systems and additional nearly linear time computation,
improving by a factor of $\tilde{\Omega}((m/\rank(\ma))^{1/2})$ over
the previous fastest method with this iteration cost due to Renegar
(1988) \cite{renegar1988polynomial}.\footnote{This paper is a journal version of the paper, ``Path-Finding Methods
for Linear Programming : Solving Linear Programs in $\otilde(\sqrt{\rank})$
Iterations and Faster Algorithms for Maximum Flow'' \cite{leeS14}
and arXiv submissions \cite{lsInteriorPoint,lsMaxflow}.

This paper contains several new results beyond these prior submissions.
This paper provides the first proof of a $\otilde(r)$-self-concordant
barrier for all polytopes $\{x\in\R^{n}:\ma x\geq b\}$ with $r=\rank(\ma)$
that is polynomial time computable (as opposed to the pseudo-polynomial
time computability of the universal barrier of \cite{nesterov1989self}).
Further, this paper provides new connections between the algorithms
presented, the barrier analyzed, and $\ell_{p}$ Lewis weights \cite{Lewis1978,CohenP15,CohenLS18}.

Several components of \cite{leeS14,lsInteriorPoint,lsMaxflow} were
not included in this journal version. Techniques, for leveraging this
paper to solve linear programs exactly are deferred to \cite{lsInteriorPoint}
and techniques for analyzing the error induced by approximate linear
system solves are deferred to \cite{lsMaxflow}. These techniques
are fairly standard and general and omitted from this paper for brevity.
Further, techniques for reducing the cost of the linear systems found
in \cite{lsInteriorPoint} are also not included and have been improved
in a sequence of recent work \cite{LeeS15,CohenLS18,AdilKPS19} and
techniques solving generalized minimum cost flow as opposed to the
more restricted minimum cost flow problem considered in this paper
are deferred to \cite{lsMaxflow}.} Further, we provide a deterministic polynomial time computable $\otilde(\rank(\ma))$-self-concordant
barrier function for the polytope, resolving an open question of Nesterov
and Nemirovski (1994) \cite{nesterov1989self} on the theory of ``universal
barriers'' for interior point methods.

Applying our techniques to the linear program formulation of maximum
flow yields an $\otilde(|E|\sqrt{|V|}\log(U))$ time algorithm for
solving the maximum flow problem on directed graphs with $|E|$ edges,
$|V|$ vertices, and integer capacities of size at most $U$. This
improves upon the previous fastest polynomial running time of $O(|E|\min\{|E|^{1/2},|V|^{2/3}\}\log(|V|^{2}/|E|)\log(U))$
achieved by Goldberg and Rao (1998) \cite{GoldbergRao}. In the special
case of solving dense directed unit capacity graphs our algorithm
improves upon the previous fastest $O(|E|\min\{|E|^{1/2},|V|^{2/3}\})$
running times achieved by Even and Tarjan (1975) \cite{even1975network}
and Karzanov (1973) \cite{k1973} and of $\otilde(|E|^{10/7})$ achieved
more recently by M\k{a}dry (2013) \cite{madryFlow}.
\end{abstract}
\pagebreak\pagenumbering{arabic}
\setcounter{page}{1}
\section{Introduction}

Given a matrix, $\ma\in\R^{m\times n}$, and vectors, $\vb\in\Rm$
and $\vc\in\Rn$, solving a linear program
\begin{equation}
\min_{\vx\in\Rn~:~\ma\vx\geq\vb}\vc^{\top}\vx\label{eq:intro:problem}
\end{equation}
is a core algorithmic task for the theory and practice of computer
science and operations research.

Since Karmarkar's breakthrough result in 1984 \cite{karmarkar1984new},
proving that interior point methods can solve linear programs in polynomial
time for a relatively small polynomial, interior point methods have
been an incredibly active area of research. Currently, the fastest
asymptotic running times for solving \eqref{eq:intro:problem} in
many regimes are interior point methods. Previously, state-of-the-art
interior point methods for solving \eqref{eq:intro:problem} compute
an $\epsilon$-approximate solution in either $\otilde(\sqrt{m}\log(1/\epsilon))$\footnote{Here and throughout the paper we use $\otilde(\cdot)$ to hide factors
polylogarithmic in $m$, $n$, $U$, $|V|$, $|E|$, and $M$.} iterations of solving linear systems \cite{renegar1988polynomial}
or $\otilde((m\cdot\rank(\ma))^{1/4}\log(1/\epsilon))$ iterations
of a more complicated but still polynomial time operation \cite{vaidya89convexSet,vaidya90parallel,vaidya1993technique,anstreicher96}.\footnote{All approximate linear programming algorithms discussed in this paper
can be leveraged to obtain exact solutions in weakly polynomial time
through standard straightforward reductions (see e.g. \cite{renegar1988polynomial}).
This transformation replaces each $\log(1/\epsilon)$ factor in running
times with $L$, a parameter that is at most the number of bits needed
to represent \eqref{eq:intro:problem} but in many cases can be much
smaller.}

However, in a breakthrough result of Nesterov and Nemirovski in 1994,
they showed that there exists a \emph{universal barrier} function
that if computable would allow \eqref{eq:intro:problem} to be solved
in $O(\sqrt{\rank(\ma)}\log(1/\epsilon))$ iterations \cite{nesterov1997self}.
Unfortunately, this barrier is more difficult to compute than solutions
to \eqref{eq:intro:problem} and despite this result, in many regimes
the fastest interior point algorithms are still based on the $\otilde(\sqrt{m}\log(1/\epsilon))$
iteration algorithm of Renegar from 1988.

In this paper we present a new interior point method that solves general
linear programs in $\otilde(\sqrt{\rank(\ma)}\log(1/\epsilon))$ iterations
thereby matching the theoretical limit proved by Nesterov and Nemirovski
up to polylogarithmic factors. Further, we show how to achieve this
convergence rate while only solving $\otilde(1)$ linear systems and
performing additional $\otilde(\nnz(\ma))$ work in each iteration.\footnote{We assume that $\ma$ has no rows or columns that are all zero as
these can be remedied by trivially removing constraints or variables
respectively or immediately solving the linear program. Therefore
$\nnz(\ma)\geq\min\{m,n\}$.} Our algorithm is easily parallelizable and in the standard PRAM model
of computation we achieve the first $\otilde(\sqrt{\rank(\ma)}\log(1/\epsilon))$-depth
polynomial-work method for solving linear programs. Using state-of-the-art
regression algorithms in \cite{nelson2012osnap,li2012iterative},
our linear programming algorithm has a running time of $\otilde((\mathrm{nnz}(\ma)+\left(\rank(\ma)\right)^{\omega})\sqrt{\rank(\ma)}\log(1/\epsilon))$
where $\omega<2.3729$ is the matrix multiplication constant \cite{williams2012matrixmult}.
Further, leveraging advances in solving sequences of linear systems
this running time is improvable to $\otilde((\mathrm{nnz}(\ma)+\rank(\ma)^{2})\sqrt{\rank(\ma)}\log(1/\epsilon))$\cite{LeeS15}. 

We achieve our results through an extension of standard path following
techniques for linear programming \cite{renegar1988polynomial,gonzaga1992path}
that we call \emph{weighted path finding}. We study the \emph{weighted
central path}, i.e. a weighted variant of the standard logarithmic
barrier function \cite{todd1994scaling,freund_weighted,megiddo_weighted}
that was used implicitly by M\k{a}dry \cite{madryFlow} to achieve
a breakthrough improvement to the running time for solving unit-capacity
maximum flow problem \cite{madryFlow}. We provide a general analysis
of the weighted central path, discuss tools for manipulating points
along the path and changing the path, and leverage this to produce
an efficiently computable path that converges in $\otilde(\sqrt{\rank(\ma)}\log(1/\epsilon)$
iterations.

Ultimately, we show approximately following the central path re-weighted
by variants of\emph{ $\ell_{p}$ Lewis weights}, a fundamental concept
in Banach space theory that has recently found applications for solving
$\ell_{p}$ regression, yields our desired running times. We provide
further intuition regarding these weighted central paths, and show
that the central path re-weighted by $\ell_{p}$-Lewis weights is
the central path induced by a $\otilde(\rank(\ma))$-self-concordant
barrier. Further, we show that the value, gradient, and Hessian of
this barrier are all computable deterministically in polynomial time.
This \emph{Lewis weight barrier} constitutes the first barrier for
polytopes whose self-concordance nearly matches that of the universal
\cite{Nesterov1994,lee2018universal} and entropic \cite{BubeckE15}
barriers; neither of which are not known to be either deterministically
or polynomial time computable. Previous methods for computing such
barriers required random sampling and run in pseudo-polynomial time,
i.e. have running times which depend polynomially (as opposed to polylogarithmically)
on the desired accuracy \cite{abernethy2016faster}.

To further demonstrate the efficacy of our proposed interior point
method, we show that it yields provably faster algorithms for solving
the maximum flow problem, one of the most well studied problems in
combinatorial optimization \cite{schrijver2003combinatorial}. By
applying our interior point method to a linear program formulation
of maximum flow and applying state-of-the-art solvers for symmetric
diagonally dominant linear systems \cite{spielman2004nearly,KoutisMP10,KMP11,Kelner2013,lee2013ACDM,cohen2014solving,lee2015sparsified,kyng2016approximate},
to implement the iterations we achieve an algorithm on $|V|$ node,
$|E|$ edge graphs with integer capacities in the range $0$ to $U$
in time $O(|E|\sqrt{|V|}\log^{O(1)}(|V|)\log(U))$ with high probability.
This improves upon the previous fastest polynomial running time of
$O(|E|\min\{|E|^{1/2},|V|^{2/3}\}\log(|V|^{2}/|E|)\log(U))$ achieved
in 1998 by Goldberg and Rao \cite{GoldbergRao} for dense graphs.
In the special case of solving dense unit capacity graphs our algorithm
improves upon the previous fastest running times of $O(|E|\min\{|E|^{1/2},|V|^{2/3}\})$
achieved by Even and Tarjan in 1975 \cite{even1975network} and Karzanov
in 1973 \cite{k1973} and of $\otilde(|E|^{10/7})$ achieved by M\k{a}dry
\cite{madryFlow} more recently. Further, our algorithm is easily
parallelizable and using \cite{peng2013efficient,lee2015sparsified,KyngLPSS16},
in the PRAM model we obtain a $\otilde(|E|\sqrt{|V|}\log(U))$-work
$\otilde(\sqrt{|V|})$-depth algorithm. Using the same technique,
we also solve the minimum cost flow problem in time $\otilde(|E|\sqrt{|V|}\log(M))$
with high probability where $M$ is an upper bound on the absolute
value of integer costs and capacities, improving upon the previous
fastest algorithm of $\otilde(E|^{1.5}\log(M))$ due to Daitch and
Spielman \cite{daitch2008faster}.

\subsection{Previous Work}

Linear programming is an extremely well studied problem with a long
history. There are numerous algorithmic frameworks for solving linear
programming problems, e.g. simplex methods \cite{dantzig1951maximization},
ellipsoid methods \cite{khachiyan1980polynomial}, and interior point
methods \cite{karmarkar1984new}. Each method has a rich history and
an impressive body of work analyzing the practical and theoretical
guarantees of the methods. Here we only present the major improvements
on the number of iterations required to solve \eqref{eq:intro:problem}
and discuss the asymptotic running times of these methods. For a more
comprehensive history linear programming and interior point methods
we refer the reader to one of the many excellent references on the
subject, e.g. \cite{Nesterov1994,ye2011interior}.

In 1984 Karmarkar \cite{karmarkar1984new} provided the first proof
of an interior point method running in polynomial time. This method
required $O(m\log(1/\epsilon))$ iterations where the running time
of each iteration was dominated by the time needed to solve a linear
system of the form $\ma^{\trans}\md\ma\vx=\mbox{\ensuremath{\vy}}$
for some diagonal matrix $\md\in\R_{>0}^{m\times m}$ and some $\vy\in\Rn$.
Using low rank matrix updates and preconditioning, Karmarkar achieved
a running time of $O(m^{3.5}\log(1/\epsilon))$ for solving \eqref{eq:intro:problem}
inspiring a long line of research into interior point methods.

In 1988 Renegar provided an improved $O(\sqrt{m}\log(1/\epsilon))$
iteration interior point method for solving \eqref{eq:intro:problem}.
His method was based on type of interior point methods known as \emph{path
following methods }which solve \eqref{eq:intro:problem} by incrementally
minimizing a $f_{t}(\vx)\defeq t\cdot\vc^{T}\vx+\phi(\vx)$ where
$\phi:\Rn\rightarrow\R$ is a \emph{barrier function }such that $\phi(\vx)\rightarrow\infty$
as $\vx$ tends to boundary of the polytope and $t$ is a parameter
changed during the algorithm. Renegar provided a method based on using
the \emph{log barrier} $\phi_{\ell}(\vx)\defeq-\sum_{i\in[m]}\log([\ma\vx-\vb]_{i})$
which serves as the foundation for many modern interior point methods.
As with Karmarkar's result the running time of each iteration of this
method was dominated by the time needed to solve a linear system of
the form $\ma^{\trans}\md\ma\vx=\vy$. Using a combination of techniques
involving low rank updates, preconditioning and fast matrix multiplication,
the amortized complexity of each iteration was improved \cite{vaidya1987speeding,gonzaga1992path,Nesterov1994}
yielding the previous best known running time of $O(m^{1.5}n\log(1/\epsilon))$
\cite{vaidya1989speeding}.

In seminal work of Nesterov and Nemirovski in 1994 \cite{Nesterov1994},
they generalized this approach and showed that path-following methods
can be applied to minimize any linear cost function over any convex
set if given a suitable \emph{barrier function}. They introduced a
measure of complexity of a barrier known as \emph{self-concordance
}and showed that given any \emph{$\nu$-self-concordant barrier} for
the set, an $\otilde(\sqrt{\nu}\log(1/\epsilon))$ iteration method
could be achieved. Further, they showed that for any convex set in
$\Rn$, there exists an $O(n)$-self-concordant barrier,called the
\emph{universal barrier} function. Therefore, in theory any such $n$-dimensional
convex optimization problem can be solved in $O(\sqrt{n}\log(1/\epsilon))$
iterations. However, this result is traditionally considered to be
primarily of theoretical interest as the universal barrier function
is difficult to compute. Given the possible algorithmic implications
of faster interior point methods, e.g. the flow problems of this paper,
obtaining a barrier with near-optimal self-concordance that is easy
to minimize is a fundamental open problem.

In 1989, Vaidya \cite{vaidya1993technique} made an important breakthrough
in this direction. He proposed two barrier functions related to the
volume of certain ellipsoids and obtained $O((m\cdot\rank(\ma))^{1/4}\log(1/\epsilon))$
and $O(\rank(\ma)\log(1/\epsilon))$ iteration linear programming
algorithms \cite{vaidya90parallel,vaidya1993technique,vaidya89convexSet}.
Unfortunately, each iteration of these methods required computing
the projection matrix $\md^{1/2}\ma(\ma^{\trans}\md\ma)^{-1}\ma^{\trans}\md^{1/2}$
for a positive diagonal matrix $\md\in\Rmm$. This was slightly improved
by Anstreicher \cite{anstreicher96} who showed it sufficed to compute
the diagonal of this projection matrix. Unfortunately, neither of
these methods yield faster running times than \cite{vaidya1989speeding}
unless $m\gg n$ and neither are immediately amenable to take full
advantage of improvements in solving structured linear system solvers
and thereby improve the running time for solving the maximum flow
problem.\vspace{5pt}\\
\begin{tabular}{|c|l|c|c|}
\hline 
Year  & Author & Number of Iterations & Nature of iterations\tabularnewline
\hline 
\hline 
1984  & Karmarkar \cite{karmarkar1984new}  & $\otilde(m\log(1/\epsilon))$  & Linear system solve\tabularnewline
\hline 
1986  & Renegar \cite{renegar1988polynomial}  & $O(\sqrt{m}\log(1/\epsilon))$  & Linear system solve\tabularnewline
\hline 
1989  & Vaidya \cite{vaidya1996new}  & $O((m\cdot\rank(\ma))^{1/4}\log(1/\epsilon))$  & Matrix Inversion\tabularnewline
\hline 
1994  & Nesterov and Nemirovskii \cite{Nesterov1994} & $O(\sqrt{\rank(\ma)}\log(1/\epsilon))$  & Volume computation\tabularnewline
\hline 
 & This paper  & $\otilde(\sqrt{\rank(\ma)}\log(1/\epsilon))$  & $\otilde(1)$ Linear system solves\tabularnewline
\hline 
\end{tabular}\vspace{5pt}

These results suggest that you can solve linear programs closer to
the $\otilde(\sqrt{\rank(\ma)}\log(1/\epsilon))$ bound achieved by
the universal barrier only if you pay more in each iteration. In this
paper, we show that this is not the case. We provide a method that
up to polylogarithmic factors matches the convergence rate of the
universal barrier function while only having iterations of cost comparable
to that of Karmarkar's \cite{karmarkar1984new} and Renegar's \cite{renegar1988polynomial}
algorithms. 

\subsection{Our Results}

Our main result is provably faster algorithms which given $\ma\in\Rmn$,
$\vb\in\Rn$, $\vc\in\Rm$, $l_{i}\in\R\cup\{-\infty\}$, and $u_{i}\in\R\cup\{+\infty\}$
for all $i\in[m]$ solve linear programs in the following form\footnote{Typically \eqref{eq:intro:two-sided} is written as $\ma\vx=\vb$
rather than $\ma^{\top}\vx=\vb$. We chose this formulation to be
consistent with the derivation of the self-concordant barrier in Section~\ref{sec:self-concordance},
and the standard use of $n$ to denote the number of vertices and
$m$ to denote the number of edges in the linear program formulation
of flow problems.}
\begin{equation}
\text{OPT}\defeq\min_{\begin{array}{c}
\vx\in\Rm~:~\ma^{\trans}\vx=\vb\\
\forall i\in[m]~:~l_{i}\leq x_{i}\leq u_{i}
\end{array}}\vc^{\top}\vx\,.\label{eq:intro:two-sided}
\end{equation}
 We assume throughout that $\ma$ is \emph{non-degenerate}, which
we define as full column rank and no rows that are all zero. Further,
we assume that for all $i\in[m]$ the set $\dom(x_{i})\defeq\{x\,:\,l_{i}<x<u_{i}\}$,
is neither the empty set or the entire real line, i.e. $l_{i}<u_{i}$
and either $l_{i}\neq-\infty$ or $u_{i}\neq+\infty$ and we assume
that the interior of the polytope, $\dInterior\defeq\{\vx\in\Rm~:~\ma^{\top}\vx=\vb,l_{i}<x_{i}<u_{i}\}$,
is non-empty. 

The problem of solving \eqref{eq:intro:two-sided} without these assumptions
is reducible to an instance where these assumptions hold, without
increasing the running times the methods of this paper by more than
polylogarithmic factors (see e.g., Appendix~E of Part~I \cite{lsInteriorPoint}).
Our main result is the following.

\begin{restatable}[Linear Programming]{thm}{LPsolvedetail}\label{thm:LPSolve_detailed}
Given interior point $\vx_{0}\in\dInterior$ for linear program \eqref{eq:intro:two-sided},
the algorithm $\code{LPSolve}$ (Algorithm~\ref{alg:lp_solve}) outputs
$\vx\in\dInterior$ with $\vc^{\top}\vx\leq\text{OPT}+\epsilon$ with
constant probability in
\[
O(\sqrt{n}\log^{13}m\cdot\log(mU/\epsilon)\cdot\mathcal{T}_{w})\text{-work and }O(\sqrt{n}\log^{13}m\cdot\log(mU/\epsilon)\cdot\mathcal{T}_{d})\text{-depth}
\]
where $U=\max\{\norm{1/(u-x_{0})}_{\infty},\norm{1/(x_{0}-l)}_{\infty},\norm{\vu-\vl}_{\infty},\norm{\vc}_{\infty}\}$
and $\mathcal{T}_{w}$ and $\mathcal{T}_{d}$ are the work and depth
needed to compute $(\ma^{\top}\md\ma)^{-1}\vq$ for input positive
diagonal matrix $\md$ and vector $\vq$.

\end{restatable}

Note that \eqref{eq:intro:two-sided} is the dual of \eqref{eq:intro:problem}
in the special case when $u_{i}=\infty$ for all $i\in[m]$. Consequently,
in obtaining this result we solve \eqref{eq:intro:problem} with the
desired complexity (see Theorem~\ref{thm:LPSolve_detailed_dual}).
We consider this formulation with two-sided constraints, \eqref{eq:intro:two-sided},
as it directly encompasses the formulation of maximum flow and minimum
cost flow as a linear program \cite{daitch2008faster}. Interestingly,
while it is well known that all linear programs, including \eqref{eq:intro:two-sided},
can be written in standard form, all known transformations to put
\eqref{eq:intro:two-sided} in standard form would increase the rank
of $\ma$ causing an $\otilde(\sqrt{\rank(\ma)}$) iteration algorithms
to be too slow to improve the running time for solving the maximum
flow problem.

Using lower bounds results of Nesterov and Nemirovski, it is not hard
to see that any general barrier for \eqref{eq:intro:two-sided} must
have self-concordance $\Omega(m)$. In particular, Proposition 2.3.6
of \cite{Nesterov1994} shows that if any vertex of a $m$-dimensional
polytope belongs to $k$ linearly independent $(m-1)$-dimensional
facets, then the self-concordance of any barrier on $\Omega$ is at
least $k$. Consequently, Theorem~\ref{thm:LPSolve_detailed} corresponds
to a method which converges at a rate faster than what would be predicted
by standard interior point theory. That we solve \eqref{eq:intro:two-sided}
in $o(\sqrt{m}\log(1/\epsilon))$ is critical for achieving our faster
maximum flow results. Leveraging Theorem~\ref{thm:LPSolve_detailed}
we show the following:

\begin{restatable}[Maximum Flow]{thm}{maxflow}\label{thm:maxflow}
Given a directed graph $G=(V,E)$ with integral costs $q\in\Z^{E}$
and capacities $c\in\Z_{\geq0}^{E}$ with $\|q\|_{\infty}\leq M$
and $\|c\|_{\infty}\leq M$, we can compute a minimum cost maximum
flow with constant probability with $O(|E|\sqrt{|V|}\log^{18}|E|\log M)$
work and $O(\sqrt{|V|}\log^{20}|E|\log M)$ depth. \end{restatable}

We complement these results by designing a new barrier whose self-concordance
nearly matches that of the the universal barrier. We show that specialized
to \eqref{eq:intro:problem} an idealized version of our algorithm
corresponds to following a path following scheme on a natural barrier
induced by Lewis weights \cite{Lewis1978}. Formally, when $\ma$
is non-degenerate we provide an $O(\rank(\ma)\log^{5}(m))$-self-concordant-barrier
such that it's gradient and Hessian are all polynomial time computable.

\begin{restatable}[Nearly Universal Barrier]{thm}{lewisbarriercomp}\label{thm:lewisbarriercomp}
Let $\interior\defeq\{x\,:\,\ma x>b\}$ be non-empty for non-degenerate
$\ma\in\R^{m\times n}$. There is an $O(n\log^{5}m)$-self concordant
barrier $\psi$ for $\interior$ such that for all $\epsilon>0$ and
$x\in\interior$ in $O(mn^{\omega-1}\cdot\log m\cdot\log(m/\epsilon))$-work
and $O(\log^{2}m\cdot\log(m/\epsilon))$-depth it is possible to compute
$g\in\R^{n}$ and $\mh\in\R^{n\times n}$ with $\|g-\nabla\psi(x)\|_{\nabla^{2}\psi(x)^{-1}}\leq\epsilon$,
and $(1-\epsilon)\nabla^{2}\psi(x)\preceq\mh\preceq(1+\epsilon)\nabla^{2}\psi(x)$. 

\end{restatable}

To obtain these results we provide several additional tools of possible
independent interest. In Section~\ref{sec:lewis_weights} we provide
several algebraic facts regarding Lewis weights and in Section~\ref{sec:weights_full:computing}
we provide several algorithms for computing Lewis weights in different
contexts. Further, in Section~\ref{sec:approx_path:chasing_zero}
we provide results for a natural online learning problem which we
leverage to handle approximation errors in our path finding schemes. 

Ultimately, we hope the varied results of this paper will open the
door towards developing even faster algorithms for convex programming
more broadly. While the analysis in the paper is quite technical,
ultimately the algorithms and heuristics they suggest, i.e. locally
re-weighting the central path by Lewis weights (and in the case of
maximum flow, effective resistance), are straightforward and we hope
may be used more broadly.

\subsection{Geometric Motivation\label{sec:geom_motiv}}

To motivate our approach, consider the slightly simplified problem
of designing an $\otilde(\sqrt{n}\log(1/\epsilon))=\otilde(\sqrt{\rank(\ma)}\log(1/\epsilon))$
iteration algorithm for solving \eqref{eq:intro:problem} for non-degenerate
$\ma$ where the running time of each iteration is dominated by the
time needed to solve a linear system $\ma^{\trans}\md\ma\vx=\vy$
for diagonal $\md\in\R_{\geq0}^{m\times m}$. The classic \emph{self-concordance
}theory for analyzing interior point methods established in \cite{Nesterov1994}
shows that it suffices to produce a simple enough $\otilde(n)$-self-concordant
barrier for the set $\interior\defeq\{x\in\R^{n}|\ma x>b\}$. This
seminal work of Nesterov and Nemirovski showed that given any $\nu$-\emph{self-concordant
barrier for an open convex} set $K$ there is an $\otilde(\sqrt{\nu}\log(1/\epsilon))$
iteration interior point method, based on a technique known as \emph{path
following}, for minimizing linear functions over $K$. Further, the
running time of each iteration is dominated by the time needed to
compute a gradient of the barrier and approximately solve a linear
system in its Hessian. 
\begin{defn}[Self-concordance]
\label{def:self_concordance} A convex, thrice continuously differentiable
function $\phi:K\rightarrow\R^{n}$ is a \emph{$\nu$-self-concordant
barrier function} for open convex set $K\subset\R^{n}$ if the following
conditions hold
\begin{itemize}
\item $\lim_{i\rightarrow\infty}\phi(x_{i})\rightarrow\infty$ for all sequences
$x_{i}\in K$ converging to boundary of $K$.
\item $|D^{3}\phi(x)[\vh,\vh,\vh]|\leq2|D^{2}\phi(x)[\vh,\vh]|^{3/2}$ for
all $x\in K$ and $\vh\in\Rn$,
\item $|D\phi(x)[\vh]|\leq\sqrt{\nu}|D^{2}\phi(x)[\vh,\vh]|^{1/2}$ for
all $x\in K$ and $\vh\in\Rn$.
\end{itemize}
\end{defn}

To achieve our goals, ideally we would produce a $\otilde(n)$-self-concordant
barrier function for the feasible region such that the resulting path
following scheme would have sufficiently low iteration costs. Unfortunately,
as we have discussed no such barrier is known to exist, all previous
$\otilde(n)$-self-concordant barriers are more difficult to evaluate
then linear programming, and it it would be unclear how to generalize
such an approach to solving \eqref{eq:intro:two-sided}. Deferring
this last issue to Section~\ref{subsec:Path-Finding}, here describe
how to overcome the first two issues and derive a deterministic polynomial-time
computable barrier functions with self-concordance $\tilde{O}(n)$. 

Our barrier function can be derived from the following intuition regarding
interior point methods. At a high level, interior point methods address
the key difficulty of linear programming, making progress in the presence
of non-differentiable inequality constraints, by leveraging a \emph{barrier
function}, $\phi$, which provides a local smooth approximation. These
methods solve the linear program by performing Newtons method, i.e.
solving a sequence of linear systems, which trade off the utility
of minimizing cost, $c^{\top}x$, and staying away from the constraints,
i.e. minimizing $\phi$. Since these Newton steps correspond to minimizing
linear functions over ellipsoids and these ellipsoids come from the
second-order approximations of the barrier functions, interior point
methods essentially approximate polytopes by a sequence of ellipsoids.
Self-concordance can be viewed as a geometric condition that relates
how well these ellipsoids approximate the domain. In particular, the
following lemma shows that the second-order approximation of the barrier
function at the minimum point well-approximates the domain.
\begin{thm}
[Dikin Ellipsoid Rounding {\cite[Thm 4.2.6]{Nesterov2003}}]\label{thm:rounding_self_concordant}
Given a $\nu$-self-concordant barrier function $\phi$ for convex
set $K\subset\Rn$, let $x_{\phi}$ be the minimizer of $\phi$ and
$E\defeq\{x\in\Rn:\ (x-x_{\phi})^{\top}\nabla^{2}\phi(x_{\phi})(x-x_{\phi})\leq1\}$
be the \emph{Dikin ellipsoid}. $E$ is a \emph{$\nu+2\sqrt{\nu}$-rounding
of $K$, i.e.} \emph{$E\subseteq K\subseteq(\nu+2\sqrt{\nu})E.$}
\end{thm}

Consequently, to obtain a $\tilde{O}(n)$-self-concordant barriers
it is necessary to obtain ellipsoids that are $\tilde{O}(n)$-roundings.
The maximum volume contained ellipsoid or \emph{John ellipsoid} has
this property.
\begin{lem}[John Ellipsoid Rounding \cite{john1948extremum}]
For convex $K\subseteq\Rn$ and John ellipsoid, $J(K)$, i.e. the
largest volume ellipsoid contained inside $K$, we have that $J(K)\subseteq K\subseteq nJ(K).$ 
\end{lem}

In contrast to other ellipsoids that yield approximation guarantees,
e.g. the covariance matrix of the uniform distribution on the body
\cite{vempala2010recent}, the John ellipsoid has the desirable property
of being defined by a convex optimization problem and therefore can
be computed in weakly polynomial time. There are multiple ways to
express the John ellipsoid as the solution to a convex problem. Our
barrier function is motivated by the following formulation, called
\emph{$D$-optimal design.}
\begin{lem}[Convex Formulation of John Ellipsoid \cite{khachiyan1996rounding}]
 For any $\ma\in\R^{m\times n}$, $b\in\R^{m}$, and polytope interior
$\interior=\{x\in\R^{n}:\ma\vx>\vb\}$ the John ellipsoid equals $\{\vy\in\Rn:(\vy-\vx)^{\top}\ma^{\top}\mw\ma(\vy-\vx)\leq1\}$
where $\{x\in\interior,w\in\R_{\geq0}^{m}\}$ is the saddle point
of the following convex concave problem
\begin{equation}
\min_{\vx\in\interior}\phi_{\infty}(\vx)\qquad\text{where}\qquad\phi_{\infty}(\vx)=\max_{\sum w_{i}=n,w_{i}\geq0}\ln\det\left(\ma^{\trans}\ms_{x}^{-1}\mw\ms_{x}^{-1}\ma\right)\label{eq:John_ellipsoid_formula}
\end{equation}
where $\ms_{x}$ and $\mw$ are diagonal $m\times m$ matrices with
$[\ms_{x}]_{ii}\defeq a_{i}^{\top}\vx-b_{i}$ and $\mw_{ii}\defeq w_{i}$.
\end{lem}

Motivated by Theorem~\ref{thm:rounding_self_concordant} a natural
approach towards obtaining a polynomial time computable $\otilde(n)$-self-concordant
barrier would simply be to pick a barrier function for $\interior$
whose minimizer is the center of John ellipsoid. The function $\phi_{\infty}(\vx)$
of \eqref{eq:John_ellipsoid_formula} is such a function, but unfortunately,
simply inducing a Dikin ellipse that approximates the feasible region
is insufficient to be a self-concordant barrier. A self-concordant
barrier also needs to not change two quickly; however $\phi_{\infty}(\vx)$
is not even continuously differentiable. To see this, let $J(\Omega,x)$
be the maximum volume ellipsoid inside $K$ and centered at $x$ and
note that $\phi_{\infty}(\vx)=c\log(\text{vol}(J(\Omega,x)))$ for
a universal constant $c$. Consequently, for $\Omega=[-1,1]$ we have
$\phi_{\infty}(x)=c\log(2(1-|x|))$, i.e. it is only affected by one
constraint at each point, except at $0$, where it is non-differentiable. 

To make $\phi_{\infty}(\vx)$ smooth, we could apply a standard approach
of adding a strongly concave term, i.e. a regularizer, to the objective
function $\ln\det(\ma^{\top}\ms_{x}^{-1}\mw\ms_{x}^{-1}\ma)$. In
general, if smooth $f(x,y)$ is strongly concave in $y$, then $\max_{y}f(x,y)$
is smooth in $x$. In fact, there are multiple ways to apply this
approach to obtain a polynomial time computable universal barrier
function. For example, it can be shown that that the following is
an $\tilde{O}(n)$ self-concordant barrier function
\[
\phi_{r}(x)\defeq\max_{\sum_{i\in[m]}w_{i}=n,w_{i}\geq0}\ln\det(\ma^{\top}\ms_{x}^{-1}\mw\ms_{x}^{-1}\ma)-\frac{n}{m}\sum_{i\in[m]}w_{i}\ln w_{i}-\frac{n}{m}\sum_{i\in[m]}\ln[\ms_{x}]_{ii}\,.
\]
\textbf{Lewis Weight Barrier}: In this paper, we provide a more elegant
barrier that we believe further elucidates the geometric structure
of the problem. In Section~\ref{sec:self-concordance} for all $p>0$
we consider the function
\[
\phi_{p}(x)\defeq\begin{cases}
\max_{w\in\R^{m}:w\geq0}\frac{1}{2}f_{p}(x,w) & \text{ if }p\geq2\\
\min_{w\in\R^{m}:w\geq0}\frac{1}{2}f_{p}(x,w) & \text{ if }p\leq2
\end{cases}
\]
where 
\[
f_{p}(x,w)\defeq\ln\det\left(\ma^{\top}\ms_{x}^{-1}\mw^{1-\frac{2}{p}}\ms_{x}^{-1}\ma\right)-\left(1-\frac{2}{p}\right)\sum_{i\in[m]}w_{i}\,.
\]
We show that the maximizing ($q>2$) or minimizing ($q<2$) weights,
$w\in\R_{\geq0}^{m}$, for $\phi_{p}$ are the \emph{$\ell_{p}$-Lewis
weights} for the matrix $\ms^{-1}\ma$ \cite{Lewis1978} and hence
we call $\phi_{p}$ the \emph{Lewis weight barrier}.

Lewis weights are fundamental in the theory of Banach spaces and a
key tool for approximating a matrix in $\ell_{p}$-norms. They generalize
a fundamental $\ell_{2}$ measure of row importance known as \emph{leverage
scores} which are defined for $\ma\in\R^{m\times n}$ as $\sigma(\ma)=\diag(\ma(\ma^{\top}\ma)^{-1}\ma^{\top}),$
i.e. the diagonals of the orthogonal projection matrix onto the image
of $\ma$. For all $p>0$ the $\ell_{p}-$Lewis weights of $\ma$
are the unique vector $\lpweight(\ma)$ which is the leverage scores
of $\mw_{p}^{(1/2)-(1/p)}\ma$ for $\mw_{p}=\mDiag(w_{p})$. Intuitively,
the \emph{$\ell_{p}$-Lewis weight of a row $i$}, $\lpweight(\ma)_{i}$,
denotes the importance of the $i^{th}$ row under $\ell_{p}$ norm
and it is known that sampling $\widetilde{O}(n^{\max\left\{ p/2,1\right\} })$
rows of $\ma$ with probability proportional to $\ell_{p}$ Lewis
weight and reweighting yields a matrix $\mb$ such that with high
probability $\norm{\mb x}_{p}\approx\norm{\ma x}_{p}$ multiplicatively
for all $x$ \cite{bourgain1989approximation}. Recently, Cohen and
Peng \cite{CohenP15} studied Lewis weights in the context of solving
$\ell_{p}$-regression, showed that Lewis weights computation can
be written as a convex optimization problem for $p\geq2$, and provided
a nearly constant iteration algorithm for computing Lewis weights
for $p\in(0,4)$.

In this paper we provide several complementary results regarding Lewis
weights, including formulating their computation as a convex optimization
problem for all $p>0$ (Section~\ref{sec:lewis_weights}) and providing
additional algorithms for computing them (Section~\ref{sec:weights_full:computing}).
Further, we study the stability of Lewis weights under re-scalings
and show that they induce ellipsoids that well approximate to the
polytope $\Omega=\{x\in\R^{n}\,|\,\norm{\ma x}_{\infty}\leq1\}$ for
large $p$ (Section~\ref{sec:lewis_weights}). Leveraging this analysis
we show that the Lewis weight barrier for $p=\Theta(\log m)$ is an
$O(n\log^{5}m)$-self-concordant barrier for $\interior$ (Section~\ref{sec:self-concordance})
and prove Theorem~\ref{thm:lewisbarriercomp}. The barrier $\phi_{\infty}$
is essentially the limit of $\phi_{p}$ for $p\rightarrow\infty$
and consequently our analysis shows that the $\ell_{\Theta(\log m)}$
generalization of the John ellipse yields a nearly universal barrier.

\subsection{Path Finding}

\label{subsec:Path-Finding}

Though the explanation of the previous section suffices to prove Theorem~\ref{thm:lewisbarriercomp},
it is unclear how to leverage this analysis to prove Theorem~\ref{eq:intro:problem}.
As discussed, there is no $O(n)$-self-concordant barrier for the
feasible region of \eqref{eq:intro:two-sided} and even if this issue
could be overcome, naively implementing such a method would require
the expensive operation of computing Lewis weights. However, computing
these weights to high precision (or even certifying their properties)
necessitates computing leverage scores which naively yields iteration
costs comparable to that of Vaidya and Anstreicher's interior point
methods \cite{vaidya90parallel,vaidya1993technique,vaidya89convexSet,anstreicher96},
i.e. slower then solving $\otilde(1)$-linear systems.

To overcome these issues we develop a scheme for dynamically re-weighting
self-concordant barriers for $\dom(x_{i})$ in \eqref{eq:intro:two-sided}.
We provide $1$-self-concordant barriers $\phi_{i}$ for each $\dom(x_{i})$
(see Section~\ref{sec:path:prelim}) and study the central path they
induce, i.e. $x_{t}$ for $t>0$, where 
\begin{equation}
x_{t}\defeq\argmin_{\ma^{\top}x=b}f_{t}(x)\enspace\text{where }\enspace f_{t}(\vx)\defeq t\cdot\vc^{\top}\vx+\sum_{i\in[m]}\phi_{i}(\vx_{i})\,.\label{eq:penalized_objective-1}
\end{equation}
Self-concordance theory yields that $\sum_{i\in[m]}\phi_{i}(\vx_{i})$
is a $m$-self-concordant barrier and therefore this yields an $\otilde(\sqrt{m})$
iteration method; we directly attempt to improve this bound.

To motivate our improvement, note that the performance of this method
is highly dependent on the representation of \eqref{eq:intro:problem}.
Duplicating a constraint, i.e. a row of $\ma$ and the corresponding
entry in $\vb_{i}$, $\ell_{i}$ and $u_{i}$, corresponds to doubling
the contribution of some $\phi_{i}$. Repeating a constraint many
times can actually slow down the convergence of standard path following
methods and in a series of papers \cite{deza2006central,deza2008good,nematollahi2008redundant,nematollahi2008simpler,mut2013tight},
it was shown that by carefully duplicating constraints on Klee-Minty
cubes standard interior point methods for the dual can take $\Omega(\sqrt{m})$
iterations.

Since the weighting of $\phi_{i}$ can affect convergence, we provide
algorithms which dynamically re-weight the $\phi_{i}$. We show that
this can improve the convergence rate from $\Omega(\sqrt{m})$ to
$\otilde(\sqrt{n})$. In Section~\ref{sec:weighted_path_finding},
we study the \emph{weighted barrier function,} $\phi(x)=\sum_{i\in[m]}g_{i}(x)\phi_{i}(x)$
where $\vg:\R_{\geq0}^{m}\rightarrow\R_{\geq0}^{m}$ is a \emph{weight
function} of the current point, and the \emph{weighted central path}
they induce, i.e.

\[
x_{t}^{g}\defeq\argmin_{\ma^{\top}x=b}f_{t}(x)\enspace\text{where }\enspace f_{t}(\vx)\defeq t\cdot\vc^{\top}\vx+\sum_{i\in[m]}g_{i}(x)\phi_{i}(\vx_{i})\text{ for all }t\geq0\,.
\]
To obtain our improved running times we investigate what properties
of $\vg(\vx)$ improve convergence. Standard analysis suggests that
$g$ should have small total size, i.e. $\norm{g(x)}_{1}=O(n)$, and
induce Newton steps that do not change the Hessian much. Optimizing
weights for these conditions suggest that $g(x)$ should be the $\ell_{1}$-Lewis
weights for the local re-weighting of the constraint matrix. In the
special case where all $\ell_{i}=0$ and $u_{i}=+\infty$ this recovers
the motivation for $\phi_{\infty}$! Here, we run into the same issues
discussed in Section~\ref{sec:geom_motiv}, e.g. instability of John
ellipse and $\ell_{1}$-Lewis weights. Consequently, we consider the
dual analog of the approach of Section~\ref{sec:geom_motiv} and
let $g$ be the $\ell_{p}$-Lewis weights for $p=1-1/\log(4m)$ plus
a fixed amount and show these regularized Lewis weights have the desired
properties. Interestingly, when $\ell_{i}=0$ and $u_{i}=+\infty$,
ignoring the constant regularization, the $x_{t}^{g}$ are dual to
the central path induced by the $\ell_{q}$-Lewis weight barrier.

This reasoning yields a dual algorithm related to path following methods
with the Lewis weight barrier: Newton step $x$ for fixed $g(x)$,
update $g(x)$, update $t$, and repeat. All that remains is the issue
of computing $\ell_{p}$-Lewis weights. To overcome this issue, we
exploit that leverage scores and consequently $\ell_{p}$-Lewis weights
for small $p$ can be efficiently approximated for small $p$, as
was shown in the aforementioned exciting result \cite{CohenP15}.
In Section~\ref{sec:weights_full:computing} we provide additional
Lewis weight computation algorithms for all $p$ which we leverage
to compute multiplicative approximations to $\ell_{p}$-Lewis weights
in our methods. Unfortunately, this error is still too much for our
methods to handle directly, as such large weights changes can greatly
decrease centrality measures.

To overcome this final issue, rather then using the weighted barrier
$\phi(x)=\sum_{i\in[m]}g_{i}(x)\phi_{i}(x)$ where the weights $\vg(\vx)$
depends on the $\vx$ directly, we instead maintain separate weights
$w\in\R_{>0}^{m}$ and current point $\vx$ and use the barrier $\phi(\vx,w)=\sum_{i\in[m]}w_{i}\phi_{i}(x_{i})$.
We then design a method where we maintain the invariants that $x$
is close to the minimum of $\phi(\vx,w)$ over $\ma^{\top}x=b$ and
$w$ is multiplicatively close to $\vg(\vx)$. Since, each fixed $w\in\R_{\geq0}^{m}$
induces a particular weighted central path, i.e. the minimizers of
$t\cdot c^{\top}x+\phi(x,w)$, our method can be viewed as alternating
between advancing along a weighted central path and changing the path.
We call this technique, \emph{path} \emph{finding.}

We design this path-finding method in two steps. First, we show that
we can take a Newton step on $x$ and update $w$ while improving
centrality and not changing $w$ too much. This requires care, as
with the weighted barrier it is difficult to certify that Newton steps
are stable, i.e. does not change points multiplicatively. To overcome
this, we explicitly measure the centrality of our points by the size
of the Newton step in a mixed norm of the form $\norm{\cdot}=\norm{\cdot}_{\infty}+\cnorm\norm{\cdot}_{\mw}$
to keep track of both the standard measure of centrality and this
multiplicative change. Second, we show that given a multiplicative
approximation to $\vg(\vx)$ and bounds on the change of $\vg(\vx)$,
we can maintain the invariant that $\vg(\vx)$ is close to $w$ multiplicatively
without moving $w$ too much. We formulate this as a general two player
game and provide an efficient solution in Section~\ref{sec:approx_path:chasing_zero}.

By combining these insights and formulating minimum cost flow as a
linear program, we prove Theorem~\ref{thm:LPSolve_detailed} and
Theorem~\ref{thm:maxflow}. Measuring Newton step sizes with respect
to the mixed norm helps explain how our method outperforms the self-concordance
of the best barrier for \eqref{eq:intro:two-sided}. Self-concordance
is based on $\ell_{2}$ analysis and lower bounds for self-concordance
stem from the failure of $\ell_{2}$ to approximate $\ell_{\infty}$.
While ideally our methods might optimize over $\ellInf$ directly,
$\ellInf$ is rife with degeneracies impairing this analysis. However,
unconstrained minimization over a box is simple and by working with
this mixed norm and carefully choosing weights we are taking advantage
of the simplicity of minimizing $\ellInf$ over most of the domain
and only paying for the $\otilde(n)$-self-concordance of a barrier
for the subspace induced by the $\ma^{\trans}\vx=\vb$ constraint.

\subsection{Paper Organization}

After providing notation in Section~\ref{sec:notation}, in Section~\ref{sec:weighted_path_finding}
we provide our analysis of weighted path finding, in Section~\ref{sec:lewis_weights}
we provide our analysis of Lewis weights, and in Section~\ref{sec:self-concordance}
we prove the self-concordance of the Lewis weight barrier. The proofs
of Theorems \ref{thm:LPSolve_detailed}, \ref{thm:maxflow}, and \ref{thm:lewisbarriercomp}
are then given in Section~\ref{sec:weights_full:computing}. Algorithms
for computing Lewis weights and many technical details are deferred
to the appendix. Note that throughout we made only limited attempts
to reduce polylogarithmic factors.

\section{Notation\label{sec:notation} }

\textbf{Vector Operations:} We frequently apply scalar operations
to vectors with the interpretation that these operations should be
applied coordinate-wise, e.g. for $\vx,\vy\in\Rn$ we let $\vx/\vy\in\Rn$
with $[\vx/\vy]_{i}\defeq(x_{i}/y_{i})$, $xy\in\R^{n}$ with $[xy]_{i}=x_{i}y_{i}$,
and $\log(\vx)\in\Rn$ with $[\log(\vx)]_{i}=\log(x_{i})$ for all
$i\in[n]$ .\textbf{\vspace{10pt}}\\
\textbf{Matrices}: We call a matrix $\ma$\emph{ non-degenerate} if
it has full column-rank and no zero rows. We call symmetric matrix
$\mb\in\Rnn$ positive semidefinite (PSD) if $\vx^{\top}\mb\vx\geq0$
for all $\vx\in\Rn$ and positive definite (PD) if $\vx^{\top}\mb\vx>0$
for all $\vx\in\Rn$.\textbf{\vspace{10pt}}\\
\textbf{Matrix Operations:} For symmetric matrices $\ma,\mb\in\Rnn$
we write $\ma\preceq\mb$ to indicate that $\vx^{\top}\ma\vx\leq\vx^{\top}\mb\vx$
for all $\vx\in\Rn$ and define $\prec$, $\preceq$, and $\succeq$
analogously. For $\ma,\mb\in\R^{n\times m}$, we let $\ma\shurProd\mb$
denote the Schur product, i.e. $[\ma\circ\mb]_{ij}\defeq\ma_{ij}\cdot\mb_{ij}$
for all $i\in[n]$ and $j\in[m]$, and we let $\shurSquared{\ma}\defeq\ma\circ\ma$.
We use $\nnz(\ma)$ to denote the number of nonzero entries in $\ma$.\textbf{\vspace{10pt}}\\
\textbf{Diagonals:} For $\ma\in\R^{n\times n}$ we define $\diag(\ma)\in\R^{n}$
with $\diag(\ma)_{i}=\ma_{ii}$ for all $i\in[n]$ and for $\vx\in\R^{n}$
we define $\mDiag(\vx)\in\R^{n\times n}$ as the diagonal matrix with
$\diag(\mDiag(\vx))=\vx$. We often use upper case to denote a vectors
associated diagonal matrix, e.g. $\mx\defeq\mDiag(x)$ and $\ms=\mDiag(s)$.
\textbf{\vspace{10pt}}\\
\textbf{Fundamental Matrices}: For non-degenerate $\ma$ we let $\mProj(\ma)\defeq\ma(\ma^{\top}\ma)^{-1}\ma^{\top}$
denote the orthogonal projection matrix onto $\ma$'s image and $\sigma(\ma)\defeq\diag(\mProj(\ma))$
denote $\ma$'s \emph{leverage scores}. We let $\mSigma(\ma)\defeq\mDiag(\sigma(\ma))$,
$\mProj^{(2)}(\ma)\defeq\mProj(\ma)\circ\mProj(\ma)$, $\mLambda(\ma)\defeq\mSigma(\ma)-\mProj^{(2)}(\ma)$,
and $\mNormProjLap(\ma)\defeq\mSigma(\ma)^{-1/2}\mLambda(\ma)\mSigma(\ma)^{-1/2}$.
$\mLambda(\ma)$ is a Laplacian matrix and $\mNormProjLap(\ma)$ is
a normalized Laplacian matrix.\textbf{\vspace{11pt}}\\
\textbf{Norms}: For PD $\ma\in\Rnn$ we let $\|\cdot\|_{\ma}$ denote
the norm where $\norm{\vx}_{\ma}^{2}\defeq\vx^{\top}\ma\vx$ for all
$\vx\in\Rn$. For positive $w\in\R_{>0}^{n}$ we let $\|\cdot\|_{w}$
denote the norm where $\norm{\vx}_{w}^{2}\defeq\sum_{i\in[n]}w_{i}x_{i}^{2}$
for all $\vx\in\Rn$. For any norm $\|\cdot\|$ and matrix $\mm$,
its induced operator norm of $\mm$ is defined by $\norm{\mm}=\sup_{\|\vx\|=1}\norm{\mm\vx}$.\textbf{\vspace{11pt}}\\
\textbf{Calculus:} For a function of two vectors, i.e. $g(x,y)\in\R$
for all $\vx\in\R^{n_{1}}$ and $\vy\in\R^{n_{2}}$, we let $\grad_{\vx}\vg(\va,\vb)\in\R^{n_{1}}$
denote the gradient of $\vg$ as a function of $x$ for fixed $\vy$
at $(\va,\vb)\in\R^{n_{1}\times n_{2}}$, i.e. $[\grad_{\vx}\vg(\va,\vb)]_{i}=\frac{\partial}{\partial x_{i}}g(a,b)$,
and define $\grad_{\vy}$, $\hessian_{\vx\vx}$, and $\hess_{\vy\vy}$
analogously. For $h:\R^{n}\rightarrow\R^{m}$ and $\vx\in\Rn$ we
let $\mj_{h}(\vx)\in\R^{m\times n}$ denote the Jacobian of $\vh$
at $\vx$, i.e. $[\mj_{h}(\vx)]_{ij}\defeq\frac{\partial}{\partial x_{j}}h(\vx)_{i}$
for all $i\in[m]$ and $j\in[n]$. For $f:\R^{n}\rightarrow\R$ and
$x,h\in\R^{n}$ we let $Df(x)[h]$ denote the directional derivative
of $f$ in direction $h$ at $x$, i.e. $Df(x)[h]\defeq\lim_{t\rightarrow0}[f(x+th)-f(x)]/t$.\textbf{\vspace{10pt}}\\
\textbf{Convex Sets:} We call $U\subseteq\R^{k}$ \emph{convex} if
$t\cdot\vx+(1-t)\cdot\vy\in U$ for all $\vx,\vy\in U$ and $t\in[0,1]$
and \emph{symmetric} if $\vx\in\R^{k}\Leftrightarrow-\vx\in\R^{k}$.
For all $\alpha>0$ and $U\subseteq\R^{k}$ we let $\alpha U\defeq\{\vx\in\R^{k}|\alpha^{-1}\vx\in U\}$.
For all $p\in[1,\infty]$ and $r>0$ we call the symmetric convex
set $\{\vx\in\R^{k}|\|\vx\|_{p}\leq r\}$ \emph{the $\ellP$ ball
of radius $r$.}\textbf{\vspace{10pt}}\\
\textbf{Misc:} For $z\in\Z$ we let $[z]\defeq\{1,2,..,z\}$. We let
$\indicVec i$ denote the vector that has value $1$ in coordinate
$i$ and $0$ elsewhere. We use $\otilde$ to hide factors polylogarithmic
in $m$, $n$, $U$, $|V|$, $|E|$, and $M$. 

\section{Weighted Path Finding\label{sec:weighted_path_finding}}

Here we introduce our \emph{weighted path finding} scheme for solving\emph{
}\eqref{eq:intro:two-sided}. First we introduce the weighted central
path (Section~\ref{subsec:weighted_central_path}) and provide key
properties of the path (Section~\ref{subsec:centrality}) and weight
functions (Section~\ref{subsec:weight_function}). Assuming a weight
function (shown to exist in Section~\ref{sec:weight-function}) we
then provide the main lemmas we need for an $\tilde{O}(\sqrt{\rank(\ma)}\log(U/\varepsilon))$
iteration weighted path following algorithm for \eqref{eq:intro:two-sided}.
In Section~\ref{subsec:weighted-path:t-delta}, \ref{subsec:weighted-path:x-delta}
and \ref{subsec:weighted-path:w-delta} we study the effect of changing
the path parameter, the point, and the weights, and in Section~\ref{subsec:weighted-path:centering}
we give our main subroutine for following the path.

\subsection{Preliminaries\label{sec:path:prelim}}

Recall that our goal is to efficiently solve \eqref{eq:intro:two-sided}
repeated below

\[
\min_{\begin{array}{c}
\vx\in\Rm~:~\ma^{\top}\vx=\vb\\
\forall i\in[m]~:~l_{i}\leq x_{i}\leq u_{i}
\end{array}}\vc^{\top}\vx\,.
\]
Here $\ma\in\Rmn$, $\vb\in\Rn$, $\vc\in\Rm$, $l_{i}\in\R\cup\{-\infty\}$,
and $u_{i}\in\R\cup\{+\infty\}$ and we assume that $\ma$ is non-degenerate,
that $\dom(x_{i})\defeq\{x\,:\,l_{i}<x<u_{i}\}$ is neither the empty
set or the entire real line for all $i\in[m]$, and the interior of
the polytope, $\dInterior\defeq\{\vx\in\Rm~:~\ma^{\top}\vx=\vb,l_{i}<x_{i}<u_{i}\}$
is non-empty.

Rather than working directly with the different domains of the $x_{i}$
we take a slightly more general approach and let $\phi_{i}:\dom(x_{i})\rightarrow\R$
for all $i\in[m]$ denote a 1-self-concordant barrier function for
$\dom(x_{i})$ (See Definition~\ref{def:self_concordance}). In the
remainder of the paper we will simply leverage that each $\phi_{i}$
is a $1$-self-concordant barrier for each of the $\dom(\phi_{i})$
and not use any further structure about the barriers or the domains.
It is easy to show that such $\phi_{i}$ exist and for completeness,
below we provide an explicit 1-self-concordant barrier function for
each possible $\dom(x_{i})$:
\begin{itemize}
\item \emph{Case (1):}\textbf{\emph{ }}\emph{$l_{i}$ finite and $u_{i}=+\infty$}:
We use a\emph{ log barrier }defined as $\phi_{i}(x)\defeq-\log(x-l_{i})$.
Here
\[
\phi_{i}'(x)=-\frac{1}{x-l_{i}}\enspace\text{,}\enspace\phi_{i}''(x)=\frac{1}{(x-l_{i})^{2}}\enspace\text{, and}\enspace\phi_{i}'''(x)=-\frac{2}{(x-l_{i})^{3}}
\]
and therefore clearly $|\phi_{i}'''(x)|=2(\phi_{i}''(x))^{3/2}$ ,
$|\phi_{i}'(x)|=\sqrt{\phi''_{i}(x)}$, and $\lim_{x\rightarrow l_{i}^{+}}\phi_{i}(x)=+\infty.$ 
\item \emph{Case (2): $l_{i}=-\infty$ and $u_{i}$ finite}: We use a\emph{
log barrier }defined as $\phi_{i}(x)\defeq-\log(u_{i}-x)$. Here
\[
\phi_{i}'(x)=\frac{1}{u_{i}-x}\enspace\text{,}\enspace\phi_{i}''(x)=\frac{1}{(u_{i}-x)^{2}}\enspace\text{, and}\enspace\phi_{i}'''(x)=-\frac{2}{(u_{i}-x)^{3}}
\]
and therefore clearly $|\phi_{i}'''(x)|=2(\phi_{i}''(x))^{3/2}$,
$|\phi_{i}'(x)|=\sqrt{\phi''_{i}(x)}$, and $\lim_{x\rightarrow u_{i}^{-}}\phi_{i}(x)=+\infty.$ 
\item \emph{Case (3): $l_{i}$ finite and $u_{i}$ finite}: We use a\emph{
trigonometric barrier }defined as $\phi_{i}(x)\defeq-\log\cos(a_{i}x+b_{i})$
for $a_{i}=\frac{\pi}{u_{i}-l_{i}}$ and $b_{i}=-\frac{\pi}{2}\frac{u_{i}+l_{i}}{u_{i}-l_{i}}$.
As $x\rightarrow u_{i}^{-}$ we have $a_{i}x+b_{i}\rightarrow\frac{\pi}{2}$
and as $x\rightarrow l_{i}^{+}$ we have $a_{i}x+b_{i}\rightarrow\frac{-\pi}{2}$
and therefore, in both cases $\phi_{i}(x)\rightarrow+\infty.$ Further,
\emph{
\[
\phi_{i}'(x)=a_{i}\tan\left(a_{i}x+b_{i}\right)\enspace,\enspace\phi_{i}''(x)=\frac{a_{i}^{2}}{\cos^{2}(a_{i}x+b_{i})}\enspace\text{, and}\enspace\phi_{i}'''=\frac{2a_{i}^{3}\sin(a_{i}x+b_{i})}{\cos^{3}(a_{i}x+b_{i})}.
\]
}Therefore, $|\phi_{i}'(x)|\leq a_{i}/\left|\cos\left(a_{i}x+b_{i}\right)\right|=\sqrt{\phi_{i}''(x)}$
and we have
\[
\left|\phi_{i}'''(x)\right|=\left|\frac{2a_{i}^{3}\sin(a_{i}x+b_{i})}{\cos^{3}(a_{i}x+b_{i})}\right|\leq\frac{2a_{i}^{3}}{|\cos(a_{i}x+b_{i})|^{3}}=2(\phi''(x))^{3/2}\,.
\]
\end{itemize}
While there is rich theory regarding self-concordance we will primarily
use only following two lemmas regarding $\phi_{i}$. Lemma~\ref{lem:gen:phi_properties_sim}
bounds the change in the Hessian of $\phi_{i}$ Lemma~\ref{lem:gen:phi_properties_force}
bounds the gradient of $\phi_{i}$.
\begin{lem}
[{\cite[Theorem 4.1.6]{Nesterov2003}}]\label{lem:gen:phi_properties_sim}
If $s\in\dom(\phi_{i})$ for $i\in[m]$, and $r\defeq\sqrt{\phi''_{i}(s)}\left|s-t\right|<1$
then $t\in\dom(\phi_{i})$ and $(1-r)\sqrt{\phi''_{i}(s)}\leq\sqrt{\phi''_{i}(t)}\leq(1-r)^{-1}\sqrt{\phi''_{i}(s)}$.
Therefore $\sqrt{\phi''_{i}(s)}\geq1/U$ where $U$ is the diameter
of $\dom(\phi_{i})$.
\end{lem}

\begin{lem}
[{\cite[Theorem 4.2.4]{Nesterov2003}}]\label{lem:gen:phi_properties_force}
$\phi'_{i}(x)\cdot(y-x)\leq1$ for all $x,y\in\dom(\phi_{i})$ and
$i\in[m]$.
\end{lem}

\subsection{The Weighted Central Path\label{subsec:weighted_central_path}}

Our path-finding algorithm maintains a feasible point $\vx\in\dInterior$,
weights $w\in\R_{>0}^{m}$, and minimizes the following \emph{penalized
objective function }for increasing $t$ and small $w$
\begin{equation}
\min_{\ma^{\top}\vx=\vb}f_{t}(\vx,w)\enspace\text{where }\enspace f_{t}(\vx,w)\defeq t\cdot\vc^{\top}\vx+\sum_{i\in[m]}w_{i}\phi_{i}(\vx_{i})\,.\label{eq:penalized_objective}
\end{equation}
For every fixed set of \emph{weights}, $w\in\R_{>0}^{m}$ the set
of points $\vx_{w}(t)=\argmin_{\vx\in\dInterior}f_{t}(\vx,w)$ for
$t\in[0,\infty)$ form a path through the interior of the polytope
that we call the \emph{weighted central path}. We call $\vx_{w}(0)$
a \emph{weighted center} of $\interior$ and note that $\lim_{t\rightarrow\infty}\vx_{w}(t)$
is a solution to \eqref{eq:intro:two-sided} (Lemma \ref{lem:weighted_path:duality_gap}). 

While all weighted central paths converge to a solution of the linear
program, different paths may have different algebraic properties which
either improve or impair the convergence of a path following scheme.
Consequently, our algorithm alternates between advancing down a central
path (i.e. increasing $t$), moving closer to the weighted central
path (i.e. updating $\vx$), and picking a better path (i.e. updating
the weights $w$). More formally, we assume we have a feasible point
$\{\vx,w\}\in\dFull$ and a weight function $\vg(\vx):\dInterior\rightarrow\R_{>0}^{m}$,
such that for any point $\vx\in\R_{>0}^{m}$ the function $\vg(\vx)$
returns a good set of weights that suggest a possibly better weighted
path. Our algorithm then repeats the following: (1) if $\vx$ close
to $\argmin_{\vy\in\Omega}f_{t}\left(\vy,w\right)$, then increase
$t$ (2) otherwise, use projected Newton step to update $\vx$ and
move $w$ closer to $\vg(\vx)$.

In the remainder of this section we present how we measure both the
quality of a current feasible point $\{\vx,w\}\in\dFull$, the quality
of the weight function, and with a weight function control centrality.
In Section~\ref{subsec:centrality} we derive and present both how
we measure how close $\{\vx,w\}$ is to the weighted central path
and the step we take to improve this \emph{centrality} and in Section~\ref{subsec:weight_function}
we present how we measure the quality of a weight function, i.e. how
good the weighted paths it finds are. The remaining subsection analyze
controlling centrality under changes to $x$, $w$, and $t$.

\subsection{Measuring Centrality\label{subsec:centrality}}

Here we explain how we measure the distance from $\vx$ to the minimum
of $f_{t}\left(\vx,\vWeight\right)$ for fixed $\vWeight$, denoted
$\delta_{t}(\vx,\vWeight)$. As $\delta_{t}(x,w)$ measures the proximity
of $\vx$ to the weighted central path, we call it a \emph{centrality}.
measure of $x$ and $w$. To motivate $\delta_{t}(\vx,\vWeight)$
we first compute a projected Newton step for $\vx$. For all $\vx\in\dInterior$,
we define $\vphi(\vx)\in\Rm$ by $\vphi(\vx)_{i}=\phi_{i}(\vx_{i})$
for $i\in[m]$, define $\vphi'(\vx)$, $\vphi''(\vx)$, and $\vphi'''(\vx)$
analogously, and let $\mPhi',\mPhi'',\mPhi'''$ denote their associated
diagonal matrices. This yields
\[
\grad_{x}f_{t}(\vx,w)=t\cdot\vc+w\vphi'(\vx)\enspace\text{ and }\enspace\grad_{xx}^{2}f_{t}(\vx,w)=\mw\mPhi''(\vx)\,.
\]
Lemma~\ref{lem:newton_step} (proved in the appendix) shows that
a Newton step for $\vx$ is given by
\begin{align}
\vh_{t}(\vx,\vWeight) & =-\left(\mi-\left(\mw\mPhi''(\vx)\right)^{-1}\ma(\ma^{\top}\left(\mw\mPhi''(\vx)\right)^{-1}\ma)^{-1}\ma^{\top}\right)\left(\mw\mPhi''(\vx)\right)^{-1}\grad_{x}f_{t}(\vx,\vWeight)\nonumber \\
 & =-\mPhi''(\vx)^{-1/2}\Pxw\mWeight^{-1}\mPhi''(\vx)^{-1/2}\grad_{x}f_{t}(\vx,\vWeight)\label{eq:newton_step}
\end{align}
where 
\begin{equation}
\mProj_{\vx,\vWeight}\defeq\iMatrix-\mWeight^{-1}\ma_{x}\left(\ma_{x}^{\top}\mWeight^{-1}\ma_{x}\right)^{-1}\ma_{x}^{\top}\enspace\text{ for }\enspace\ma_{x}\defeq\mPhi''(\vx)^{-1/2}\ma\,.\label{eq:def_Pxw}
\end{equation}
As with standard convergence analysis of interior point methods, we
wish to keep the Newton step size in the Hessian norm, i.e. $\norm{\vh_{t}(\vx,\vWeight)}_{w\phi''(\vx)}=\norm{\sqrt{\vphi''(\vx)}\vh_{t}(\vx,\vWeight)}_{w}$,
small and the multiplicative change in the Hessian, $\norm{\sqrt{\phi''(\vx)}\vh_{t}(\vx,\vWeight)}_{\infty}$,
small. While in standard logarithmic barrier analysis, i.e. $w_{i}=\onesVec$
for all $i$, we can bound the multiplicative change by the change
in the hessian norm (since $\|\cdot\|_{\infty}\leq\|\cdot\|_{2}$),
here we would like to use small weights and this comparison would
be insufficient.

To track both these quantities simultaneously, we define the \emph{mixed
norm} for all $\vy\in\Rm$ by
\begin{equation}
\mixedNorm{\vy}w\defeq\norm{\vy}_{\infty}+\cnorm\norm{\vy}_{w}\label{eq:mixed_norm}
\end{equation}
for $\cnorm>0$ defined in Definition~\ref{def:gen:weight_function}.
Note that $\mixedNorm{\cdot}{\vWeight}$ is indeed a norm for $\vWeight\in\R_{>0}^{m}$
as in this case both $\normInf{\cdot}$and $\norm{\cdot}_{w}$ are
norms. However, rather than measuring centrality by the quantity 
\[
\mixedNorm{\sqrt{\vphi''(\vx)}\vh_{t}(\vx,\vWeight)}w=\mixedNormFull{\mProj_{\vx,\vWeight}\left(\frac{\grad_{x}f_{t}(\vx,\vWeight)}{\vWeight\sqrt{\phi''(x)}}\right)}w
\]
we instead find it more convenient to use the following idealized
form
\[
\delta_{t}(\vx,\vWeight)\defeq\min_{\veta\in\Rn}\normFull{\frac{\grad_{x}f_{t}(\vx,\vWeight)-\ma\veta}{\vWeight\sqrt{\vphi''(\vx)}}}_{w+\infty}.
\]
This definition is justified by the following lemma which shows that
these two quantities differ by at most a multiplicative factor of
$\mixedNorm{\mProj_{\vx,\vWeight}}w$.
\begin{lem}
\label{lem:max_flow:projection_lemma} For any norm $\norm{\cdot}$
and $\norm{\vy}_{Q}\defeq\min_{\veta\in\Rn}\normFull{\vy-\frac{\ma\veta}{\vWeight\sqrt{\phi''(\vx)}}}$,
we have 
\[
\norm{\vy}_{Q}\leq\normFull{\Pxw\vy}\leq\normFull{\Pxw}\cdot\normFull{\vy}_{Q}
\]
 and therefore for all $\{\vx,\vWeight\}\in\dFull$ we have\textup{
\begin{equation}
\delta_{t}(\vx,\vWeight)\leq\mixedNorm{\sqrt{\vphi''(\vx)}\vh_{t}(\vx,\vWeight)}w\leq\mixedNorm{\mProj_{\vx,\vWeight}}w\cdot\delta_{t}(\vx,\vWeight).\label{eq:centrality_equivalence}
\end{equation}
}
\end{lem}

\begin{proof}
By definition $\Pxw\vy=\vy-\frac{\ma\veta_{y}}{\vWeight\sqrt{\phi''(\vx)}}$
for some $\veta_{y}\in\Rn$. Consequently,
\begin{eqnarray*}
\norm{\vy}_{Q} & = & \min_{\veta\in\Rn}\normFull{\vy-\frac{\ma\eta}{\vWeight\sqrt{\phi''(\vx)}}}\leq\normFull{\Pxw\vy}.
\end{eqnarray*}
Further, letting $\eta_{q}$ be such that $\norm{\vy}_{Q}=\normFull{\vy-\frac{\ma\veta_{q}}{\vWeight\sqrt{\phi''(\vx)}}}$
and noting $\Pxw\mWeight^{-1}(\mPhi'')^{-1/2}\ma=\mZero$ yields
\begin{eqnarray*}
\normFull{\Pxw\vy} & = & \normFull{\Pxw\left(\vy-\frac{\ma\veta_{q}}{\vWeight\sqrt{\phi''}}\right)}\leq\normFull{\Pxw}\cdot\normFull{\vy-\frac{\ma\veta_{q}}{\vWeight\sqrt{\phi''}}}=\normFull{\Pxw}\cdot\normFull{\vy}_{Q}.
\end{eqnarray*}
\end{proof}
We summarize this section with the following definition.
\begin{defn}[Centrality Measure]
\label{Def:centrality_measure} For $\{\vx,\vWeight\}\in\dFull$
and $t\geq0$, we let $\vh_{t}(\vx,\vWeight)$ denote the \emph{projected
newton step} for $\vx$ on the penalized objective $f_{t}$ given
by
\[
\vh_{t}(\vx,\vWeight)\defeq-\frac{1}{\sqrt{\phi''(\vx)}}\mProj_{\vx,\vWeight}\left(\frac{\grad_{x}f_{t}(\vx,\vWeight)}{\vWeight\sqrt{\phi''(\vx)}}\right)
\]
where $\mProj_{\vx,\vWeight}$ is defined in \eqref{eq:def_Pxw}.
We measure the \emph{centrality} of $\{\vx,\vWeight\}$ by
\begin{equation}
\delta_{t}(\vx,\vWeight)\defeq\min_{\veta\in\Rn}\normFull{\frac{\grad_{x}f_{t}(\vx,\vWeight)-\ma\veta}{\vWeight\sqrt{\vphi''(\vx)}}}_{\vWeight+\infty}\label{eq:centrality_definition}
\end{equation}
where for all $\vy\in\Rm$ we let $\mixedNorm{\vy}{\vWeight}\defeq\norm{\vy}_{\infty}+\cnorm\norm{\vy}_{\mWeight}$
for $\cnorm>0$ defined in Definition~\ref{def:gen:weight_function}.
\end{defn}

\subsection{The Weight Function\label{subsec:weight_function}}

With the Newton step and centrality conditions defined, the specification
of our algorithm becomes cleaerr. Our algorithm simply repeatedly
(1) increases $t$ provided $\delta_{t}(\vx,\vWeight)$ is small and
(2) decreases $\delta_{t}(\vx,\vWeight)$ by setting $\next{\vx}\leftarrow\vx+\vh_{t}(\vx,\vWeight)$
and (3) moving $\next{\vWeight}$ towards $\vg(\next{\vx})$ for some
\emph{weight function} $\vg(\vx):\dInterior\rightarrow\R_{>0}^{m}$.
To prove this algorithm converges, we need to show what happens to
$\delta_{t}\left(\vx,\vWeight\right)$ when we change $t$, $\vx$,
$\vWeight$. At the heart of this paper is understanding what conditions
we need to impose on the weight function $g$ so that we can bound
this change in $\delta_{t}(\vx,\vWeight)$ and hence achieve a fast
convergent rate. In Lemma~\ref{lem:gen:t_step} we show that the
effect of changing $t$ on $\delta_{t}$ is bounded by $\cnorm$ and
$\norm{\vg(\vx)}_{1}$, in Lemma~\ref{lem:gen:x_progress} we show
that the effect that a Newton Step on $\vx$ has on $\delta_{t}$
is bounded by $\mixedNorm{\mProj_{\vx,\vg(\vx)}}{\vg(\vx)}$, and
in Lemma~\ref{lem:gen:w_step} and \ref{lem:gen:w_change} we show
the change of $\vWeight$ as $\vg(\vx)$ changes is bounded by $\mixedNorm{\fmWeight(\vx)^{-1}\mj_{g}(x)(\mPhi''(\vx))^{-1/2}}{\vg(\vx)}$. 

For the remainder of the paper we assume we have a weight function
$\vg(\vx):\dInterior\rightarrow\R_{>0}^{m}$ and make the following
assumptions regarding our weight function. In Section~\ref{sec:weight-function}
we prove that one exists. 
\begin{defn}
[Weight Function]\label{def:gen:weight_function} Differentiable
$\fvWeight:\dInterior\rightarrow\R_{>0}^{m}$ is a \emph{$(\cWeightSize,\cWeightStab$,$c_{k}$)
-weight function} if the following hold for all $\vx\in\dInterior$
and $i\in[m]$:

\begin{itemize}
\item The \emph{size, $c_{1}$, }satisfies \emph{$\cWeightSize\geq\max\{1,\normOne{\fvWeight(\vx)}\}$.
}This bounds how quickly centrality changes as $t$ changes.
\item The \emph{sensitivity,} $c_{s}$, satisfies $c_{s}\geq\cordVec i^{\top}\mg(x)^{-1}\ma_{x}\left(\ma_{x}^{\top}\mg(x)^{-1}\ma_{x}\right)^{-1}\ma_{x}^{\top}\mg(x)^{-1}\cordVec i$.
This bounds how quickly Hessian change as $x$ changes.
\item The\emph{ consistency, }$c_{k}$,\emph{ }satisfies $\mixedNorm{\fmWeight(\vx)^{-1}\mj_{g}(x)(\mPhi''(\vx))^{-1/2}}{\vg(\vx)}\leq1-c_{k}^{-1}<1$.
This bounds how much the weights change as $x$ changes, thereby governing
how consistent the weights are with changes to $x$ along the weighted
central path.
\end{itemize}
Through we assume we have such a weight function and define $\cnorm\defeq24\sqrt{c_{s}}c_{k}$.
\end{defn}

To motivate slack sensitivity, we show that it bounds $\mixedNorm{\mProj_{\vx,\vWeight}}{\vWeight}$.
This is used in Lemma~\ref{lem:gen:x_progress}.
\begin{lem}
\label{lem:c_gamma} For any $\vWeight$ such that $\frac{4}{5}\vg(x)\leq\vWeight\leq\frac{5}{4}\vg(x)$,
we have that 
\[
\mixedNorm{\mProj_{\vx,\vWeight}}{\vWeight}\leq\cWeightStab\quad\text{where}\quad\cWeightStab\defeq1+\frac{\sqrt{2c_{s}}}{\cnorm}\leq1+\frac{1}{16c_{k}}.
\]
\end{lem}

\begin{proof}
Letting $\normFull{\iMatrix-\mProj_{\vx,\vWeight}}_{w\rightarrow\infty}\defeq\max_{\norm z_{w}=1}\norm{(\iMatrix-\mProj_{\vx,\vWeight})z}_{\infty}$
we see that for any $y\in\Rm$, we have 
\begin{align}
\mixedNorm{\mProj_{\vx,\vWeight}\vy}{\vWeight} & =\normFull{\mProj_{\vx,\vWeight}\vy}_{\infty}+\cnorm\normFull{\mProj_{\vx,\vWeight}\vy}_{\vWeight}\leq\normFull{\vy}_{\infty}+\normFull{(\iMatrix-\mProj_{\vx,\vWeight})\vy}_{\infty}+\cnorm\normFull{\vy}_{\vWeight}\nonumber \\
 & \leq\normFull{\vy}_{\infty}+(\mixedNorm{\mi-\mProj_{x,w}}w+\cnorm)\normFull{\vy}_{\vWeight}\leq\left(1+\frac{\normFull{\iMatrix-\mProj_{\vx,\vWeight}}_{\vWeight\rightarrow\infty}}{\cnorm}\right)\mixedNorm{\vy}{\vWeight}\,.\label{eq:gamma_1}
\end{align}
where we used the fact that $\normFull{\mProj_{\vx,\vWeight}\vy}_{\vWeight}\leq\normFull{\vy}_{\vWeight}$
for all $y$ in the first inequality. Further, note that
\begin{align*}
\normFull{\iMatrix-\mProj_{\vx,\vWeight}}_{\vWeight\rightarrow\infty}^{2} & =\max_{i\in[m]}\max_{\|\vy\|_{\vWeight}\leq1}(\cordVec i^{\top}(\iMatrix-\mProj_{\vx,\vWeight})\vy)^{2}\leq\max_{i\in[m]}\norm{\cordVec i^{\top}(\iMatrix-\mProj_{\vx,\vWeight})\mw^{-\frac{1}{2}}}_{2}^{2}\\
 & =\max_{i\in[m]}\cordVec i^{\top}\mWeight^{-1}\ma_{x}\left(\ma_{x}^{\top}\mWeight^{-1}\ma_{x}\right)^{-1}\ma_{x}^{\top}\mWeight^{-1}\cordVec i.
\end{align*}
Since $\frac{4}{5}\vg(x)\leq\vWeight\leq\frac{5}{4}\vg(x)$, we have
that 
\begin{equation}
\normFull{\iMatrix-\mProj_{\vx,\vWeight}}_{\vWeight\rightarrow\infty}^{2}\leq2\max_{i\in[m]}\cordVec i^{\top}\mg(x)^{-1}\ma_{x}\left(\ma_{x}^{\top}\mg(x)^{-1}\ma_{x}\right)^{-1}\ma_{x}^{\top}\mg(x)^{-1}\cordVec i\leq2c_{s}.\label{eq:gamma_2}
\end{equation}
Combing \eqref{eq:gamma_1} and \eqref{eq:gamma_2} and using $\cnorm=24\sqrt{c_{s}}c_{k}$
yields the claims. 
\end{proof}

\subsection{Changing $t$\label{subsec:weighted-path:t-delta}}

Here we bound how much centrality increases as we increase $t$. We
show that this rate of increase is governed by $\cnorm$ and $\norm{\vWeight}_{1}.$
\begin{lem}
\label{lem:gen:t_step}For all $\{\vx,\vWeight\}\in\dFull$, $t>0$
and $\alpha\geq0$, we have
\[
\delta_{(1+\alpha)t}(\vx,\vWeight)\leq(1+\alpha)\delta_{t}(\vx,\vWeight)+\alpha(1+\cnorm\sqrt{\norm{\vWeight}_{1}})\,.
\]
\end{lem}

\begin{proof}
Let $\veta_{t}\in\Rn$ be such that
\[
\delta_{t}(\vx,\vWeight)=\mixedNormFull{\frac{\grad_{x}f_{t}(\vx,\vWeight)-\ma\veta_{t}}{\vWeight\sqrt{\vphi''(\vx)}}}w=\mixedNormFull{\frac{t\cdot\vc+\vWeight\vphi'(\vx)-\ma\veta_{t}}{\vWeight\sqrt{\vphi''(\vx)}}}w.
\]
Applying this to the definition of $\delta_{(1+\alpha)t}$ and using
that $\ensuremath{\mixedNorm{\cdot}{\vWeight}}$ is a norm then yields
\begin{align*}
\delta_{(1+\alpha)t}(\vx,\vWeight) & =\min_{\veta\in\Rn}\mixedNormFull{\frac{(1+\alpha)t\cdot\vc+\vWeight\vphi'(\vx)-\ma\veta}{\vWeight\sqrt{\vphi''(\vx)}}}w\leq\mixedNormFull{\frac{(1+\alpha)t\cdot\vc+\vWeight\vphi'(\vx)-(1+\alpha)\ma\veta_{t}}{\vWeight\sqrt{\vphi''(\vx)}}}w\\
 & \leq(1+\alpha)\mixedNormFull{\frac{t\cdot\vc+\vWeight\vphi'(\vx)+\ma\veta_{t}}{\vWeight\sqrt{\vphi''(\vx)}}}w+\alpha\mixedNormFull{\frac{\vphi'(\vx)}{\sqrt{\vphi''(\vx)}}}w\\
 & =(1+\alpha)\delta_{t}(\vx,\vWeight)+\alpha\left(\normFullInf{\frac{\vphi'(\vx)}{\sqrt{\vphi''(\vx)}}}+\cnorm\normFull{\frac{\vphi'(\vx)}{\sqrt{\vphi''(\vx)}}}_{w}\right)
\end{align*}
The result follows from the fact that $|\phi_{i}'(\vx)|\leq\sqrt{\phi_{i}''(\vx)}$
for all $i\in[m]$ and $\vx\in\Rm$ by Definition~\ref{def:self_concordance}.
\end{proof}

\subsection{Changing $\protect\vx$\label{subsec:weighted-path:x-delta}}

Here we analyze the effect of a Newton step of $\vx$ on centrality.
We show for sufficiently central $\{\vx,\vWeight\}\in\dFull$ and
$\vWeight$ sufficiently close to $\vg(\vx)$ Newton steps converge
quadratically.
\begin{lem}
\label{lem:gen:x_progress}Let $\{\vx_{0},\vWeight\}\in\dFull$ such
that $\delta_{t}(\vx_{0},\vWeight)\leq\frac{1}{10}$ and $\frac{4}{5}\vg(x)\leq\vWeight\leq\frac{5}{4}\vg(x)$
and consider a Newton step $\vx_{1}=\vx_{0}+\vh_{t}(\vx,\vWeight)$.
Then, $\delta_{t}(\vx_{1},\vWeight)\leq4(\delta_{t}(\vx_{0},\vWeight))^{2}.$
\end{lem}

\begin{proof}
Let $\vphi_{0}\defeq\vphi(\vx_{0})$ and let $\vphi_{1}\defeq\vphi(\vx_{1})$.
By the definition of $\vh_{t}(\vx_{0},\vWeight)$ and the formula
of $\mProj_{\vx_{0},\vWeight}$ we know that there is some $\veta_{0}\in\Rn$
such that
\[
-\sqrt{\phi_{0}''}\vh_{t}(\vx_{0},\vWeight)=\frac{t\cdot\vc+\vWeight\vphi_{0}'-\ma\veta_{0}}{\vWeight\sqrt{\vphi_{0}''}}.
\]
Therefore, $\ma\veta_{0}=t\cdot\vc+\vWeight\phi_{0}'+\vWeight\phi_{0}''h_{t}(\vx_{0},\vWeight)$.
Recalling the definition of $\delta_{t}$ this implies that
\begin{eqnarray*}
\delta_{t}(\vx_{1},\vWeight) & = & \min_{\veta\in\Rn}\mixedNormFull{\frac{t\cdot\vc+\vWeight\phi_{1}'-\ma\veta}{\vWeight\sqrt{\phi_{1}''}}}w\leq\mixedNormFull{\frac{t\cdot\vc+\vWeight\phi_{1}'-\ma\veta_{0}}{\vWeight\sqrt{\phi_{1}''}}}w\\
 & \leq & \mixedNormFull{\frac{\vWeight(\phi_{1}'-\vphi_{0}')-\vWeight\vphi_{0}''\vh_{t}(\vx_{0},\vWeight)}{\vWeight\sqrt{\phi_{1}''}}}w=\mixedNormFull{\frac{(\phi_{1}'-\vphi_{0}')-\vphi_{0}''\vh_{t}(\vx_{0},\vWeight)}{\sqrt{\phi_{1}''}}}w\,.
\end{eqnarray*}
By mean value theorem $\phi_{1}'-\vphi_{0}'=\phi''(\theta)\vh_{t}(\vx_{0},\vWeight)$
for $\theta$ between $\vx_{0}$ and $\vx_{1}$ coordinate-wise. Hence,
\begin{eqnarray}
\delta_{t}(\vx_{1},\vWeight) & \leq & \mixedNormFull{\frac{\phi''(\theta)\vh_{t}(\vx_{0},\vWeight)-\vphi_{0}''\vh_{t}(\vx_{0},\vWeight)}{\sqrt{\phi_{1}''}}}w=\mixedNormFull{\frac{\left(\phi''(\theta)-\phi_{0}''\right)}{\sqrt{\phi_{1}''}\sqrt{\vphi_{0}''}}(\sqrt{\phi_{0}''}\vh_{t}(\vx_{0},\vWeight))}w\nonumber \\
 & \leq & \normFull{\frac{\phi''(\theta)-\phi_{0}''}{\sqrt{\phi_{1}''}\sqrt{\vphi_{0}''}}}_{\infty}\cdot\mixedNormFull{\sqrt{\phi_{0}''}\vh_{t}(\vx_{0},\vWeight)}w.\label{eq:gen:x_progress1}
\end{eqnarray}
To bound the first term, we use Lemma~\ref{lem:gen:phi_properties_sim}
as follows
\begin{eqnarray*}
\normFull{\frac{\phi''(\theta)-\phi_{0}''}{\sqrt{\phi_{1}''}\sqrt{\phi_{0}''}}}_{\infty} & \leq & \normFull{\frac{\vphi''(\theta)}{\vphi_{0}''}-\onesVec}_{\infty}\cdot\normFull{\frac{\sqrt{\phi_{0}''}}{\sqrt{\phi_{1}''}}}_{\infty}\\
 & \leq & \left|\left(1-\normFull{\sqrt{\phi_{0}''}\vh_{t}(\vx_{0},\vWeight)}_{\infty}\right)^{-2}-1\right|\cdot\left(1-\normFull{\sqrt{\phi_{0}''}\vh_{t}(\vx_{0},\vWeight)}_{\infty}\right)^{-1}.
\end{eqnarray*}
Using (\ref{eq:centrality_equivalence}), i.e. Lemma \ref{lem:max_flow:projection_lemma},
the bound $c_{\gamma}\leq2$ (Lemma \ref{lem:c_gamma}), and that
$\delta_{t}(\vx_{0},\vWeight)\leq\frac{1}{10}$ yields
\[
\normFull{\sqrt{\phi_{0}''}\vh_{t}(\vx_{0},\vWeight)}_{\infty}\leq\mixedNormFull{\sqrt{\phi_{0}''}\vh_{t}(\vx_{0},\vWeight)}w\leq c_{\gamma}\cdot\delta_{t}(\vx_{0},\vWeight)\leq\frac{1}{5}.
\]
Using $\left((1-t)^{-2}-1\right)\cdot(1-t)^{-1}\leq4t$ for $t\leq1/5$,
we have
\begin{equation}
\normFull{\frac{\phi''(\theta)-\phi_{0}''}{\sqrt{\phi_{1}''}\sqrt{\phi_{0}''}}}_{\infty}\leq4\normFull{\sqrt{\phi_{0}''}h_{t}(\vx_{0},\vWeight)}_{\infty}.\label{eq:gen:x_progress2}
\end{equation}
Combining the formulas (\ref{eq:gen:x_progress1}) and (\ref{eq:gen:x_progress2})
yields that $\delta_{t}(\vx_{1},\vWeight)\leq4(\delta_{t}(\vx_{0},\vWeight))^{2}$
as desired.
\end{proof}

\subsection{Changing $\protect\vWeight$\label{subsec:weighted-path:w-delta}}

In Section~\ref{subsec:weighted-path:x-delta} we used the assumption
that the weights, $\vWeight$, were multiplicatively close to $\vg(\vx)$,
for the current point $\vx\in\dInterior.$ To maintain this invariant
when we change $\vx$ we will need to change $\vWeight$ to move it
closer to $\vg(\vx).$ Here we bound how much $\vg(\vx)$ can move
as we move $\vx$ (Lemma \ref{lem:gen:w_step}) and we bound how much
changing $\vWeight$ can decrease centrality (Lemma \ref{lem:gen:w_change}).
Together these lemmas will allow us to show that we can keep $\vWeight$
close to $\vg(\vx)$ while still improving centrality (Section~\ref{subsec:weighted-path:centering}).
\begin{lem}
\label{lem:gen:w_step} For all $t\in[0,1]$, let $\vx_{t}\defeq\vx_{0}+t\vDelta$
for $\vDelta\in\Rm$ and $\vg_{t}=\vg(\vx_{t})$ such that $\vx_{t}\in\dInterior$.
Then for $\epsilon=\mixedNorm{\sqrt{\phi_{0}''}\vDelta}{\vg_{0}}\leq\frac{1}{10}$
we have $\mixedNorm{\log(g_{1})-\log(g_{0})}{\vg_{0}}\leq(1-c_{k}(g)^{-1}+4\epsilon)\epsilon\leq\frac{1}{5}$
\textup{and for all $s,t\in[0,1]$ and for all $\vy\in\Rm$ we ha}\emph{ve}\textup{
$\mixedNorm{\vy}{\vg_{s}}\leq(1+2\epsilon)\mixedNorm{\vy}{\vg_{t}}$.}
\end{lem}

\begin{proof}
Let $\vq:[0,1]\rightarrow\R^{m}$ be given by $\vq(t)\defeq\log(\vg_{t})$
for all $t\in[0,1]$. Then, $\vq'(t)=\mg_{t}^{-1}\mj_{g}(x_{t})\vDelta_{x}$.
Letting $Q(t)\defeq\mixedNorm{q(t)-q(0)}{\vg_{0}}$ and using Jensen's
inequality yields that for all $u\in[0,1]$,
\begin{eqnarray*}
Q(u) & \leq & \overline{Q}(u)\defeq\int_{0}^{u}\mixedNormFull{\mg_{t}^{-1}\mj_{g}(x_{t})(\mPhi_{t}'')^{-1/2}}{g_{0}}\mixedNormFull{\sqrt{\phi_{t}''}\vDelta}{g_{0}}dt.
\end{eqnarray*}
Using Lemma~\ref{lem:gen:phi_properties_sim} and $\epsilon\leq\frac{1}{10}$,
we have for all $t\in[0,1]$,
\[
\mixedNormFull{\sqrt{\phi_{t}''}\vDelta_{x}}{g_{0}}\leq\normFull{\sqrt{\frac{\phi_{t}''}{\phi_{0}''}}}_{\infty}\mixedNormFull{\sqrt{\phi_{0}''}\vDelta_{x}}{g_{0}}\leq\frac{\mixedNorm{\sqrt{\phi_{0}''}\vDelta_{x}}{\vg_{0}}}{1-\normFull{\sqrt{\phi_{0}''}\vDelta_{x}}_{\infty}}\leq\frac{\epsilon}{1-\epsilon}.
\]
Thus, we have
\begin{eqnarray}
\overline{Q}(u) & \leq & \frac{\epsilon}{1-\epsilon}\int_{0}^{u}\mixedNorm{\mg_{t}^{-1}\mj_{g}(x_{t})(\mPhi_{t}'')^{-1/2}}{\vg_{0}}dt.\label{eq:change_of_g_norm_est_int-1}
\end{eqnarray}
Note that $\overline{Q}$ is monotonically increasing. Let $\theta=\sup_{u\in[0,1]}\{\overline{Q}(u)\leq(1-c_{k}^{-1}+4\epsilon)\epsilon\}$.
Since $\|q(t)-q(0)\|_{\infty}\leq\overline{Q}(\theta)\leq\frac{1}{2}$
and $q(t)=\log(g_{t})$, we know that for all $s,t\in[0,\theta]$,
we have
\[
\normFull{\frac{g_{s}-g_{t}}{g_{t}}}_{\infty}\leq\norm{q(s)-q(t)}_{\infty}+\norm{q(s)-q(t)}_{\infty}^{2}
\]
and therefore $\normFull{\vg_{s}/\vg_{t}}_{\infty}\leq(1+\norm{q(s)-q(t)}_{\infty})^{2}\leq(1+(1-c_{k}^{-1}+4\epsilon)\epsilon)^{2}$.
Consequently,
\begin{eqnarray*}
\mixedNorm{\vy}{\vg_{s}} & \leq & (1+(1-c_{k}^{-1}+4\epsilon)\epsilon)\mixedNorm{\vy}{\vg_{t}}\leq\left(1+2\epsilon\right)\mixedNorm{\vy}{\vg_{t}}.
\end{eqnarray*}
Using (\ref{eq:change_of_g_norm_est_int-1}), we have for all $u\in[0,\theta]$,
\begin{eqnarray*}
Q(u)\leq\overline{Q}(u) & \leq & \frac{\epsilon}{1-\epsilon}\int_{0}^{u}\mixedNorm{\mg_{t}^{-1}\mj_{g}(x_{t})(\mPhi_{t}'')^{-1/2}}{\vg_{0}}dt\\
 & \leq & \frac{\epsilon}{1-\epsilon}\int_{0}^{u}\left(1+2\epsilon\right)\mixedNorm{\mg_{t}^{-1}\mj_{g}(x_{t})(\mPhi_{t}'')^{-1/2}}{\vg_{t}}dt\\
 & \leq & \frac{\epsilon}{1-\epsilon}(1+2\epsilon)(1-c_{k}^{-1})\theta<(1-c_{k}^{-1}+4\epsilon)\epsilon.
\end{eqnarray*}
Consequently, $\theta=1$ and we have the desired result by the above
bound on $Q(1)$.
\end{proof}
\begin{lem}
\label{lem:gen:w_change}Let $\vv,\vWeight\in\R_{>0}^{m}$ such that
$\mbox{\ensuremath{\epsilon=\mixedNorm{\log(\vWeight)-\log(\vv)}{\vWeight}}}\le\frac{1}{10}$.
Then for $\vx\in\dInterior$ we have
\[
\delta_{t}(\vx,\vv)\leq(1+4\epsilon)(\delta_{t}(\vx,\vWeight)+\epsilon).
\]
\end{lem}

\begin{proof}
Let $\veta_{w}$ be such that
\begin{equation}
\delta_{t}(\vx,\vWeight)=\mixedNormFull{\frac{\vc+\vWeight\phi'(\vx)-\ma\veta_{w}}{\vWeight\sqrt{\phi''(\vx)}}}w\,.\label{eq:gen:w_change:1}
\end{equation}
The assumptions imply that $(1+\epsilon)^{-2}\vWeight_{i}\leq\vv_{i}\leq(1+\epsilon)^{2}\vWeight_{i}$
for all $i$ and consequently
\begin{align*}
\delta_{t}(\vx,\vv) & =\min_{\eta}\mixedNormFull{\frac{\vc+\vv\phi'(\vx)-\ma\veta}{\vv\sqrt{\phi''(\vx)}}}v\leq\mixedNormFull{\frac{\vc+\vv\phi'(\vx)-\ma\veta_{w}}{\vv\sqrt{\phi''(\vx)}}}v\leq(1+\epsilon)\mixedNormFull{\frac{\vc+\vv\phi'(\vx)-\ma\veta_{w}}{\vv\sqrt{\phi''(\vx)}}}w\\
 & \leq\left(1+\epsilon\right)\cdot\left(\mixedNormFull{\frac{\vc+\vWeight\phi'(\vx)-\ma\veta_{w}}{\vv\sqrt{\phi''(\vx)}}}w+\mixedNormFull{\frac{(\vv-\vWeight)\phi'(\vx)}{\vv\sqrt{\phi''(\vx)}}}w\right)\\
 & \leq\left(1+\epsilon\right)^{3}\delta_{t}(\vx,\vWeight)+(1+\epsilon)\cdot\normFullInf{\frac{\vphi'(\vx)}{\sqrt{\vphi''(\vx)}}}\cdot\mixedNormFull{\frac{(\vv-\vWeight)}{\vv}}w\,.
\end{align*}
The result follows from the fact that $|\phi_{i}'(\vx)|\leq\sqrt{\phi_{i}''(\vx)}$
for all $i\in[m]$ and $\vx\in\Rm$ by Definition~\ref{def:self_concordance}.
\end{proof}

\subsection{Centering\label{subsec:weighted-path:centering}}

The results of the previous sections imply an efficient linear programming
algorithm provided the weight function can be computed efficiently
to high-precision. Unfortunately such a weight computation algorithm
is unknown and instead only efficient approximate weight computation
algorithms are presently available (See Appendix~\ref{sec:weights_full:computing}).
Consequently, here we show how to improve centrality even when the
weight function is only computed approximately. Our algorithm is based
on a solution to the ``$\ell_{\infty}$ Chasing Game'' summarized
in Theorem~\ref{thm:zerogame} and proved in Appendix~\ref{sec:approx_path:chasing_zero}.

\begin{restatable}[$\ell_\infty$ Chasing Game]{thm}{zerogame}\label{thm:zerogame}For
$x^{(0)},y^{(0)}\in\Rm$ and $\epsilon\in(0,1/5)$, consider the two
player game consisting of repeating the following for $k=1,2,\ldots$

\begin{enumerate}
\item The adversary chooses $\trSetCurr\subseteq\R^{m}$, $\vu^{(k)}\in U^{(k)}$,
and sets $\trAdve=y^{(k-1)}+\vu^{(k)}$.
\item The adversary chooses $\trMeas$ with $\norm{\trMeas-\trAdve}_{\infty}\leq R$
and reveals $\vz^{(k)}$ and $U^{(k)}$ to the player.
\item The player chooses $\Delta^{(k)}\in(1+\epsilon)\trSetCurr$ and sets
$x^{(k)}=x^{(k-1)}+\Delta^{(k)}.$
\end{enumerate}
Suppose that each $\trSetCurr$ is a symmetric convex set that contains
an $\ellInf$ ball of radius $r_{k}$ and is contained in a $\ellInf$
ball of radius $R_{k}\leq R$ and the player plays the strategy
\[
\Delta^{(k)}=\argmin_{\Delta\in(1+\epsilon)U^{(k)}}\left\langle \nabla\Phi_{\mu}(x^{(k-1)}-z^{(k)}),\Delta\right\rangle \text{ where }\Phi_{\mu}(\vx)\defeq\sum_{i\in[m]}\left(e^{\mu x_{i}}+e^{-\mu x_{i}}\right)\text{ and }\mu\defeq\frac{\epsilon}{12R}\,.
\]
If $\Phi_{\mu}(x^{(0)}-y^{(0)})\leq\frac{12m\tau}{\epsilon}$ for
$\tau=\max_{k}\frac{R_{k}}{r_{k}}$ then this strategy guarantees
that for all $k$ we have
\[
\Phi_{\mu}(x^{(k)}-y^{(k)})\leq\frac{12m\tau}{\epsilon}\quad\text{and}\quad\norm{x^{(k)}-y^{(k)}}_{\infty}\leq\frac{12R}{\epsilon}\log\left(\frac{12m\tau}{\epsilon}\right).
\]
\end{restatable}

We can think of the problem of maintaining weights as playing this
game; we want to keep $\vWeight$ is close to $\vg(\vx)$ while the
adversary controls the change in $\vg(\vx)$ and the noise in approximating
$\vg(\vx)$. Theorem~\ref{thm:zerogame} shows that we control the
error $\ell_{\infty}$ if we can approximate $\vg(\vx)$ in $\ell_{\infty}$.
Since we wish to maintain multiplicative approximations to $g(x)$
we play this game with the $\log$ of $w$ and $g(x)$. Formally,
our goal is to not move $\vWeight$ too much in $\mixedNorm{\cdot}w$
while keeping $\|\log(g(x))-\log(\vWeight)\|_{\infty}\leq K$ for
some error $K$ just small enough to not impair our ability to decrease
$\delta_{t}$ and approximate $\vg$.

\begin{algorithm2e}[H]

\caption{$(\next{\vx},\next{\vWeight})=\code{centeringInexact}(\vx,\vWeight,K)$}

\SetAlgoLined

$R=\frac{K}{48c_{k}\log(36c_{1}c_{s}c_{k}m)}$, $\delta=\delta_{t}(\vx,\vWeight)$
and $\epsilon=\frac{1}{2c_{k}}$.

$\next{\vx}=\vx-\frac{1}{\sqrt{\phi''(\vx)}}\mProj_{\vx,\vWeight}\left(\frac{t\vc-\vWeight\phi'(\vx)}{\vWeight\sqrt{\phi''(\vx)}}\right).$

Let $U=\{\vx\in\Rm~|~\mixedNorm{\vx}{\vWeight}\leq(1-\frac{6}{7c_{k}})\delta\}$.

Find $\vz$ such that $\normInf{\vz-\log(\vg(\next{\vx}))}\leq R$.

$\next{\vWeight}=\exp\left(\log(\vWeight)+\argmin_{\vu\in(1+\epsilon)U}\left\langle \nabla\Phi_{\frac{\epsilon}{12R}}(\vz-\log(w)),\vu\right\rangle \right)$. 

\end{algorithm2e}

In the following Theorem~\ref{thm:smoothing:centering_inexact_weight}
we show that the above algorithm, $\code{centeringInexact}$, achieves
precisely these goals. The algorithm consists primarily of taking
a projected Newton step, which corresponds to solving a linear system,
and a projection onto $U$ (step 5), which in Section~\ref{sec:app:Project_ball_box}
we show can be done polylogarithmic depth and nearly linear work.
Consequently, Theorem~\ref{thm:smoothing:centering_inexact_weight}
is our primary subroutine for designing efficient linear programming
algorithms in Section~\ref{sec:master_thm}.
\begin{thm}
\label{thm:smoothing:centering_inexact_weight} Assume that $K\leq\frac{1}{16c_{k}}$.
Suppose that
\[
\delta\defeq\delta_{t}(\vx,\vWeight)\leq R\enspace\text{ and }\enspace\Phi_{\mu}(\log(\vg(x))-\log(\vWeight))\leq36c_{1}c_{s}c_{k}m
\]
where $\mu=\frac{\epsilon}{12R}$ and $R=\frac{K}{48c_{k}\log(36c_{1}c_{s}c_{k}m)}$.
Let $(\next{\vx},\next{\vWeight})=\code{centeringInexact}(\vx,\vWeight,K)$,
then
\[
\delta_{t}(\next{\vx},\next{\vWeight})\leq\left(1-\frac{1}{4c_{k}}\right)\delta\enspace\text{ and }\enspace\Phi_{\mu}(\log(\vg(\next{\vx}))-\log(\next{\vWeight}))\leq36c_{1}c_{s}c_{k}m.
\]
\end{thm}

\begin{proof}
By Lemma~\ref{lem:gen:w_step}, inequality (\ref{eq:centrality_equivalence}),
$\cWeightStab(\fvWeight)\leq1+(1/16c_{k}(g))$ (Lemma \ref{lem:c_gamma})
and $\delta\leq\frac{1}{160c_{k}}$, we have
\begin{align}
\mixedNorm{\log(\vg(\next{\vx}))-\log(\vg(\vx))}{\vg(\vx)} & \leq(1-\frac{1}{c_{k}}+4c_{\gamma}\delta)\cdot c_{\gamma}\delta\leq(1-\frac{1}{c_{k}})\cdot c_{\gamma}\delta+5\delta^{2}\nonumber \\
 & \leq\left(1-\frac{13}{14c_{k}}\right)\delta.\label{eq:log_g_change}
\end{align}
Using $\Phi_{\mu}(\log(\vg(x))-\log(\vWeight))\leq36c_{1}c_{s}c_{k}m$,
the definition of $\mu$ and $R$, and $K\leq\frac{1}{16c_{k}}$,
we have
\[
\normFull{\frac{\vWeight-\vg(x)}{\vg(x)}}_{\infty}\leq\norm{\log(w)-\log(g(x))}_{\infty}+\norm{\log(w)-\log(g(x))}_{\infty}^{2}\leq\frac{17}{16}K\leq\frac{1}{14c_{k}}.
\]
Hence, we have
\[
\mixedNorm{\log(\vg(\next{\vx}))-\log(\vg(\vx))}w\leq\left(1+\frac{1}{14c_{k}}\right)\left(1-\frac{13}{14c_{k}}\right)\delta\leq\left(1-\frac{6}{7c_{k}}\right)\delta.
\]
Therefore, we know that for the Newton step, we have $\log(\vg(\next{\vx}))-\log(\vg(\vx))\in U$
where $U$ is the symmetric convex set given by $U\defeq\{\vx\in\R^{m}\,|\,\mixedNorm{\vx}{\vWeight}\leq C\}$
where $C\defeq(1-\frac{6}{7c_{k}})\delta.$ Note that from our assumption
on $\delta$, we have
\[
C\leq\delta\leq\frac{K}{48c_{k}\log(36c_{1}c_{s}c_{k}m)}=R.
\]
and therefore $U$ is contained in a $\ellInf$ ball of radius $R$.
Therefore, we can play the $\ell_{\infty}$ chasing game on $\log(g(x))$
attempting to maintain the invariant that $\normInf{\log(w)-\log(g(x))}\leq K$
without taking steps that are more than $1+\epsilon$ times the size
of $U$ where we pick $\epsilon=\frac{1}{2c_{k}}$ so to not interfere
with our ability to decrease $\delta_{t}$ linearly. To apply Theorem~\ref{thm:zerogame}
we need to ensure that $R$ satisfies $\frac{12R}{\epsilon}\log\left(\frac{12m\tau}{\epsilon}\right)\leq K$
where here $\tau$ is as defined in Theorem~\ref{thm:zerogame}.

To bound $\tau$, we need to lower bound the radius of $\ellInf$
ball it contains. Since by assumption $\|g(x)\|_{1}\leq\cWeightSize(\fvWeight)$
and $\normInf{\log(\vg(\next{\vx}))-\log(\vWeight)}\leq\frac{1}{2}$,
we have that $\|w\|_{1}\leq2c_{1}(g)$. Hence, we have $\norm{\vu}_{\infty}^{2}\geq\frac{1}{2c_{1}(g)}\norm{\vu}_{\vWeight}^{2}$
for all $u\in\Rm$ and consequently, if $\norm{\vu}_{\infty}\leq\frac{\delta}{4\cnorm\sqrt{c_{1}(g)}}$,
then $\vu\in U$. Therefore, $U$ contains a box of radius $\frac{\delta}{4\cnorm\sqrt{c_{1}(g)}}$
and since $U$ is contained in a box of radius $\delta$, we have
that
\begin{eqnarray*}
\tau & \leq & 4\cnorm\sqrt{c_{1}}\leq96\sqrt{c_{1}}\sqrt{c_{s}}c_{k}.
\end{eqnarray*}
where we used the fact that $\cnorm=24\sqrt{c_{s}}c_{k}$. Using that
$\epsilon=\frac{1}{2c_{k}}$, we have that
\[
\frac{12R}{\epsilon}\log\left(\frac{12m\tau}{\epsilon}\right)\leq48Rc_{k}\log(36c_{1}c_{s}c_{k}m)=K\text{ and }\frac{12m\tau}{\epsilon}\leq36c_{1}c_{s}c_{k}m\,.
\]
This proves that we meet the conditions of Theorem \ref{thm:zerogame}.
Consequently, $\normInf{\log(\vg(\next{\vx}))-\log(\next w)}\leq K$
and $\Phi_{\mu}\leq36c_{1}c_{s}c_{k}m$. 

Since $K\leq\frac{1}{4}$, Lemma \ref{lem:gen:x_progress} and $\delta\leq\frac{1}{160c_{k}}$
shows that 
\begin{equation}
\delta_{t}(\next{\vx},\vWeight)\leq4\delta_{t}(\vx,\vWeight)^{2}\leq\frac{\delta}{40c_{k}}.\label{eq:smoothing:centering_inexact_weight1}
\end{equation}
Step 5 shows that
\[
\mixedNorm{\log(\vWeight)-\log(\next{\vWeight})}{\vWeight}\leq\left(1+\frac{1}{2c_{k}}\right)\left(1-\frac{6}{7c_{k}}\right)\delta\leq\left(1-\frac{5}{14c_{k}}\right)\delta.
\]
Using this and (\ref{eq:smoothing:centering_inexact_weight1}), Lemma
\ref{lem:gen:w_change} shows that
\begin{align*}
\delta_{t}(\next{\vx},\next{\vWeight}) & \leq\left(1+4\left(1-\frac{5}{14c_{k}}\right)\delta\right)\left(\delta_{t}(\next{\vx},\vWeight)+\left(1-\frac{5}{14c_{k}}\right)\delta\right)\\
 & \leq\left(1+\frac{1}{40c_{k}}\right)\left(\frac{\delta}{40c_{k}}+\left(1-\frac{5}{14c_{k}}\right)\delta\right)\\
 & \leq\left(1+\frac{1}{40c_{k}}-\left(\frac{5}{14c_{k}}-\frac{1}{40c_{k}}\right)\right)\delta\leq\left(1-\frac{1}{4c_{k}}\right)\delta.
\end{align*}
\end{proof}

\section{Lewis Weights\label{sec:lewis_weights}}

As discussed in Section~\ref{sec:geom_motiv} and Section~\ref{subsec:Path-Finding},
Lewis weights play a key role in designing a weight function needed
to obtain our fastest linear programming algorithms and designing
an efficiently computable self-concordant barrier. Formally, Lewis
weights are defined as follows.
\begin{defn}[Lewis Weight]
\label{def:lewis_weight} For all $p>0$ and non-degenerate\footnote{The non-degeneracy assumptions on $\ma$ are mild and made primarily
for notational convenience. If $\ma$ has a zero row its corresponding
Lewis weight is defined to be $0$ and if $\ma$ is not full rank,
much of the definitions and analysis still apply where inverses and
determinants are replaced with pseudoinverse and pseudodeterminants
respectively. } $\ma\in\Rmn$ we define the \emph{$\ell_{p}$ Lewis weight} $w^{(p)}(\ma)$
as the unique vector $w\in\Rpm$ such that $w=\sigma(\mw^{\frac{1}{2}-\frac{1}{p}}\ma)$
where $\mw=\mDiag(w)$.
\end{defn}

In this section we provide facts about Lewis weights that we use throughout
the paper. First in Section~\ref{subsec:lewis_weight_convex} we
show how Lewis weights can be written as the minimizer of a convex
problem for all $p>0$. This convex formulation is a critical for
our self-concordant barrier construction. Then in Section~\ref{subsec:lewis_weight_convex}
we provide facts about the stability of Lewis weights which are critical
for analyzing the performance of this barrier. Further, in Section~\ref{sec:lewis_rounding}
we show that Lewis weights yield ellipsoids that are provably good
approximations to the polytope $\Omega=\{x\in\R^{n}\,|\,\norm{\ma x}_{\infty}\leq1\}$.
Finally, in Section~\ref{sec:weight-function} we use this analysis
to show that Lewis weights yield a weight function in the context
of Definition~\ref{def:gen:weight_function}. In Appendix~\ref{sec:extreme_lewis_weights}
we shed further light on Lewis weights and the algorithms we build
with them showing that they interpolate between the natural uniform
distribution over the rows of the matrix (i.e. $p\rightarrow0$) and
the weights that yield a John ellipse of $\Omega$ (i.e. $p\rightarrow\infty$). 

\subsection{Convex Formulation of Lewis Weights}

\label{subsec:lewis_weight_convex}Here we show that Lewis weights
are the result of solving a particular convex optimization problem
for $p>0$ with $p\neq2$. This convex formulation relies on the following
potential. 
\begin{defn}[Volumetric Potential]
 For non-degenerate $\ma\in\R^{m\times n}$ and $p>0$ with $p\neq2$
we define the \emph{volumetric potential }as
\[
\mathcal{V}_{p}^{\ma}(w)\defeq-\frac{1}{1-\frac{2}{p}}\log\det\left(\ma^{\top}\mw^{1-\frac{2}{p}}\ma\right).
\]
We will omit $\ma$ and $p$ when they are clear from context.
\end{defn}

The main result of this section is the following lemma, claiming that
Lewis weights are the unique solution to $\min_{w_{i}\geq0}\mathcal{V}_{p}^{\ma}(w)+\sum_{i\in[m]}w_{i}$,
or equivalently, $\min_{w_{i}\geq0,\sum_{i\in[m]}w_{i}=n}\mathcal{V}_{p}^{\ma}(w)$. 
\begin{lem}
\label{lem:unique_lewis} For all non-degenerate $\ma\in\Rmn$ its
$\ell_{p}$ Lewis weights exist and are unique for $p>0$. For $p\neq2$
the weights $\lpweight(\ma)$ are the unique minimizer of the following
equivalent convex problems: 
\[
\min_{w\in\R_{>0}^{m}}\mathcal{V}_{p}^{\ma}(w)+\sum_{i\in[m]}w_{i}\text{ \,and\, }\min_{w\in\R_{>0}^{m}:\sum_{i\in[m]}w_{i}=n}\mathcal{V}_{p}^{\ma}(w)\,.
\]
\end{lem}

To prove Lemma~\ref{lem:unique_lewis} we first compute and bound
the gradient and Hessian of $\volPot_{p}^{\ma}(w)$.
\begin{lem}[Gradient and Hessian of Volumetric Potential]
\label{lem:lewis_potential_derivatives} For all non-degenerate $\ma\in\Rmn$,
$w\in\Rpm$, and $p>0$ with $p\neq2$ we have
\[
\nabla\mathcal{V}_{p}^{\ma}(w)=-\mw^{-1}\sigma_{w}\enspace\text{ and }\enspace\nabla^{2}\mathcal{V}_{p}^{\ma}(w)=\mw^{-1}\left(\mSigma_{w}-\left(1-\frac{2}{p}\right)\mLambda_{w}\right)\mw^{-1}
\]
where $\mw\defeq\mDiag(w)$, $\sigma_{w}\defeq\sigma(\mw^{\frac{1}{2}-\frac{1}{p}}\ma)$,
$\mSigma_{w}\defeq\mSigma(\mw^{\frac{1}{2}-\frac{1}{p}}\ma)$, and
$\mLambda_{w}\defeq\mProj(\mw^{\frac{1}{2}-\frac{1}{p}}\ma$). Consequently,
$\mathcal{V}_{p}^{\ma}$ is convex in $w$ and
\begin{equation}
\frac{2}{\max\{p,2\}}\cdot\mw^{-1}\mSigma_{w}\mw^{-1}\preceq\nabla^{2}\mathcal{V}_{p}^{\ma}(w)\preceq\frac{2}{\min\{p,2\}}\cdot\mw^{-1}\mSigma_{w}\mw^{-1}.\label{eq:hess_lower}
\end{equation}
\end{lem}

\begin{proof}
The formulas for $\nabla\mathcal{V}_{p}^{\ma}(w)$ and $\nabla^{2}\mathcal{V}_{p}^{\ma}(w)$
follow from a more general Lemma~\ref{lem:all_volumetric_derivs}.
Now, $\mZero\preceq\mLambda_{w}\preceq\mSigma_{w}$ by Lemma~\ref{lem:tool:projection_matrices};
consequently $\mWeight^{-1}\mSigma_{w}\mw^{-1}\preceq\nabla^{2}\mathcal{V}_{p}^{\ma}(w)\preceq\frac{2}{p}\mWeight^{-1}\mSigma_{w}\mw^{-1}$
when $p<2$ and $\frac{2}{p}\mWeight^{-1}\mSigma_{w}\mw^{-1}\preceq\nabla^{2}\mathcal{V}_{p}^{\ma}(w)\preceq\mWeight^{-1}\mSigma_{w}\mw^{-1}$
when $p>2$. In either case (\ref{eq:hess_lower}) holds.
\end{proof}
Using Lemma~\ref{lem:lewis_potential_derivatives} we prove Lemma~\ref{lem:unique_lewis},
our main result of this section.
\begin{proof}[Proof of Lemma~\ref{lem:unique_lewis}]
 For all $w\in\R_{>0}^{m}$ let $f(w)=\mathcal{V}_{p}^{\ma}(w)+\sum_{i=1}^{m}w_{i}$
and $\sigma_{w}\defeq\sigma(\mw^{\frac{1}{2}-\frac{1}{p}}\ma)$ where
$\mw=\mDiag(w)$. Lemma~\ref{lem:lewis_potential_derivatives} and
$[\sigma_{w}]_{i}\leq1$ (Lemma~\ref{lem:tool:projection_matrices})
yields that for all $w\in\R_{>0}^{m}$ if $w_{i}>1$ then
\[
\frac{\partial f(w)}{\partial w_{i}}=1-\frac{[\sigma_{w}]_{i}}{w_{i}}\geq1-\frac{1}{w_{i}}>0\,.
\]
Hence, we have $\inf_{w_{i}>0}f(w)=\inf_{1>w_{i}>0}f(w)$.

Now, if $p>2$ and $w_{i}\in[0,1]$ for all $i\in[m]$ then since
$1-\frac{2}{p}>0$
\begin{align}
[\sigma_{w}]_{i} & =\sigma\left(\mw^{\frac{1}{2}-\frac{1}{p}}\ma\right)_{i}=w_{i}^{1-\frac{2}{p}}\left[\ma(\ma^{\top}\mw^{1-\frac{2}{p}}\ma)^{-1}\ma^{\top}\right]_{ii}\nonumber \\
 & \geq w_{i}^{1-\frac{2}{p}}\left[\ma(\ma^{\top}\ma)^{-1}\ma^{\top}\right]_{ii}=w_{i}^{1-\frac{2}{p}}\sigma(\ma)_{i}\,.\label{eq:lever_bound}
\end{align}
Since $\ma$ is non-degenerate, $\sigma(\ma)_{i}\in(0,1]$ for all
$i$. Therefore for any $j$ with $w_{j}<\sigma(\ma)_{j}^{p/2}$,
we have
\[
\frac{\partial f(w)}{\partial w_{i}}=1-\frac{[\sigma_{w}]_{j}}{w_{j}}\leq1-w_{j}^{-\frac{2}{p}}\sigma(\ma)_{j}<0
\]
and consequently, $\inf_{w_{i}>0}f(w)=\inf_{1>w_{i}\geq0}f(w)=\inf_{1>w_{i}>\sigma_{i}^{p/2}}f(w)$.

Similarly, if $p<2$, $w_{i}\in[0,1]$ for all $i\in[m]$, and $w_{\min}=\min_{i\in[m]}w_{i}$
then since $1-\frac{2}{p}<0$ we have $\mw^{1-\frac{2}{p}}\preceq w_{\min}^{1-\frac{2}{p}}\mi$.
Consequently, by analogous derivation to (\ref{eq:lever_bound}) we
have $[\sigma_{w}]_{i}\geq\left(w_{i}/w_{\min}\right)^{1-(2/p)}\sigma(\ma)_{i}$.
Consequently, if $j\in\argmin_{i\in[m]}w_{i}$ this implies that $[\sigma_{w}]_{j}\geq\sigma(\ma)_{j}$
and therefore if $w_{j}<\sigma(\ma)_{j}$ we have $\frac{\partial f(w)}{\partial w_{j}}<0$.
Therefore, if we let $\sigma_{\min}=\min_{i\in[m]}\sigma_{i}>0$ (since
$\ma$ is non-degenerate) we have $\inf_{w_{i}>0}f(w)=\inf_{1>w_{i}\geq\sigma(\ma)_{i}}f(w)$.

In either case, since $f$ is continuous the above reasoning argues
that $f$ achieves its minimum on the interior of the domain and therefore
we have that the minimizer of $w_{*}$ of $f(w)$ satisfies $\grad f(w_{*})=\vzero$,
i.e. $[w_{*}]_{i}=[\sigma_{w}]_{i}$ for all $i\in[n].$ This proves
that the minimizer $f(w)$ on $w\in\R_{>0}^{m}$ exists and are Lewis
weights. Further, Lemma~\ref{lem:lewis_potential_derivatives} shows
that
\[
\nabla^{2}f(w)\succeq\frac{2}{\max(p,2)}\cdot\mw^{-1}\mSigma_{w}\mw^{-1}\succ\mZero
\]
for all $w>0$ and therefore $f$ is strictly convex where $1\geq w_{i}\geq\min(\sigma_{i},\sigma_{i}^{p/2})$
for all $i$. Consequently, the minimizer of $f$ is unique and it
is the unique point satisfying $\grad f(w)=\vzero$ for $w\in\R_{>0}^{n}$.
Further, since $\sum_{i=1}^{m}[\sigma_{w}]_{i}{}_{i}=n$ by Lemma~\ref{lem:tool:projection_matrices}
we have $\sum_{i=1}^{m}w_{p}(\ma)_{i}=n$ and we have the desired
equivalence of the two given objective functions.
\end{proof}

\subsection{Stability of Lewis Weight Under Rescaling\label{subsec:lewis_weight_rescale}}

Here we study the sensitivity of Lewis weight under rescaling $\ma$.
In Lemma~\ref{lem:LS_weights_de} we compute the Jacobian of Lewis
weights with respect to rescaling and in Lemma~\ref{lem:LS_weights_stable}
we bound it.
\begin{lem}
\label{lem:LS_weights_de} For all non-degenerate $\ma\in\R^{m\times n}$,
$p>0$ with $p\neq2$, and $v\in\R_{.}^{m}$, let $w(v)\defeq\lpweight(\mv\ma)$
where $\mv=\mDiag(v)$. Then, for $\mLambda_{v}\defeq\mLapProj\left(\mw^{\frac{1}{2}-\frac{1}{p}}\mv\ma\right)$
and $\mw_{v}\defeq\mDiag(w(v))$ we have
\[
\mj_{w}(v)=2\mw_{v}\left(\mw_{v}-\left(1-\frac{2}{p}\right)\mLambda_{v}\right)^{-1}\mLambda_{v}\mv^{-1}\,.
\]
\end{lem}

\begin{proof}
Let $f(v,w)\defeq\frac{-1}{1-\frac{2}{p}}\log\det(\ma^{\top}\mv\mw^{1-\frac{2}{p}}\mv\ma)+\sum_{i=1}^{m}w_{i}$.
Lemma~\ref{lem:unique_lewis} shows that
\[
w(v)=\argmin_{w\in\R_{>0}^{m}}f(v,w)
\]
and that the optimal is in the interior. Hence, the optimality conditions
yield $\nabla_{w}f(v,w(v))=0$. Taking derivative with respect to
$v$ on both sides, we have that
\[
\nabla_{wv}^{2}f(v,w(v))+\nabla_{ww}^{2}f(v,w(v))\mj_{w}(v)=0.
\]
Therefore, we have that
\begin{equation}
\mj_{w}(v)=-(\nabla_{ww}^{2}f(v,w(v)))^{-1}\nabla_{wv}^{2}f(v,w(v)).\label{eq:lewis_change_1}
\end{equation}
Lemma~\ref{lem:lewis_potential_derivatives} and Lemma~\ref{lem:all_volumetric_derivs}
yield that 
\begin{align}
\nabla_{ww}^{2}f(v,w) & =\mw^{-1}\left(\mSigma_{w}-\left(1-\frac{2}{p}\right)\mLambda_{w}\right)\mw^{-1}.\label{eq:lewis_change_2}
\end{align}

For $\nabla_{wv}^{2}f(v,w(v))$, we note that $\nabla_{w}f(v,w)=-\mw^{-1}\sigma(\mw^{\frac{1}{2}-\frac{1}{p}}\mv\ma)$.
Taking derivative with respect to $v$ and using Lemma~\ref{lem:deriv:proj}
gives that
\begin{equation}
[\nabla_{wv}^{2}f(v,w)]_{ij}=-\frac{2}{w_{i}}\mLambda[\mw^{\frac{1}{2}-\frac{1}{p}}\mv\ma]_{ij}\cdot v_{j}^{-1}.\label{eq:lewis_change_3}
\end{equation}
Combining (\ref{eq:lewis_change_1}), (\ref{eq:lewis_change_2}),
and (\ref{eq:lewis_change_3}), we have
\begin{align*}
\mj_{w}(v) & =\mw_{v}\left(\mSigma_{w}-\left(1-\frac{2}{p}\right)\mLambda_{w}\right)^{-1}\mw_{v}\cdot2\mw_{v}^{-1}\mLambda_{v}\mv^{-1}\,.
\end{align*}
The result follows from $w(v)=\sigma(\mw^{\frac{1}{2}-\frac{1}{p}}\mv\ma)$
by Lemma~\ref{lem:unique_lewis}.
\end{proof}
\begin{lem}
\label{lem:LS_weights_stable} Under the setting of Lemma~\ref{lem:LS_weights_de}
for any $v\in\Rpm$ and $h\in\Rm$, we have that
\begin{align}
\norm{\mw_{v}^{-1}\mj_{w}(v)h}_{w(v)} & \leq p\cdot\normFull{\mv^{-1}h}_{w(v)},\label{eq:weight_change_w_norm}
\end{align}
and that 
\begin{equation}
\norm{[\mw_{v}^{-1}\mj_{w}(v)-p\mv^{-1}]h}_{\infty}\leq p\cdot\max\left\{ \frac{p}{2},1\right\} \cdot\norm{\mv^{-1}h}_{w(v)}\,.\label{eq:weight_change_infty_norm}
\end{equation}
\end{lem}

\begin{proof}
Fix an arbitrary $v\in\Rpm$ and $h\in\R^{m}$ and let $w\defeq w(v)$,
$\mSigma\defeq\mSigma(\mw^{\frac{1}{2}-\frac{1}{p}}\mv\ma)$, $\mLapProj\defeq\mLambda(\mw^{\frac{1}{2}-\frac{1}{p}}\mv\ma)$,
$\mNormProjLap\defeq\mNormProjLap(\mw^{\frac{1}{2}-\frac{1}{p}}\mv\ma)$,
and $\mProj^{(2)}=\mProj^{(2)}(\mw^{\frac{1}{2}-\frac{1}{p}}\mv\ma)$.
By Lemma~\ref{lem:LS_weights_de} and the fact that $w\defeq\lpweight(\mv\ma)=\sigma(\mw^{\frac{1}{2}-\frac{1}{p}}\mv\ma)\in\R_{>0}^{m}$,
we have that 
\begin{align}
\mj_{w}(v)h & =2\mSigma\left(\mSigma-\left(1-\frac{2}{p}\right)\mLambda\right)^{-1}\mLambda\mv^{-1}h=2\mSigma^{1/2}\overline{\mLambda}\left(\mi-\left(1-\frac{2}{p}\right)\overline{\mLambda}\right)^{-1}\mw^{1/2}\mv^{-1}h\label{eq:jacobian_expression}
\end{align}
Consequently (\ref{eq:weight_change_w_norm}) follows from
\[
\norm{\mw_{v}^{-1}\mj_{w}(v)h}_{w}=\norm{\mSigma^{-1/2}z}_{2}=2\normFull{\overline{\mLambda}\left(\mi-\left(1-\frac{2}{p}\right)\overline{\mLambda}\right)^{-1}\mw^{1/2}\mv^{-1}h}_{2}\leq p\norm{\mw^{1/2}\mv^{-1}h}_{2}
\]
where in the last step we used that $\overline{\mLambda}(\iMatrix-(1-\frac{2}{p})\overline{\mLambda})^{-1}$
is a symmetric matrix whose eigenvalues are of the form $\lambda/(1-(1-\frac{2}{p})\lambda)$
for each eigenvalue of $\lambda$ of $\overline{\mLambda}$ and since
$\mZero\preceq\bar{\mLambda}\preceq\mi$
\begin{equation}
\normFull{\overline{\mLambda}\left(\mi-\left(1-\frac{2}{p}\right)\overline{\mLambda}\right)^{-1}}_{2}\leq\max_{0\leq\lambda\leq1}\frac{\lambda}{1-(1-\frac{2}{p})\lambda}=\frac{p}{2}\,.\label{eq:norm_inv_easy_bound}
\end{equation}

Next, note that $\mLapProj=\mLever-\mProj^{(2)}$ and therefore $\mi-\mNormProjLap=\mw^{-1/2}\mProj^{(2)}\mw^{-1/2}$.
Combining this with (\ref{eq:jacobian_expression}) and that $\mi-\left(1-\frac{2}{p}\right)\overline{\mLambda}$
is invertible yields
\begin{align*}
[\mw^{-1}\mj_{w}(v)-p\mv^{-1}]h & =\mw^{-1/2}\left[2\overline{\mLambda}-p\left(\mi-\left(1-\frac{2}{p}\right)\overline{\mLambda}\right)\right]\left(\mi-\left(1-\frac{2}{p}\right)\overline{\mLambda}\right)^{-1}\mw^{1/2}\mv^{-1}h\\
 & =p\mw^{-1}\mProj^{(2)}\mw^{-1/2}\left(\mi-\left(1-\frac{2}{p}\right)\overline{\mLambda}\right)^{-1}\mw^{1/2}\mv^{-1}h\,.
\end{align*}
However, by Lemma~\ref{lem:tool:projection_matrices} we know that
$\norm{\mSigma^{-1}\mProj^{(2)}x}_{\infty}\leq\norm x_{\mSigma}=\norm{\mSigma^{1/2}x}_{2}$
for all $x$ and therefore (\ref{eq:weight_change_infty_norm}) follows
from
\begin{align*}
\norm{[\mw^{-1}\mj_{w}(v)-p\mv^{-1}]h}_{\infty} & \leq p\normFull{\left(\mi-\left(1-\frac{2}{p}\right)\overline{\mLambda}\right)^{-1}\mw^{1/2}\mv^{-1}h}_{2}\leq p\cdot\max\left\{ 1,\frac{p}{2}\right\} \cdot\norm{\mv^{-1}h}_{w}\,.
\end{align*}
To prove the inequality in the last step above, note that $(\iMatrix-(1-\frac{2}{p})\overline{\mLambda})^{-1}$
is a symmetric matrix whose eigenvalues are of the form $1/(1-(1-\frac{2}{p})\lambda)$
for each eigenvalue of $\lambda$ of $\overline{\mLambda}$ and since
$\mZero\preceq\bar{\mLambda}\preceq\mi$
\begin{equation}
\normFull{\left(\mi-\left(1-\frac{2}{p}\right)\overline{\mLambda}\right)^{-1}}_{2}\leq\max_{0\leq\lambda\leq1}\frac{1}{1-(1-\frac{2}{p})\lambda}=\max\left\{ 1,\frac{p}{2}\right\} \,.\label{eq:norm_inv_bound}
\end{equation}
\end{proof}

\subsection{Lewis Weight Rounding Properties\label{sec:lewis_rounding}}

Here we show that the Lewis weights for a matrix $\ma\in\R^{m\times n}$
provide ellipses that provably approximate the polytope $\Omega=\{x\in\R^{n}|\norm{\ma x}_{\infty}\leq1\}$.
This bound helps analyze both the self-concordance of the Lewis weight
barrier and the efficacy as Lewis weights for the weighted central
path. The main result of this section is the more general Lemma~\ref{lem:lewis_p_change}
which relates every $\ell_{p}$ Lewis weight to $\ell_{r}$ Lewis
weight for $r\geq p$. In particular, it bounds how well $\ell_{p}$
lewis weights $w$ satisfy the optimality conditions of being a $\ell_{r}$
Lewis weight, i.e. the size of $\sigma(\mw^{\frac{1}{2}-\frac{1}{r}}\ma)_{i}w_{i}^{-1}$.
\begin{lem}
\label{lem:lewis_p_change} For all non-degenerate $\ma\in\Rmn$ and
$w\defeq\lpweight(\ma)$ for $0<p<r$ we have
\[
\sigma(\mw^{\frac{1}{2}-\frac{1}{r}}\ma)_{i}w_{i}^{-1}\leq c_{p,r,m}\defeq\left(1+\alpha\right)^{\frac{1}{1+\alpha}}\left(\left(1+\frac{1}{\alpha}\right)m\right)^{\frac{\alpha}{1+\alpha}}\leq2m^{\frac{\alpha}{1+\alpha}}\text{ where }\alpha=\frac{2}{p}-\frac{2}{r}\,.
\]
\end{lem}

Lemma~\ref{lem:lewis_p_change} shows that $\ell_{p}$ Lewis weights
for large $p$ yield an ellipse that well approximates $\Omega$ though
the following simple lemma.
\begin{lem}
\label{lem:lewis_rounding} For non-degenerate $\ma\in\Rmn$, $w\defeq\lpweight(\ma)$
for $p>0$, and $\mw\defeq\mDiag(w)$, define 
\[
E\defeq\{x\in\Rn:\ x^{\top}\ma^{\top}\mw\ma x\leq1\}\text{ and }K\defeq\{x\in\Rn:\norm{\ma x}_{\infty}\leq1\}\,.
\]
Then $R\defeq\max_{i\in[m]}\sigma(\mw^{\frac{1}{2}}\ma)_{i}w_{i}^{-1}$
is the smallest value such that $E\subseteq\sqrt{R}K.$ Further, $K$
is a $\sqrt{c_{p,\infty,m}n}$-rounding (See Lemma~\ref{lem:lewis_p_change})
as $K\subset\sqrt{n}E$. 
\end{lem}

\begin{proof}
Note that $\sigma(\mw^{\frac{1}{2}}\ma)_{i}w_{i}^{-1}=\cordVec i^{\top}\ma(\ma^{\top}\mw\ma)^{-1}\ma^{\top}\cordVec i$
and 
\begin{align*}
\max_{x\in E}\norm{\ma x}_{\infty}^{2} & =\max_{i\in[m],x\in E}\left(\cordVec i^{\top}\ma(\ma^{\top}\mw\ma)^{-\frac{1}{2}}(\ma^{\top}\mw\ma)^{\frac{1}{2}}x\right)^{2}=\max_{i\in[m]}\cordVec i^{\top}\ma(\ma^{\top}\mw\ma)^{-1}\ma^{\top}\cordVec i=R.
\end{align*}
Consequently, $R$ is as desired. Further, for any $x\in K$ we have
$x^{\top}\ma^{\top}\mw\ma x\leq\sum_{i\in[m]}w_{i}\leq n$ and therefore
$K\subset\sqrt{n}E$.
\end{proof}
Note that in Lemma~\ref{lem:lewis_p_change} we have $\lim_{p\rightarrow\infty}c_{p,\infty,m}=1$.
This suggest that that as $p\rightarrow\infty$ the ellipse $E$ behaves
like a John ellipse for $\Omega$. Indeed, in Appendix~\ref{sec:pinf}
we show that $E$ converges to the John ellipse of $\Omega$.

To prove Lemma~\ref{lem:lewis_p_change} we first provide the following
helper lemma analyzing the effect of changing the power to which we
might raise $\mw$ in $\ma^{\top}\mw\ma$.
\begin{lem}
\label{lem:AWAASigmaA} Let $\ma\in\Rmn$, $p>0$ and $w=\lpweight(\ma)$.
Then, for any $r\geq p$, we have that
\begin{align*}
\ma^{\top}\mw^{1-\frac{2}{r}}\ma\preceq\ma^{\top}\mw^{1-\frac{2}{p}}\ma & \preceq(1+\alpha)\left(\left(1+\frac{1}{\alpha}\right)m\right)^{\alpha}\ma^{\top}\mw^{1-\frac{2}{r}}\ma\text{ where }\alpha=\frac{2}{p}-\frac{2}{r}\,.
\end{align*}
\end{lem}

\begin{proof}
Since $r\geq p$ and $w_{i}\in(0,1]$ for all $i\in[m]$ we have that
$w_{i}^{1-\frac{2}{r}}\leq w_{i}^{1-\frac{2}{p}}$ for all $i\in[m]$
and therefore $\ma^{\top}\mw^{1-\frac{2}{r}}\ma\preceq\ma^{\top}\mw^{1-\frac{2}{p}}\ma$. 

To prove the other direction, let $\epsilon>0$ be arbitrary and let
$\mi_{w>\frac{\epsilon}{m}}\in\R^{m\times m}$ be the diagonal matrix
where $[\mi_{\epsilon}]_{ii}=1$ if $w_{i}>\frac{\epsilon}{m}$ and
$[\mi_{\epsilon}]_{ii}=0$ otherwise and let $\mi_{w\leq\frac{\epsilon}{m}}=\mi-\mi_{w>\frac{\epsilon}{m}}$.
Note that 
\[
\tr\left[(\ma^{\top}\mw^{1-\frac{2}{p}}\ma)^{-1}\ma^{\top}\mw^{1-\frac{2}{p}}\mi_{w\leq\frac{\varepsilon}{m}}\ma\right]=\sum_{i\in[m]:\ w_{i}\leq\frac{\varepsilon}{m}}w_{i}\leq m\cdot\frac{\varepsilon}{m}=\varepsilon
\]
where we used that $w=\sigma(\mw^{\frac{1}{2}-\frac{1}{p}}\ma)$.
Therefore, $\ma^{\top}\mw^{1-\frac{2}{p}}\mi_{w\leq\frac{\varepsilon}{m}}\ma\preceq\varepsilon\cdot\ma^{\top}\mw^{1-\frac{2}{p}}\ma$
and hence
\begin{equation}
\ma^{\top}\mw^{1-\frac{2}{p}}\ma\preceq\frac{1}{1-\varepsilon}\ma^{\top}\mw^{1-\frac{2}{p}}\mi_{w>\frac{\varepsilon}{m}}\ma.\label{eq:AWA_and_ASigmaA1}
\end{equation}
Now, we note that 
\begin{equation}
\ma^{\top}\mw^{1-\frac{2}{p}}\mi_{w>\frac{\varepsilon}{m}}\ma\prec\left(\frac{m}{\varepsilon}\right)^{\frac{2}{p}-\frac{2}{r}}\ma^{\top}\mw^{1-\frac{2}{r}}\mi_{w>\frac{\varepsilon}{m}}\ma.\label{eq:AWA_and_ASigmaA2}
\end{equation}
Combining (\ref{eq:AWA_and_ASigmaA1}) and (\ref{eq:AWA_and_ASigmaA2}),
recalling $\alpha=\frac{2}{p}-\frac{2}{r}$ and choosing the minimizing
$\varepsilon=\frac{\alpha}{1+\alpha}$ yields
\begin{align*}
\ma^{\top}\mw^{1-\frac{2}{p}}\ma & \preceq\frac{1}{1-\varepsilon}\left(\frac{m}{\varepsilon}\right)^{\alpha}\ma^{\top}\mw^{1-\frac{2}{r}}\ma=\frac{(1+\alpha)^{1+\alpha}}{\alpha^{\alpha}}m^{\alpha}\ma^{\top}\mw^{1-\frac{2}{r}}\ma\,.
\end{align*}
\end{proof}
Using Lemma~\ref{lem:AWAASigmaA} we prove Lemma~\ref{lem:lewis_p_change},
the main result of this section.
\begin{proof}[Proof of Lemma~\ref{lem:lewis_p_change}]
 Note that 
\begin{equation}
\sigma(\mw^{1-\frac{2}{r}}\ma)_{i}w_{i}^{-1}=\cordVec i^{\top}\ma(\ma^{\top}\mw^{1-\frac{2}{r}}\ma)^{-1}\ma^{\top}\cordVec iw_{i}^{-\frac{2}{r}}\,.\label{eq:round_bound:0}
\end{equation}
Applying Lemma~\ref{lem:AWAASigmaA} yields that 
\[
\ma^{\top}\mw^{1-\frac{2}{r}}\ma\succeq\frac{1}{(1+\alpha)((1+\frac{1}{\alpha})m)^{\alpha}}\ma^{\top}\mw^{1-\frac{2}{p}}\ma\text{ where }\alpha=\frac{2}{p}-\frac{2}{r}\,.
\]
Consequently, for all $i\in[m]$ it follows that
\begin{align}
\cordVec i^{\top}\ma(\ma^{\top}\mw^{1-\frac{2}{r}}\ma)^{-1}\ma^{\top}\cordVec i & w_{i}^{-\frac{2}{r}}\leq\left(1+\alpha\right)\left(\left(1+\alpha\right)m\right)^{\alpha}\cordVec i^{\top}\ma(\ma^{\top}\mw^{1-\frac{2}{p}}\ma)^{-1}\ma^{\top}\cordVec iw_{i}^{-\frac{2}{r}}\,.\label{eq:round_bound:1}
\end{align}
Further, since $w=\sigma(\mw^{\frac{1}{2}-\frac{1}{p}}\ma)$ we have
that 
\begin{equation}
\cordVec i^{\top}\ma(\ma^{\top}\mw^{1-\frac{2}{p}}\ma)^{-1}\ma^{\top}\cordVec iw_{i}^{-\frac{2}{r}}=w_{i}^{\frac{2}{p}-1}\sigma_{i}(\mw^{\frac{1}{2}-\frac{1}{p}}\ma)w_{i}^{-\frac{2}{r}}=w_{i}^{\alpha}\label{eq:round_bound:2}
\end{equation}
Additionally, since $w_{i}^{1-\frac{2}{r}}\ma^{\top}\cordVec i\cordVec i^{\top}\ma\preceq\ma^{\top}\mw^{1-\frac{2}{r}}\ma$
we have $w_{i}^{1-\frac{2}{r}}\cordVec i^{\top}\ma(\ma^{\top}\mw^{1-\frac{2}{r}}\ma)^{-1}\ma^{\top}\cordVec i\leq1$
and 
\begin{equation}
\cordVec i^{\top}\ma(\ma^{\top}\mw^{1-\frac{2}{r}}\ma)^{-1}\ma^{\top}\cordVec iw_{i}^{-\frac{2}{r}}\leq w_{i}^{-1}\,.\label{eq:round_bound:3}
\end{equation}
Combining (\ref{eq:round_bound:0}), (\ref{eq:round_bound:1}), (\ref{eq:round_bound:2}),
and (\ref{eq:round_bound:3}) yields
\begin{align*}
\sigma(\mw^{1-\frac{2}{r}}\ma)_{i}w_{i}^{-1} & \leq\min\left\{ \left(1+\alpha\right)\left(\left(1+\frac{1}{\alpha}\right)m\right)^{\alpha}w_{i}^{\alpha},w_{i}^{-1}\right\} \leq\left(1+\alpha\right)^{\frac{1}{1+\alpha}}\left(\left(1+\frac{1}{\alpha}\right)m\right)^{\frac{\alpha}{1+\alpha}}\,.
\end{align*}
where we used that $\min\{ax^{b},x^{c}\}\leq a^{\frac{-c}{b-c}}$
for $a\geq1$, $b\geq c$, and $x\in[0,1]$. The final inequality
follows from the fact that if we let $f(\alpha)\defeq(1+\alpha)^{\frac{1}{1+\alpha}}(1+\frac{1}{\alpha})^{\frac{\alpha}{1+\alpha}}$
then $f(\alpha)\leq2$ for all $\alpha\geq0$ as the concavity of
log shows that
\[
\log f(\alpha)=\frac{\log(1+\alpha)}{1+\alpha}+\frac{\alpha\cdot\log(1+(1/\alpha))}{1+\alpha}\leq\log\left(\frac{(1+\alpha)}{1+\alpha}+\frac{\alpha\cdot1+(1/\alpha)}{1+\alpha}\right)=\log2\,.
\]
\end{proof}

\subsection{Weight Function\label{sec:weight-function}}

Here, we show that regularized Lewis weights for suitable $p$ and
small enough regularization are a valid weight function (Definition~\ref{def:gen:weight_function}).
The main result of this section is the following theorem which bounds
the weight function parameters. This theorem follows almost immediately
from the calculations in Section~\ref{sec:lewis_weights}. To obtain
our fastest algorithms for linear programming, we choose $p=1-1/\log(4m)$
and $c_{0}=n/(2m)$.
\begin{thm}
\label{thm:max_flow:weight_properties} For any $p\in(0,1)$ and any
$c_{0}\geq0$, the weight function $\vg:\dInterior\rightarrow\Rpm$
defined for all $\vx\in\Rpm$ as
\begin{equation}
\vg(\vx)\defeq\lpweight(\ma_{x})+c_{0}\quad\text{where}\quad\ma_{x}\defeq(\mPhi''(\vx))^{-1/2}\ma.\label{eq:sec:weights:weight_function}
\end{equation}
is a weight function in the context of Definition~\ref{def:gen:weight_function}
and satisfies $c_{1}(g)\leq n+c_{0}m$, $c_{s}(\fvWeight)\leq2m^{1-p}$,
and $c_{k}(\fvWeight)\leq\frac{2}{1-p}$. Further, for $p=1-\frac{1}{\log(4m)}$
and $c_{0}=\frac{n}{2m}$, we have $c_{1}(g)\leq\frac{3}{2}n$, $c_{s}(\fvWeight)\leq4$,
and $c_{k}(\fvWeight)\leq2\log(4m)$.
\end{thm}

\begin{proof}
To bound the size, $c_{1}(g)$, recall that $\lpweight(\ma_{x})=\sigma(\mw^{\frac{1}{2}-\frac{1}{p}}\ma_{x})$
and therefore Lemma~\ref{lem:tool:projection_matrices} implies $\sum_{i\in[m]}\lpweight(\ma_{x})_{i}=n$.
To bound the sensitivity, $c_{s}(\fvWeight)$, note that Lemma~\ref{lem:lewis_p_change}
and that $p\leq1$ yield
\[
\cordVec i^{\top}\mg(x)^{-1}\ma_{x}\left(\ma_{x}^{\top}\mg(x)^{-1}\ma_{x}\right)^{-1}\ma_{x}^{\top}\mg(x)^{-1}\cordVec i=g(x)_{i}^{-1}\sigma(\mg(x)^{\frac{1}{2}-\frac{1}{1}}\ma_{x})_{i}\leq2m^{\frac{\alpha}{1+\alpha}}
\]
where $\alpha=\frac{2}{p}-\frac{2}{1}=\frac{2}{p}(1-p)$. As $\frac{\alpha}{1+\alpha}=\frac{2-2p}{2+p}\leq1-p$,
the bound on $c_{s}(\fvWeight)$ follows. 

To bound the consistency, $c_{k}(\fvWeight)$, note that for arbitrary
$h\in\R^{m}$ and $w(v)=\lpweight(\mv\ma)$ we have
\begin{equation}
\mg(x)^{-1}\mj_{g}(x)\left(\mPhi''(x)\right)^{-\frac{1}{2}}h=\mg(x)^{-1}\mj_{w}((\mPhi''(x))^{-\frac{1}{2}})z\label{eq:weight_func_1}
\end{equation}
where $z=-\frac{1}{2}(\mPhi''(x))^{-2}\mPhi'''(x)h$. By Lemma \ref{lem:LS_weights_stable},
we have that
\begin{equation}
\normFull{\mg(x)^{-1}\mj_{w}((\mPhi''(x))^{-\frac{1}{2}})z}_{g(x)}\leq p\normFull{(\mPhi''(x))^{\frac{1}{2}}z}_{g(x)}\label{eq:weight_func_2}
\end{equation}
and
\begin{equation}
\normFull{\mg(x)^{-1}\mj_{w}((\mPhi''(x))^{-\frac{1}{2}})z}_{\infty}\leq p\normFull{(\mPhi''(x))^{\frac{1}{2}}z}_{\infty}+p\normFull{(\mPhi''(x))^{\frac{1}{2}}z}_{g(x)}.\label{eq:weight_func_3}
\end{equation}
Combining (\ref{eq:weight_func_1}), (\ref{eq:weight_func_2}), (\ref{eq:weight_func_3})
and using the definition $\normFull{\cdot}_{\vg(x)+\infty}\defeq\|\cdot\|_{\infty}+\cnorm\|\cdot\|_{g(x)}$
yields
\[
\normFull{\mg(x)^{-1}\mj_{g}(x)\left(\mPhi''(x)\right)^{-\frac{1}{2}}h}_{\vg(x)+\infty}\leq p\normFull{(\mPhi''(x))^{\frac{1}{2}}z}_{\infty}+p(1+\cnorm)\cdot\normFull{(\mPhi''(x))^{\frac{1}{2}}z}_{g(x)}.
\]
Note that $\left|\Phi''(x)^{\frac{1}{2}}z\right|_{i}=\frac{1}{2}\left|\Phi''(x)^{-\frac{3}{2}}\Phi'''(x)h\right|_{i}\leq\left|h_{i}\right|$
by the self-concordance of $\Phi$. Therefore, 
\[
\normFull{\mg(x)^{-1}\mj_{g}(x)\left(\mPhi''(x)\right)^{-\frac{1}{2}}h}_{\vg(x)+\infty}\leq p\normFull h_{\infty}+p(1+\cnorm)\cdot\normFull h_{g(x)}\leq p\left(1+\frac{1}{\cnorm}\right)\normFull h_{g(x)+\infty}.
\]
Recalling that $\cnorm=24\sqrt{c_{s}(g)}c_{k}(g)$ and using $c_{s}(g)\geq1$,
the bound of $c_{k}(g)=\frac{2}{1-p}$ follows from
\[
p\left(1+\frac{1}{\cnorm}\right)\leq p+\frac{1}{24c_{k}(g)}=1-\frac{2}{c_{k}(g)}+\frac{1}{24c_{k}(g)}\leq1-\frac{1}{c_{k}(g)}\,.
\]
\end{proof}

\section{A Nearly Linear Self-concordant Lewis Weight Barrier\label{sec:self-concordance}}

In this section, we construct an $\otilde(n)$-self-concordant barrier
for the set $\interior\defeq\{x\in\R^{n}\,|\,\ma x>b\}$ for non-degenerate
$\ma\in\R^{m\times n}$ and vector $b\in\R^{m}$ using $\ell_{q}$
Lewis weights.\footnote{We use $q$ throughout rather than $p$ as in Section~\ref{sec:lewis_weights}
to clearly distinguish between the different (but closely related)
functions considered in each section.} Interestingly, the central path for this barrier is the points $x^{(t)}\in\R^{n}$,
$\lambda^{(t)}\in\R_{>0}^{m}$, and $s^{(t)}\in\R_{>0}^{m}$ for $t>0$
satisfying 
\begin{align*}
\lambda_{i}^{(t)}\cdot s_{i}^{(t)} & =t\cdot w_{q}(\ma_{x^{(t)}})_{i}\text{ for all }i\in[m]\\
\ma^{\top}\lambda^{(t)} & =c,\\
\ma x^{(t)}+s^{(t)} & =b.
\end{align*}
where throughout this section we let $\ma_{x}\defeq\ms_{x}^{-1}\ma$
and $\ms_{x}=\mDiag(\ma x-b)$. For all $x\in\interior$ and $w\in\R_{>0}^{m}$
we let 
\[
f(x,w)\defeq\ln\det\left(\ma_{x}^{\top}\mw^{1-\frac{2}{q}}\ma_{x}\right)-\left(1-\frac{2}{q}\right)\sum_{i=1}^{m}w_{i}
\]
and define the barrier as 
\begin{equation}
\psi(x)\defeq\begin{cases}
\max_{w\in\R^{m}:w\geq0}\frac{1}{2}f(x,w) & \text{ if }q\geq2\\
\min_{w\in\R^{m}:w\geq0}\frac{1}{2}f(x,w) & \text{ if }q\leq2
\end{cases}\,.\label{eq:LS_barrier}
\end{equation}
Note with respect to $w$ the function $f(x,w)$ is just a scaling
of the function we used for defining Lewis weights. Here we need these
two cases to maintain that $f$ is a convex function in $x.$

Note that when $q=2$ the function $f(x,w)$ does not depend on $w$
and therefore $\psi$ is well defined. Further in this case $\psi$
is exactly the volumetric barrier function, i.e. $f(x,w)=\frac{1}{2}\ln\det(\ma_{x}^{\top}\ma_{x})$.
Further, note that as $q\rightarrow0$, Lemma~\ref{lem:l0} shows
that $w_{q}(\ma)_{i}=1$ (as long as $\ma$ is in general position).
and in this case $\psi$ is the log barrier function.

We call this the \emph{Lewis weight barrier function} as by Lemma~\ref{lem:unique_lewis}
we can equivalently write 
\[
\psi(x)=\ln\det\left(\ma_{x}^{\top}\mw_{x}^{1-\frac{2}{q}}\ma_{x}\right)\text{ where }\mw_{x}=\mDiag(w_{q}(\ma_{x}))\,.
\]
The main result of this section is the following theorem which shows
that the Lewis weight barrier is a self-concordant barrier. In particular
this theorem shows that for $q=$ $\Theta(\log(m))$ the Lewis-weight
barrier is a $O(n\log^{5}m)$-self concordant. Further, when $q=2$
this theorem recovers the fact that the volumetric barrier function
is a $O(\sqrt{m}n)$-self concordant \cite{nesterov1989self,anstreicher96}.
\begin{thm}
\label{thm:LS_barrier_sc} Let $\interior=\{x\,:\,\ma x>b\}$ denote
the interior of non-empty polytope for non-degenerate $\ma$. For
any $q>0$, $\psi:\interior\rightarrow\R$ defined in (\ref{eq:LS_barrier})
is a barrier function such that for all $x\in\interior$ and $h\in\R^{n}$,
we have
\begin{enumerate}
\item $\nabla\psi(x)^{\top}\nabla^{2}\psi(x)^{-1}\nabla\psi(x)\leq n$,
\item $D^{3}\psi[h,h,h]\leq2v_{q}\norm h_{\nabla^{2}\psi(x)}^{3/2}\text{ for }v_{q}=(q+2)^{3/2}m^{\frac{1}{q+2}}+4\max\{q,2\}^{2.5}$
\end{enumerate}
Consequently, $v_{q}^{2}\psi$ is a $nv_{q}^{2}$-self-concordant
barrier function for $\interior$.
\end{thm}

In the remainder of this section we prove Theorem~\ref{thm:LS_barrier_sc}.
Leveraging the analysis of Section~\ref{sec:lewis_weights}, this
is a straightforward but tedious calculus exercise. We split the proof
into parts. In Section~\ref{subsec:Basic-properties}, we compute
the gradient of the Hessian of $f$ and $\psi$ and prove the first
item in Theorem~\ref{thm:LS_barrier_sc} (Lemma~\ref{lem:force}).
In Section~\ref{subsec:w_Sigma_Px_stable}, we then prove the second
item of Theorem~\ref{thm:LS_barrier_sc} (Lemma~\ref{lem:stable_hess})
by bounding the stability of each component of the Hessian.

\subsection{Notation and Basic Properties of Lewis Weight Barrier \label{subsec:Basic-properties}}

For brevity, throughout the remainder of this section, we let $w_{x}=\argmax_{w\in\R_{\geq0}^{m}}\frac{1}{2}f(x,w)$
when $q\geq2$, $w_{x}=\argmin_{w\in\R_{\geq0}^{m}}\frac{1}{2}f(x,w)$
when $q\leq2$, and $w_{x}=\sigma(\ma_{x})$ when $q=2$. We will
show in Lemma~\ref{lem:p_derivatives} that $w_{x}=\lqweight(\ma_{x})$
for all $q$. Further, for all $x\in\interior$ we let
\[
\mProj_{x}\defeq\mProj(\mw_{x}^{\frac{1}{2}-\frac{1}{q}}\ma_{x})\text{, }\mProj_{x,w}^{(2)}\defeq\mProj^{(2)}(\mw_{x}^{\frac{1}{2}-\frac{1}{q}}\ma_{x})\text{, }\mNormProjLap_{x}\defeq\mNormProjLap(\mw_{x}^{\frac{1}{2}-\frac{1}{q}}\ma_{x})\text{, and }\sigma_{x}\defeq\sigma(\mw_{x}^{\frac{1}{2}-\frac{1}{q}}\ma_{x}).
\]

Leveraging this notation we compute and bound the gradient and Hessian
of $\psi$.
\begin{lem}
\label{lem:p_derivatives} For all $x\in\interior$, $\sigma_{x}=w_{x}=\lqweight(\ma_{x})$
and for $\mn_{x}\defeq2\mNormProjLap_{x}(\mi-(1-\frac{2}{q})\mNormProjLap_{x})^{-1}$
we have
\begin{equation}
\grad\psi(x)=-\ma_{x}^{\top}\sigma_{x}\text{ and }\hess\psi(x)=\ma_{x}^{\top}\mSigma_{x}^{1/2}(\mi+\mn_{x})\mSigma_{x}^{1/2}\ma_{x}\label{eq:sec_ent_barrier_hess_p}
\end{equation}
Further, $\mn_{x}$ is a symmetric matrix with $\mZero\preceq\mn_{x}\preceq q\mi$
and therefore
\begin{equation}
\ma_{x}^{\top}\mSigma_{x}\ma_{x}\preceq\nabla^{2}\psi(x)\preceq(1+q)\ma_{x}^{\top}\mSigma_{x}\ma_{x}\,.\label{eq:psi_2sigma}
\end{equation}
\end{lem}

\begin{proof}
By Lemma~\ref{lem:all_volumetric_derivs}, deferred to the appendix,
recalling that $c_{q}\defeq1-\frac{2}{q}$ we have
\begin{align*}
\nabla_{x}f(x,w) & =-2\ma_{x}^{\top}\sigma_{x,w}, & \nabla_{w}f(x,w) & =c_{q}\mw^{-1}\sigma_{x,w}-c_{q},\\
\nabla_{xx}^{2}f(x,w) & =\ma_{x}^{\top}(2\mSigma_{x,w}+4\mLambda_{x,w})\ma_{x}, & \nabla_{ww}^{2}f(x,w) & =-c_{q}\mw^{-1}(\mSigma_{x,w}-c_{q}\mLambda_{x,w})\mw^{-1},\text{ and}\\
\nabla_{xw}^{2}f(x,w) & =-2c_{q}\ma_{x}^{\top}\mLambda_{x,w}\mw^{-1}\,.
\end{align*}
where $\mProj_{x,w}\defeq\mProj(\mw^{\frac{1}{2}-\frac{1}{q}}\ma_{x})$,
$\mLambda_{x,w}\defeq\mLambda(\mw^{\frac{1}{2}-\frac{1}{q}}\ma_{x})$,
and $\sigma_{x,w}\defeq\sigma(\mw^{\frac{1}{2}-\frac{1}{q}}\ma_{x})$.

Consequently, $f(x,w)$ is concave in $w$ when $q>2$, convex in
$w$ when $q<2$ and each case the optimizer is in the interior of
the set $\{w_{i}\geq0\}$ by Lemma~\ref{lem:unique_lewis}. Further,
whenever $q\neq2$ the optimality conditions imply that $\nabla_{w}f(x,w_{x})=0$
and considering the $q=2$ case directly we see that in all cases
$\sigma_{x}=w_{x}=\lqweight(\ma_{x})$. 

For $q\neq2$ taking the derivative of $\nabla_{w}f(x,w_{x})=0$ with
respect to $x$ yields that for $w(x)\defeq w_{x}$, 
\[
\nabla_{wx}^{2}f(x,w_{x})+\nabla_{ww}^{2}f(x,w_{x})\mj_{w}(x)=0.
\]
Since $\nabla_{ww}^{2}f(x,w_{x})$ is invertible, we have that in
this case
\[
\mj_{w}(x)=-(\nabla_{ww}^{2}f(x,w_{x})){}^{-1}\nabla_{wx}^{2}f(x,w_{x})\,.
\]
Using that $\psi(x)=\frac{1}{2}f(x,w_{x})$ and taking the derivative
of $x$ on both sides, we have 
\begin{equation}
2\nabla\psi(x)=\nabla_{x}f(x,w_{x})+\mj_{w}(x)^{\top}\grad_{w}f(x,w_{x})=\nabla_{x}f(x,w_{x})=-2\ma_{x}^{\top}\sigma_{x}\label{eq:self_concord_grad}
\end{equation}
where we used that $\grad_{w}f(x,w_{x})=0$ by optimality. Next, taking
the derivative again yields that 
\begin{eqnarray*}
2\nabla^{2}\psi(x) & = & \nabla_{xx}^{2}f(x,w_{x})+\nabla_{xw}^{2}f(x,w_{x})\mj_{w}(x)\\
 & = & \nabla_{xx}^{2}f(x,w_{x})-\nabla_{xw}^{2}f(x,w_{x})\left(\nabla_{ww}^{2}f(x,w_{x})\right)^{-1}\left(\nabla_{wx}^{2}f(x,w_{x})\right)
\end{eqnarray*}
Substituting in the computed values for $\nabla_{xx}^{2}f(x,w_{x})$,
$\nabla_{xw}^{2}f(x,w_{x})$, $\nabla_{wx}^{2}f(x,w_{x})$ and using
that $\mSigma_{x}=\mw_{x}$ then yields that
\begin{equation}
\hess\psi(x)=\ma_{x}^{\top}\left(\mSigma_{x}+2\mLambda_{x}\right)\ma_{x}+2c_{q}\ma_{x}^{\top}\mLambda_{x}\left(\mSigma_{x}-c_{q}\mLambda_{x}\right)^{-1}\mLambda_{x}\ma_{x}\label{eq:self_concord_hess}
\end{equation}
Further, since when $q=2$ we have $\psi(x)=\frac{1}{2}f(x,w)$ for
any $w\in\R_{>0}^{m}$ and $c_{q}=0$ we see that (\ref{eq:self_concord_grad})
and (\ref{eq:self_concord_hess}) are correct for all $q>0$. Rearranging,
scaling, and leveraging that $\mSigma_{x}$ is PD (i.e. all leverage
scores are positive) yields 
\[
\hess\psi(x)=\ma_{x}^{\top}\mSigma_{x}^{1/2}\left(\mi+2\mNormProjLap_{x}+2c_{q}\mNormProjLap_{x}\left(\mi-c_{q}\mNormProjLap_{x}\right)^{-1}\mNormProjLap_{x}\right)\mSigma_{x}^{1/2}\ma_{x}\,.
\]
Now, note that $\mZero\preceq\mNormProjLap\preceq\mi$ and $c_{q}\in(-\infty,1)$
for $q\in(0,\infty)$ and therefore no eigenvalue of $\mNormProjLap$
has value $1/c_{q}$. Since, $x+c_{q}x^{2}(1-c_{q}x)^{-1}=x(1-c_{q}x)^{-1}$
for $x\neq1/c_{q}$ and $\mNormProjLap$ and $\mi$ trivially commute
we have that $\hess\psi(x)=\ma_{x}^{\top}\mSigma_{x}^{1/2}(\mi+\mn_{x})\mSigma_{x}^{1/2}\ma_{x}$
as desired. Further, this implies that $\mn_{x}$ is symmetric with
all eigenvalues in the range $[0,q]$ (see e.g. (\ref{eq:norm_inv_easy_bound})),
proving (\ref{eq:psi_2sigma}).
\end{proof}
Using Lemma~\ref{lem:p_derivatives} we can immediately bound $\nabla\psi(x)^{\top}\nabla^{2}\psi(x)^{-1}\nabla\psi(x)$.
\begin{lem}
\label{lem:force} For all $x\in\interior$, we have $\nabla\psi(x)^{\top}\nabla^{2}\psi(x)^{-1}\nabla\psi(x)\leq n$.
\end{lem}

\begin{proof}
Since $\nabla\psi(x)=-\ma_{x}^{\top}\sigma_{x}$ and $\nabla^{2}\psi(x)\succeq\ma_{x}^{\top}\mSigma_{x}\ma_{x}$
by Lemma~\ref{lem:p_derivatives} we have. 
\[
\nabla\psi(x)^{\top}\nabla^{2}\psi(x)^{-1}\nabla\psi(x)\leq\sigma_{x}^{\top}\ma_{x}\left(\ma_{x}^{\top}\mSigma_{x}\ma_{x}\right)^{-1}\ma_{x}^{\top}\sigma_{x}=1^{\top}\mSigma_{x}^{1/2}\mProj\mSigma_{x}^{1/2}1
\]
where $\mProj\defeq\mSigma_{x}^{1/2}\ma_{x}(\ma_{x}^{\top}\mSigma_{x}\ma_{x})^{-1}\ma_{x}^{\top}\mSigma_{x}^{1/2}.$
Since $\mProj$ is a projection matrix, $\mProj\preceq\mi$ (Lemma~\ref{lem:tool:projection_matrices})
and 
\[
\nabla\psi(x)^{\top}\nabla^{2}\psi(x)^{-1}\nabla\psi(x)\leq1^{\top}\mSigma_{x}1=\sum_{i\in[m]}[\sigma_{x}]_{i}=n\,.
\]
\end{proof}

\subsection{Stability of Lewis Weight Barrier \label{subsec:w_Sigma_Px_stable}}

Here we bound the directional derivatives of the barrier and show
that they are not too large. Lemma~\ref{lem:stable_hess} proved
in this section, combined with Lemma~\ref{lem:force} of the previous
section immediately prove Theorem~\ref{thm:LS_barrier_sc}, bounding
the self-concordance of $\psi$.

Throughout this section, to simplify the notation, we fix an arbitrary
point $x\in\interior$ and a direction $h\in\R^{n}$ and define $x_{t}\defeq x+th$,
$s_{t}=\ma x_{t}-b$, and $\ma_{t}=\ma_{x_{t}}$ and further define
$w_{t}$, $\mw_{t}$, $\mSigma_{t}$, $\mProj_{t}^{(2)}$, $\mLambda_{t}$,
$\mNormProjLap_{t}$,, and $\mn_{t}$ (Lemma~\ref{lem:p_derivatives})
analogously.

First, we bound the derivatives of the slacks and weights in the following
Lemma~\ref{lem:s_lipschitz} and \ref{lem:w_lipschitz}.
\begin{lem}
\label{lem:s_lipschitz} For all $x\in\interior$ and $h\in\Rn$ we
have
\[
\normFull{\ms_{t}^{-1}\frac{d}{dt}s_{t}}_{\mw_{t}}\leq\norm h_{\ma_{t}^{\top}\mw_{t}\ma_{t}}\text{ and }\normFull{\ms_{t}^{-1}\frac{d}{dt}s_{t}}_{\infty}\leq\sqrt{2}m^{\frac{1}{q+2}}\norm h_{\ma_{t}^{\top}\mw_{t}\ma_{t}}\,.
\]
\end{lem}

\begin{proof}
Since $\ms_{t}^{-1}\frac{d}{dt}s_{t}=\ma_{t}h$ we have $\norm{\ms_{t}^{-1}\frac{d}{dt}s_{t}}_{\mw_{t}}=\norm h_{\ma_{t}^{\top}\mw_{t}\ma_{t}}$.
For the second inequality note that by Cauchy Schwarz,
\begin{align*}
\normFull{\ms_{t}^{-1}\frac{d}{dt}s_{t}}_{\infty} & =\norm{\ma_{x}h}_{\infty}=\max_{i\in[m]}\left|\left\langle \cordVec i,\ma_{x}h\right\rangle \right|=\max_{i\in[m]}\left|\left\langle (\ma_{x}^{\top}\mw_{x}\ma_{x})^{-1/2}\ma_{x}^{\top}\cordVec i,(\ma_{x}^{\top}\mw_{x}\ma_{x})^{1/2}h\right\rangle \right|\\
 & \leq\sqrt{\max_{i\in[m]}\left[\ma_{x}(\ma_{x}^{\top}\mw_{x}\ma_{x})^{-1}\ma_{x}^{\top}\right]_{ii}}\norm h_{\ma_{x}^{\top}\mw_{x}\ma_{x}}\,.
\end{align*}
The result follows that Lemma~\ref{lem:lewis_p_change} shows 
\begin{align*}
\max_{i\in[m]}\left[\ma_{x}(\ma_{x}^{\top}\mw_{x}\ma_{x})^{-1}\ma_{x}^{\top}\right]_{ii} & =\max_{i\in[m]}\sigma(\mw_{x}^{\frac{1}{2}}\ma_{x})_{i}[w_{x}]_{i}^{-1}\leq2m^{\frac{2}{q+2}}.
\end{align*}
\end{proof}
\begin{lem}
\label{lem:w_lipschitz} For all $x\in\interior$ and $h\in\Rn$ we
have
\[
\normFull{\mw_{t}^{-1}\frac{d}{dt}w_{t}}_{\mw_{t}}\leq q\norm h_{\ma_{t}^{\top}\mw_{t}\ma_{t}}\text{ and }\normFull{\mw_{t}^{-1}\frac{d}{dt}w_{t}}_{\infty}\leq q\left(\sqrt{2}\cdot m^{\frac{1}{q+2}}+\max\left\{ \frac{q}{2},1\right\} \right)\norm h_{\ma_{t}^{\top}\mw_{t}\ma_{t}}.
\]
\end{lem}

\begin{proof}
Since that $w_{t}=w_{q}(\ms_{t}^{-1}\ma)$, chain rule, Lemma~\ref{lem:LS_weights_stable},
and Lemma~\ref{lem:s_lipschitz} shows that the function $p(x)\defeq w_{q}(\ms_{x}^{-1}\ma)$
satisfies 
\[
\norm{\mj_{p}(x_{t})h}_{\mw_{t}^{-1}}\leq q\normFull{(\ms_{t}^{-1})^{-1}\ms_{t}^{-2}\frac{d}{dt}s_{t}}_{\mw_{t}}=q\normFull{\ms_{t}^{-1}\frac{d}{dt}s_{t}}_{\mw_{t}}=q\norm h_{\ma_{t}^{\top}\mw_{t}\ma_{t}}.
\]
The same tools also show that 
\begin{align*}
\norm{\mw_{t}^{-1}\mj_{p}(x_{t})h}_{\infty} & \leq q\normFull{(\ms_{t}^{-1})^{-1}\ms_{t}^{-2}\frac{d}{dt}s_{t}}_{\infty}+q\cdot\max\left\{ \frac{q}{2},1\right\} \cdot\normFull{(\ms_{t}^{-1})^{-1}\ms_{t}^{-2}\frac{d}{dt}s_{t}}_{\mw_{t}}\\
 & \leq q\normFull{\ms_{t}^{-1}\frac{d}{dt}s_{t}}_{\infty}+q\cdot\max\left\{ \frac{q}{2},1\right\} \cdot\normFull{\ms_{t}^{-1}\frac{d}{dt}s_{t}}_{\mw_{t}}\\
 & \leq q\sqrt{2}\cdot m^{\frac{1}{q+2}}\cdot\norm h_{\ma_{t}^{\top}\mw_{t}\ma_{t}}+q\cdot\max\left\{ \frac{q}{2},1\right\} \cdot\norm h_{\ma_{t}^{\top}\mw_{t}\ma_{t}}.
\end{align*}
Since $\frac{d}{dt}w_{t}=\mj_{p}(x_{t})h$ the result follows.
\end{proof}
Now, recall that by Lemma~\ref{lem:p_derivatives} we have $\hess\psi(x)=\ma_{t}^{\top}\mSigma_{t}^{1/2}(\mi+\mn_{t})\mSigma_{t}^{1/2}\ma_{x}$
where $\mn_{t}\defeq2\mNormProjLap_{t}(\mi-c_{q}\mNormProjLap_{t})^{-1}$
and $c_{q}=1-(2/q)$. Since we have already bounded the stability
of $\ma_{t}$ and $\mSigma_{t}$ all that remains is to bound the
stability of $\mNormProjLap_{t}$ and leverage this to bound the stability
$\mn_{t}$ and $\hess\psi(x_{t})$.

To simplify these calculation for all $t>0$ and $\alpha\in\R$ we
define $z_{t,\alpha}\in\R^{n}$ be defined for all $i\in[n]$ by $[z_{t,\alpha}]_{i}=\frac{d}{dt}\ln\left([w_{t}]_{i}^{\alpha}/[s_{t}]_{i}\right)$
and $\mz_{t,\alpha}\defeq\mDiag(z_{t,\alpha})$. We will repeatedly
use the fact
\begin{equation}
\frac{d}{dt}\mw_{t}^{\alpha}\ms_{t}^{-1}=\mw_{t}^{\alpha}\ms_{t}^{-1}\frac{d}{dt}\ln(\mw_{t}^{\alpha}\ms_{t}^{-1})=\mw_{t}^{\alpha}\ms_{t}^{-1}\mz_{t,\alpha}.\label{eq:Z_fact}
\end{equation}
In the following lemma we bound $z_{t,\alpha}$ and use this to simplify
these derivative bounds. 
\begin{lem}
\label{lem:z} For all $x\in\interior$ and $h\in\Rn$ we have and
$z_{t,\alpha}\in\R^{n}$ defined for all $i\in[n]$ by $[z_{t,\alpha}]_{i}=\frac{d}{dt}\ln\left([w_{t}]_{i}^{\alpha}/[s_{t}]_{i}\right)$
we have that 
\[
\norm{z_{t}}_{\mSigma}\leq\left(|\alpha|q+1\right)\norm h_{\hess\psi(x_{t})}\text{ and }\norm{z_{t}}_{\infty}\leq\left((|\alpha|q+1)\sqrt{2}m^{\frac{1}{q+2}}+q|\alpha|\max\left\{ \frac{q}{2},1\right\} \right)\norm h_{\nabla^{2}\psi(x_{t})}
\]
\end{lem}

\begin{proof}
Note that $[z_{t,\alpha}]_{i}=\alpha\cdot(\frac{d}{dt}[w_{t}]_{i}/[w_{t}]_{i})-(\frac{d}{dt}[s_{t}]_{i}/[s_{t}]_{i})$
and consequently the result therefore follows from triangle inequality,
Lemma~\ref{lem:s_lipschitz}, Lemma~\ref{lem:w_lipschitz}, and
$\ma_{t}^{\top}\mw_{t}\ma_{t}\preceq\hess\psi(x_{t})$.
\end{proof}
Using this we can bound the stability of $\mNormProjLap_{t}$
\begin{lem}
\label{lem:change_of_nor_proj_lap} For all $x\in\interior$ and $h\in\Rn$
we have
\[
\normFull{\frac{d}{dt}\mNormProjLap_{t}}_{2}\leq\max\{3q,16\}\norm h_{\hess\psi(x_{t})}\,.
\]
\end{lem}

\begin{proof}
Let $\mq_{t}\defeq\mw_{t}^{-1/4}\mProj_{t}\mw_{t}^{-1/4}$. Since
$\mw_{t}=\mSigma_{t}$ this implies that 
\[
\mNormProjLap_{t}=\mi-\mSigma_{t}^{-1/2}\mProj_{t}^{(2)}\mSigma_{t}^{-1/2}=\mi-\mq_{t}^{(2)}\,.
\]
Now, let, $z_{t,\alpha}$ be defined as in Lemma~\ref{lem:z} and
let $\mz_{t,\alpha}\defeq\mDiag(z_{t,\alpha})$. Since 
\[
\frac{d}{dt}\left(\ma_{t}^{\top}\mw_{t}^{1-\frac{2}{q}}\ma_{t}\right)^{-1}=-\left(\ma_{t}^{\top}\mw_{t}^{1-\frac{2}{q}}\ma_{t}\right)^{-1}\frac{d}{dt}\left[\ma_{t}^{\top}\mw_{t}^{1-\frac{2}{q}}\ma_{t}\right]\left(\ma_{t}^{\top}\mw_{t}^{1-\frac{2}{q}}\ma_{t}\right)^{-1}
\]
and 
\[
\frac{d}{dt}\ms_{t}^{-2}\mw_{t}^{1-\frac{2}{q}}=\frac{d}{dt}\left[\left(\ms_{t}^{-1}\mw_{t}^{\frac{1}{2}-\frac{1}{q}}\right)^{2}\right]=2\ms_{t}^{-2}\mw_{t}^{1-\frac{2}{q}}\mz_{t,\frac{1}{2}-\frac{1}{q}}
\]
(using (\ref{eq:Z_fact})), we have that 
\begin{align*}
\frac{d}{dt}\left[\mq_{t}\right]_{ij} & =\frac{d}{dt}\cordVec i^{\top}\mw_{t}^{\frac{1}{4}-\frac{1}{q}}\ma_{t}\left(\ma_{t}^{\top}\mw_{t}^{1-\frac{2}{q}}\ma_{t}\right)^{-1}\ma_{t}^{\top}\mw_{t}^{\frac{1}{4}-\frac{1}{q}}\cordVec j\\
 & =[z_{t,\frac{1}{4}-\frac{1}{q}}]_{i}\left[\mq_{t}\right]_{ij}+\left[\mq_{t}\right]_{ij}[z_{t,\frac{1}{4}-\frac{1}{q}}]_{j}-2\left[\mq_{t}\mw_{t}^{1/4}\mz_{t,\frac{1}{2}-\frac{1}{q}}\mw_{t}^{1/4}\mq_{t}\right]_{ij}\,.
\end{align*}
Consequently, 
\[
\frac{d}{dt}\mq_{t}=\mz_{t,\frac{1}{4}-\frac{1}{q}}\mq_{t}+\mq_{t}\mz_{t,\frac{1}{4}-\frac{1}{q}}-2\mq_{t}\mw_{t}^{1/4}\mz_{t,\frac{1}{2}-\frac{1}{q}}\mw_{t}^{1/4}\mq_{t}
\]
and by chain rule we have that 
\begin{align*}
\frac{d}{dt}\mq_{t}^{(2)} & =2\mq_{t}\circ\left[\frac{d}{dt}\mq_{t}\right]=2\mz_{t,\frac{1}{4}-\frac{1}{q}}\mq_{t}^{(2)}+2\mq_{t}^{(2)}\mz_{t,\frac{1}{4}-\frac{1}{q}}-4\left[\mq_{t}\circ\mq_{t}\mw_{t}^{1/4}\mz_{t,\frac{1}{2}-\frac{1}{q}}\mw_{t}^{1/4}\mq_{t}\right]
\end{align*}
Therefore, for all $y\in\R^{n}$ we have
\[
y^{\top}\left[\frac{d}{dt}\mq_{t}^{(2)}\right]y=4y^{\top}\mSigma_{t}^{-1/2}\mz_{t,\frac{1}{4}-\frac{1}{q}}\mProj_{t}^{(2)}\mSigma_{t}^{-1/2}y-4y^{\top}\left[\mProj_{t}\circ\mSigma_{t}^{-1/2}\mProj_{t}\mz_{t,\frac{1}{2}-\frac{1}{q}}\mProj_{t}\mSigma_{t}^{-1/2}\right]y
\]
Applying Lemma~\ref{lem:tool:projection_matrices} yields
\[
\left|y^{\top}\left[\frac{d}{dt}\mq_{t}^{(2)}\right]y\right|\leq4\norm{\mSigma_{t}^{-1/2}y}_{\mSigma_{t}}^{2}\norm{z_{t,\frac{1}{4}-\frac{1}{q}}}_{\mSigma_{t}}+4\norm{\mSigma_{t}^{-1/2}y}_{\mSigma_{t}}^{2}\norm{z_{t,\frac{1}{2}-\frac{1}{q}}}_{\mSigma_{t}}\,.
\]
Since, $\norm{\mSigma_{t}^{-1/2}y}_{\mSigma_{t}}=\norm y_{2}$ applying
Lemma~\ref{lem:z} then yields that 
\[
\normFull{\frac{d}{dt}\mq_{t}^{(2)}}_{2}\leq4\left[\norm{z_{t,\frac{1}{4}-\frac{1}{q}}}_{\mSigma_{t}}+\norm{z_{t,\frac{1}{2}-\frac{1}{q}}}_{\mSigma_{t}}\right]\leq4\left[\left(\left|\frac{1}{4}-\frac{1}{q}\right|q+1\right)+\left(\left|\frac{1}{2}-\frac{1}{q}\right|q+1\right)\right]\norm h_{\hess\psi(x_{t})}\,.
\]
Since $\left|q-4\right|+\left|2q-4\right|+8\leq\max\{3q,16\}$ and
$\frac{d}{dt}\mq_{t}^{(2)}=\frac{d}{dt}\mNormProjLap_{t}$ the result
follows. 
\end{proof}
Using this we can now bound the stability of $\mn_{t}$.
\begin{lem}
\label{lem:change_of_n} For all $x\in\interior$ and $h\in\Rn$ we
have
\[
\normFull{\left(\mi+\mn_{t}\right)^{-\frac{1}{2}}\frac{d}{dt}\mn_{t}\left(\mi+\mn_{t}\right)^{-\frac{1}{2}}}_{2}\leq2\max\left\{ 1,\frac{q}{2}\right\} \normFull{\frac{d}{dt}\mNormProjLap_{t}}_{2}\leq4\max\{q,2\}^{2}\norm h_{\hess\psi(x_{t})}\,.
\]
\end{lem}

\begin{proof}
Recall that $\mn_{t}\defeq2\mNormProjLap_{t}(\mi-c_{q}\mNormProjLap_{t})^{-1}$
for $c_{q}=1-\frac{2}{q}$ . Direct calculation yields
\begin{align*}
\frac{d}{dt}\mn_{t} & =2\left[\frac{d}{dt}\mNormProjLap_{t}\right]\left(\mi-c_{q}\mNormProjLap_{t}\right)^{-1}-2\mNormProjLap_{t}\left(\mi-c_{q}\mNormProjLap_{t}\right)^{-1}\left[(-\alpha_{q})\frac{d}{dt}\mNormProjLap_{t}\right]\left(\mi-c_{q}\mNormProjLap_{t}\right)^{-1}\\
 & =2\left[\mi+c_{q}\mNormProjLap_{t}\left(\mi-c_{q}\mNormProjLap_{t}\right)^{-1}\right]\cdot\left[\frac{d}{dt}\mNormProjLap_{t}\right]\cdot\left[\mi-c_{q}\mNormProjLap_{t}\right]^{-1}\\
 & =2\cdot\left[\mi-c_{q}\mNormProjLap_{t}\right]^{-1}\cdot\left[\frac{d}{dt}\mNormProjLap_{t}\right]\cdot\left[\mi-c_{q}\mNormProjLap_{t}\right]^{-1}\,.
\end{align*}
Now, since $c_{q}\in(-\infty,1)$ and $\mZero\preceq\mNormProjLap\preceq\mi$
we have 
\[
\mi+\mn_{t}=(\mi-c_{q}\mNormProjLap_{t})^{-1}(\mi+(2-c_{q})\mNormProjLap)\succeq(\mi-c_{q}\mNormProjLap_{t})^{-1}
\]
and therefore $\norm{(\mi+\mn_{t})^{-\frac{1}{2}}(\mi-c_{q}\mNormProjLap_{t})^{-\frac{1}{2}}}_{2}\leq1$.
Further as $\norm{(\mi-c_{q}\mNormProjLap_{t})^{-1}}_{2}\leq\max\{1,\frac{q}{2}\}$
(see, e.g. (\ref{eq:norm_inv_bound})), the result follows from Lemma~\ref{lem:change_of_nor_proj_lap}.
\end{proof}
Now we can combine everything to prove the desired result
\begin{lem}
\label{lem:stable_hess}For all $x\in\interior$ and $h,y\in\Rn$
we have
\[
\left|y^{\top}\left[\frac{d}{dt}\nabla^{2}\psi(x_{t})\right]y\right|\preceq\left((q+2)^{3/2}\sqrt{2}m^{\frac{1}{q+2}}+6\max\{q,2\}^{2.5}\right)\norm h_{\nabla^{2}\psi(x_{t})}\left[y^{\top}\nabla^{2}\psi(x_{t})y\right].
\]
\end{lem}

\begin{proof}
Since $\hess\psi(x)=\ma_{t}^{\top}\mSigma_{t}^{1/2}(\mi+\mn_{t})\mSigma_{t}^{1/2}\ma_{x}$
by Lemma~\ref{lem:p_derivatives} we have
\[
y^{\top}\left[\frac{d}{dt}\nabla^{2}\psi(x_{t})\right]y=2y^{\top}\ma_{t}^{\top}\mw_{t}^{1/2}\mz_{t,1/2}\left[\mi+\mn_{t}\right]\mw_{t}^{1/2}\ma_{t}y+y^{\top}\ma_{t}^{\top}\mw_{t}^{1/2}\left[\frac{d}{dt}\mn_{t}\right]\mw_{t}^{1/2}\ma_{t}y\,.
\]
Cauchy Schwarz, and the facts that $\mi+\mn_{t}$ is PSD and $\hess\psi(x)=\ma_{t}^{\top}\mSigma_{t}^{1/2}(\mi+\mn_{t})\mSigma_{t}^{1/2}\ma_{x}$
yield
\begin{align*}
\left|y^{\top}\ma_{t}^{\top}\mw_{t}^{1/2}\mz_{t,1/2}\left[\mi+\mn_{t}\right]\mw_{t}^{1/2}\ma_{t}y\right| & \leq\norm{\mw_{t}^{1/2}\ma_{t}y}_{2}\norm{\mz_{t,1/2}\left[\mi+\mn_{t}\right]\mw_{t}^{1/2}\ma_{t}y}_{2}\\
 & \leq\norm{\mz_{t,1/2}}_{2}\sqrt{\norm{\mi+\mn_{t}}_{2}}\norm y_{\hess\psi(x_{t})}^{2}\,.
\end{align*}
and
\[
\left|y^{\top}\ma_{t}^{\top}\mw_{t}^{1/2}\left[\frac{d}{dt}\mn_{t}\right]\mw_{t}^{1/2}\ma_{t}y\right|\leq\norm y_{\hess\psi(x_{t})}^{2}\normFull{\left(\mi+\mn_{t}\right)^{-\frac{1}{2}}\frac{d}{dt}\mn_{t}\left(\mi+\mn_{t}\right)^{-\frac{1}{2}}}_{2}\,.
\]
Now, by Lemma~\ref{lem:z} 
\[
\norm{\mz_{t,1/2}}_{2}=\norm{z_{t,1/2}}_{\infty}\leq\left(\left(\frac{q}{2}+1\right)\sqrt{2}m^{\frac{1}{q+2}}+\frac{q}{2}\max\left\{ \frac{q}{2},1\right\} \right)\norm h_{\nabla^{2}\psi(x_{t})}\,.
\]
Further, since $\norm{\mi+\mn_{t}}_{2}\leq1+q$ combining and applying
Lemma~\ref{lem:change_of_n} yields that $\left|y^{\top}\left[\frac{d}{dt}\nabla^{2}\psi(x_{t})\right]y\right|$
is bounded by
\[
\left(2\sqrt{1+q}\left(\left(\frac{q}{2}+1\right)\sqrt{2}\cdot m^{\frac{1}{q+2}}+\frac{q}{2}\max\left\{ \frac{q}{2},1\right\} \right)+4\max\{q,2\}^{2}\right)\norm h_{\nabla^{2}\psi(x_{t})}\left[y^{\top}\nabla^{2}\psi(x_{t})y\right]
\]
and the result follows by basic calculations.
\end{proof}

\section{Efficient Algorithms}

\label{sec:master_thm}\label{sec:algorithms}

In this section we show how to leverage the results of the previous
sections to obtain efficient algorithms and derive the main results
of this paper. In Section~\ref{sec:lp_alg} we prove Theorem~\ref{thm:LPSolve_detailed}
and Theorem~\ref{thm:LPSolve_detailed_dual}, our main results on
linear programming, in Section~\ref{sec:app:mincost} we prove Theorem~\ref{thm:maxflow},
our main result on minimum cost maximum flow, and in Section~\ref{sec:barrier_algorithm}
we prove Theorem~\ref{thm:lewisbarriercomp} and a more general Theorem~\ref{thm:lewisbarriercomputefull},
our main results on a polynomial time computable nearly-universal
self-concordant barrier. The algorithms in this section make critical
use of algorithms for computing Lewis weights provided an analyzed
in Appendix~\ref{sec:weights_full:computing} and are stated as needed.

\subsection{Linear Programming Algorithm}

\label{sec:lp_alg}

Here we show how to combine the results of the preceding sections
to obtain our efficient linear programming algorithm and prove Theorem~\ref{thm:LPSolve_detailed}
and Theorem~\ref{thm:LPSolve_detailed_dual}. Our algorithm uses
the following result regarding approximately computing Lewis weights
proved in Appendix~\ref{sec:weights_full:computing}.

\begin{restatable}[Approximate Weight Computation]{thm}{lewisweightapproxfull}\label{thm:lewisweightexactapproxfull}
Let $\ma\in\R^{m\times n}$ be non-degenerate and let $\mathcal{T}_{w}$
and $\mathcal{T}_{d}$ denote the work and depth needed to compute
$(\ma^{\top}\md\ma)^{-1}z$ for arbitrary positive diagonal matrix
$\md$ and vector $z$. For all $\epsilon\in(0,1)$, $p\in(0,4)$,
$w^{(0)}\in\dWeights$ with $\normInf{w_{(0)}^{-1}(\lpweight(\ma)-w^{(0)})}\leq2^{-20}p^{2}(4-p)$,
the algorithm\textbf{ $\code{computeApxWeight}(\vx,w^{(0)},\varepsilon)$}
can be implemented to return $w$ that with high probability in $n$
$\norm{\lpweight(\ma)^{-1}(\lpweight(\ma)-w)}_{\infty}\leq\epsilon$
in $O(p^{-1}(4-p)^{-2}\epsilon^{-2}\log^{2}(n/(p\epsilon))$ steps
each of which can be implemented in $O(\nnz(\ma)+\mathcal{T}_{w})$
work and $O(\mathcal{T}_{d})$ depth. 

Without $w^{(0)}$ the algorithm $\code{computeInitialWeight}(\ma,p,\epsilon)$
(Algorithm~\ref{alg:initialweight}) can be implemented to have the
same guarantee with $O(\sqrt{n}(4-p)^{-3}p^{-3})\log\frac{m}{n}\log^{2}(n/(p\epsilon))$
steps of the same cost.

\end{restatable}

Leveraging this result in Algorithm~\ref{alg:pathFollow} we give
the procedure, $\code{pathFollowing}$, for approximately following
the weighted central path induced by regularized Lewis weights and
in Theorem~\ref{thm:LPSolve} we analyze it. Interestingly, the regularization
(i.e. choosing $c_{0}>0$ in Section~\ref{sec:weight-function})
is not needed for this path following procedure to work. Instead,
it is used to reason about the conditioning of the systems encountered
by this method and for leveraging the procedure to efficiently solve
linear programs.

\begin{algorithm2e}[H]\label{alg:pathFollow}

\caption{$\ensuremath{(\vx^{\mathrm{(final)}},w^{\mathrm{(final)}})=\code{pathFollowing}(\vx,\vWeight,t_{\mathrm{start}},t_{\mathrm{end}},\epsilon)}$}

\SetAlgoLined

$t=t_{\mathrm{start}},K=\frac{1}{16c_{k}},\alpha=\frac{R}{1600\sqrt{n}\log^{2}m}$
where $R$ is defined in Theorem~\ref{thm:smoothing:centering_inexact_weight}.

\Repeat{$t=t_{\mathrm{end}}$}{

$(\next{\vx},\next{\vWeight})=\code{centeringInexact}(\vx,\vWeight,K)$
where $\code{computeApxWeight}$ to approximate $\vg(\vx)$ (defined
in Section~\ref{sec:weight-function} Theorem~\ref{thm:max_flow:weight_properties}
for $p=1-\frac{1}{\log4m}$ and $c_{0}=\frac{2n}{m}$).

$t\leftarrow\mathtt{median}((1-\alpha)\cdot t,t_{\mathrm{end}},(1+\alpha)\cdot t)$.

$\vx\leftarrow\next{\vx}$, $\vWeight\leftarrow\next{\vWeight}$.

}

\For{$i=1,\cdots,4c_{k}\log(\frac{1}{\epsilon})$ }{

$(\vx,\vWeight)=\code{centeringInexact}(\vx,\vWeight,K)$ where $\code{computeApxWeight}$
is used to approximate $\vg(\vx)$.

}

\textbf{Output:} $(\vx,\vWeight)$.

\end{algorithm2e}
\begin{thm}
\label{thm:LPSolve} Define $\mu$ as defined in Theorem \ref{thm:smoothing:centering_inexact_weight}and
suppose that
\[
\delta_{t_{\text{start}}}(\vx,\vWeight)\leq\frac{1}{2^{16}\log^{3}m}\enspace\text{ and }\enspace\Phi_{\mu}(\log g(x)-\log w)\leq36c_{1}c_{s}c_{k}m.
\]
where. If $\ensuremath{(\vx^{\mathrm{(final)}},\next{\vWeight})=\code{pathFollowing}(\vx,\vWeight,t_{\mathrm{start}},t_{\mathrm{end}})}$,
then with high probability in $n$,
\[
\delta_{t_{\mathrm{end}}}(\vx^{\mathrm{(final)}},w^{\mathrm{(final)}})\leq\epsilon\enspace\text{ and }\enspace\Phi_{\mu}(\log g(\vx^{\mathrm{(final)}})-\log w^{\mathrm{(final)}})\leq36c_{1}c_{s}c_{k}m.
\]
Further, $\code{pathFollowing}(\vx,\vWeight,t_{\mathrm{start}},t_{\mathrm{end}})$
can be implemented
\[
O\left(\sqrt{n}\log^{13}m\cdot\kappa\cdot\mathcal{T}_{w}\right)\text{ work and }O\left(\sqrt{n}\log^{13}m\cdot\kappa\cdot\mathcal{T}_{w}\right)\text{ depth}
\]
where $\kappa=\left|\log\frac{t_{\mathrm{end}}}{t_{\mathrm{start}}}\right|+\log\frac{1}{\epsilon}$
and $\mathcal{T}_{w}$ and $\mathcal{T}_{d}$ are the work and depth
needed to compute $(\ma^{\top}\md\ma)^{-1}\vq$ for input positive
diagonal matrix $\md$ and vector $\vq$. Furthermore, with high probability
in $n$ during the whole algorithm, we have $\delta_{t}(\vx,\vWeight)\leq R\defeq\frac{1}{768c_{k}^{2}\log(36c_{1}c_{s}c_{k}m)}$
and $\Phi_{\mu}(\log g(x)-\log w)\leq36c_{1}c_{s}c_{k}m$.
\end{thm}

\begin{proof}
We first show that $\code{pathFollowing}$ maintains the invariant
that $\delta_{t}(\vx,\vWeight)\leq R$ and $\Phi_{\mu}(\log g(x)-\log w)\leq36c_{1}c_{s}c_{k}m$
in each iteration where
\[
R\defeq\frac{K}{48c_{k}\log(36c_{1}c_{s}c_{k}m)}=\frac{1}{768c_{k}^{2}\log(36c_{1}c_{s}c_{k}m)}=\frac{1}{3072\log^{2}m\log(288nm\log m)}\geq\frac{1}{2^{16}\log^{3}m}.
\]
Note that this holds for the input $(x,w)$ by assumption, so suppose
that this holds at the start of one of the loops. By the definition
of $\Phi_{\mu}$, $\mu$ and $K$, we have $\norm{\log\vg(x)-\log\vWeight}_{\infty}\leq R$
and by \eqref{eq:log_g_change}, we have $\norm{\log\vg(x)-\log g(\next x)}_{\infty}\leq R$.
Therefore, we have
\begin{equation}
\norm{\log\vg(\next x)-\log\vWeight}_{\infty}\leq2R\leq\frac{1}{80(\frac{p}{2}+\frac{2}{p})}\label{eq:gw_bound}
\end{equation}
where we used the formula of $R$ and $p=\frac{1}{\log(4m)}$ at the
end. Thus, the weight $w$ satisfies the conditions for Theorem~\ref{thm:lewisweightexactapproxfull}
and the algorithm $\code{centeringInexact}$ can use the function
$\code{computeApxWeight}$ to find the approximation of $\vg(\next{\vx})$.
Consequently, by Lemma~\ref{thm:smoothing:centering_inexact_weight}
with high probability in $n$
\[
\delta_{t}(\next{\vx},\next{\vWeight})\leq\left(1-\frac{1}{4c_{k}}\right)\delta_{t}(x,w)\enspace\text{ and }\enspace\Phi_{\mu}(\log g(\next{\vx})-\log\next{\vWeight})\leq36c_{1}c_{s}c_{k}m\,.
\]
Using Lemma~\ref{lem:gen:t_step}, \eqref{eq:gw_bound} and Theorem
\ref{thm:max_flow:weight_properties},we have
\begin{align*}
\delta_{\next t}(\next{\vx},\next{\vWeight}) & \leq(1+\alpha)\left(1-\frac{1}{4c_{k}}\right)\delta_{t}(x,w)+\alpha\left(1+\cnorm\sqrt{\norm{\vWeight}_{1}}\right)\\
 & \leq\delta_{t}(x,w)-\frac{\delta_{t}(x,w)}{8c_{k}}+100\alpha\sqrt{n}\log m\leq\delta_{t}(x,w).
\end{align*}
Hence, we proved that the invariant. Note that in the second loop,
$t$ does not change and therefore $\delta_{t}(x,w)$ decreases by
$(1-\frac{1}{4c_{k}})$ in each step with high probability in $n$
yielding that $\delta_{t_{\mathrm{end}}}(\vx^{\mathrm{(final)}},w^{\mathrm{(final)}})\leq\epsilon$
as desired.\footnote{Note that the with high probability claim of this theorem requires
that $|\log(t_{\mathrm{end}}/t_{\mathrm{start}})|+|\log(1/\epsilon)|=O(\poly(n))$.
However, if that is not the case, then every $O(\poly(n))$ steps
of the loops we can afford to exactly check the invariants and compute
the weights by Theorem~\ref{thm:lewisexactfull} which we introduce
later and repeat the steps if the invariants do not hold. This increases
the expected running time only by multiplicative constants and we
can run the algorithm $O(\log(n))$ times in parallel to ensure that
one of them outputs a point with the correct invariants without taking
more than twice the desired runtime with high probability.}

To bound the runtime, note that $R=\Omega(\log^{-3}m)$ and hence
$\alpha=\Omega(n^{-1/2}\log^{-5}m)$. Therefore, the total number
of step is $O(\sqrt{n}\log^{5}m\cdot[\left|\log(t_{\mathrm{end}}/t_{\mathrm{start}})\right|+\log\left(1/\epsilon\right)])$.
Finally, each step involves computing projection to the mixed ball
and computing Lewis weights. Theorem~\ref{thm:project_mixed_norm}
shows that the projection can be formed in $O(m\log m)$ time and
$O(\log m)$ depth. Theorem~\ref{thm:zerogame} shows that we need
to compute Lewis weight with $1\pm R=1\pm\Theta(\log^{-3}m)$ multiplicative
approximation. Theorem \ref{thm:lewisweightexactapproxfull} shows
that we can compute the Lewis weights using $O((1/R^{2})\log^{2}(m/R))=O(\log^{8}m)$
linear systems solves of the desired form.
\end{proof}
To leverage this result we first provide Lemma~\ref{lem:weighted_path:duality_gap},
which bounds how large a $t$ is needed guarantee an approximately
optimal solution. Further, in Lemma~\ref{lem: distance gurantee},
we show how much the approximate centrality hurts our guarantee. Using
these lemmas and the previous section, we conclude by describing our
linear programming algorithm, $\code{LPSolve}$, and prove Theorem~\ref{thm:LPSolve_detailed}
and Theorem~\ref{thm:LPSolve_detailed_dual}.
\begin{lem}
[{\cite[Theorem 4.2.7]{Nesterov2003}}]\label{lem:weighted_path:duality_gap}
Let $x^{*}\in\Rm$ denote an optimal solution to \eqref{eq:intro:two-sided}
and $\vx_{t}=\arg\min f_{t}\left(\vx,\vWeight\right)$ for some $t>0$
and $\vWeight\in\dWeight$. Then the following holds
\[
\vc^{\top}\vx_{t}(\vWeight)-\vc^{\top}\vx^{*}\leq\frac{\norm{\vWeight}_{1}}{t}.
\]
\end{lem}

\begin{proof}
By the optimality conditions of \eqref{eq:intro:two-sided} we know
that $\grad_{x}f_{t}(\vx_{t}(\vWeight))=t\cdot\vc+\vWeight\phi'(\vx_{t}(\vWeight))$
is orthogonal to the kernel of $\ma^{\top}$. Furthermore since $\vx_{t}(\vWeight)-\vx^{*}\in\ker(\ma^{\top})$
we have
\[
\left(t\cdot\vc+\vWeight\phi'(\vx_{t}(\vWeight))\right)^{\top}(\vx_{t}(\vWeight)-\vx^{*})=0.
\]
Using that $\phi_{i}'(x_{t}(\vWeight)_{i})\cdot(x_{i}^{*}-x_{t}(\vWeight)_{i})\leq1$
by Lemma~\ref{lem:gen:phi_properties_force} then yields
\begin{align*}
\vc^{\top}(\vx_{t}(\vWeight)-\vx^{*}) & =\frac{1}{t}\sum_{i\in[m]}w_{i}\cdot\phi_{i}'(x_{t}(\vWeight)_{i})\cdot(x_{i}^{*}-x_{t}(\vWeight)_{i})\leq\frac{\norm{\vWeight}_{1}}{t}.
\end{align*}
\end{proof}
\begin{lem}
\label{lem: distance gurantee} For $x$ such that $\delta_{t}(x,\vg(x))\leq\frac{1}{2^{16}\log^{3}m}$
and $\vx_{t}\defeq\arg\min f_{t}\left(\vx,\vWeight\right)$ we have
\[
\normFull{\sqrt{\phi''(\vx_{t})}\left(x-\vx_{t}\right)}_{\infty}\leq8\delta_{t}(x,\vg(x)).
\]
\end{lem}

\begin{proof}
We prove this statement via our centering algorithm. We use Theorem~\ref{thm:smoothing:centering_inexact_weight}
with exact weight computation and start with $\vx^{(1)}=x$ and $\vWeight^{(1)}=\vg(\vx^{(1)})$.
In each iteration, $\delta_{t}$ is decreased by a factor of $1-(4c_{k})^{-1}$.
\eqref{eq:centrality_equivalence} shows that 
\begin{equation}
\norm{\sqrt{\phi''(\vx^{(k)})}\left(\vx^{(k+1)}-\vx^{(k)}\right)}_{\infty}\leq2\delta_{t}(\vx^{(k)},\vWeight^{(k)})\label{eq:distance_gurantee}
\end{equation}
where we used $c_{\gamma}\leq2$ (Lemma \ref{lem:c_gamma}). The Lemma
\ref{lem:gen:phi_properties_sim} shows that
\[
\normFull{\log\left(\vphi''(\vx^{(k)})\right)-\log\left(\vphi''(\vx^{(k+1)})\right)}_{\infty}\leq\left(1-4\delta_{t}(\vx^{(k)},\vWeight^{(k)})\right)^{-1}\leq e^{8\delta_{t}(\vx^{(k)},\vWeight^{(k)})}.
\]
Therefore, for any $k$, we have
\[
\normFull{\log\left(\vphi''(\vx^{(k)})\right)-\log\left(\vphi''(x_{t})\right)}_{\infty}\leq e^{8\sum_{i=1}^{k}\delta_{t}(\vx^{(i)},\vWeight^{(i)})}\leq e^{32c_{k}\delta_{t}(\vx^{(1)},\vg(\vx^{(1)}))}\leq2
\]
where we used that $\delta_{t}$ is decreased by a factor of $(1-\frac{1}{4c_{k}})$,
$c_{k}\leq2\log m$ and that $\delta_{t}\leq\frac{1}{2^{16}\log^{3}m}$.
Using this on \eqref{eq:distance_gurantee}, we have
\[
\normFull{\sqrt{\phi''(\vx_{t})}\left(\vx^{(1)}-x_{t}\right)}_{\infty}\leq4\sum_{i=1}^{k}\delta_{t}(\vx^{(i)},\vWeight^{(i)})\leq8\delta_{t}(\vx^{(1)},w^{(1)}).
\]
\end{proof}
\begin{algorithm2e}[H]

\caption{$\ensuremath{\vx^{\mathrm{(final)}}=\code{LPSolve}(\vx_{0},\epsilon)}$}

\label{alg:lp_solve}

\SetAlgoLined

\textbf{Input:} an initial point $\vx_{0}$ such that $\ma^{\top}\vx_{0}=b$.

$\vWeight=\code{computeInitialWeight}(\vx_{0},\frac{1}{2^{16}\log^{3}m})+\frac{n}{2m},$
$d=-\vWeight_{i}\phi_{i}'(\vx_{0})$.

$t_{1}=(2^{27}m^{3/2}U^{2}\log^{4}m)^{-1}$, $t_{2}=\frac{2m}{\epsilon}$,
$\epsilon_{1}=\frac{1}{2^{18}\log^{3}m}$, $\epsilon_{2}=\frac{\epsilon}{8U^{2}}$.

$(\next{\vx},\next{\vWeight})=\ensuremath{\code{pathFollowing}(\vx_{0},\vWeight,1,t_{1},\epsilon_{1})}$
with cost vector $\vd$.

$(\vx^{\mathrm{(final)}},\vWeight^{\mathrm{(final)}})=\ensuremath{\code{pathFollowing}(\next{\vx},\next{\vWeight},t_{1},t_{2},\epsilon_{2})}$
with cost vector $\vc$.

\textbf{Output:} $\vx^{\mathrm{(final)}}$.

\end{algorithm2e}
\begin{proof}[Proof of Theorem~\ref{thm:LPSolve_detailed}]
By Theorem \ref{thm:lewisweightexactapproxfull}, we know $\code{computeInitialWeight}$
gives an weight 
\[
\normInf{\mg(\vx)^{-1}(\vg(\vx_{0})-\vWeight)}\leq\frac{1}{2^{16}\log^{3}m}\leq R.
\]
By the definition of $R$, we have that $\Phi_{\mu}(\log g(\vx_{0})-\log w)\leq36c_{1}c_{s}c_{k}m$
and that $\vx_{0}$ is the minimum of
\[
\min_{x}\vd^{\top}\vx-\sum_{i}\vWeight_{i}\phi_{i}(\vx)\text{ given }\ma^{\top}\vx=\vb.
\]
Therefore, $(\vx,\vWeight)$ satisfies the assumption of theorem \ref{thm:LPSolve}
because $\delta_{t}=0$ and $\Phi_{\mu}$ is small enough. Hence,
we have
\[
\delta_{t_{1}}^{d}(\next{\vx},\next{\vWeight})\leq\frac{1}{2^{18}\log^{3}m}\enspace\text{ and }\enspace\Phi_{\mu}(\log g(\next{\vx})-\next{\vWeight})\leq36c_{1}c_{s}c_{k}m
\]
where we used the superscript $d$ to indicate $\delta$ is defined
using the cost vector $d$. Using this notation and \eqref{eq:centrality_equivalence},
we have
\begin{align}
 & \delta_{t_{1}}^{\vc}(\next{\vx},\next{\vWeight})\leq\mixedNormFull{\mProj_{\next{\vx},\next w}\left(\frac{t_{1}c+\next{\vWeight}\phi'(\next{\vx})}{\next w\sqrt{\phi''(\next{\vx})}}\right)}w\nonumber \\
\leq & \mixedNormFull{\mProj_{\next{\vx},\next w}\left(\frac{t_{1}d+\next{\vWeight}\phi'(\next{\vx})}{\next w\sqrt{\phi''(\next{\vx})}}\right)}w+t_{1}\mixedNormFull{\mProj_{\next{\vx},\next w}\left(\frac{c-d}{\next w\sqrt{\phi''(\next{\vx})}}\right)}w\nonumber \\
\leq & 2\cdot\delta_{t_{1}}^{d}(\next{\vx},\next{\vWeight})+100\sqrt{n}\log m\cdot t_{1}\cdot\normFull{\mProj_{\next{\vx},\next w}\left(\frac{c-d}{\next w\sqrt{\phi''(\next{\vx})}}\right)}_{\infty}\label{eq:delta_c}
\end{align}
where we used Lemma~\ref{lem:c_gamma} at the end. 

Next, we note that for any $x$ and $w$, let $\mq=\mWeight^{-1}\ma_{x}\left(\ma_{x}^{\top}\mWeight^{-1}\ma_{x}\right)^{-1}\ma_{x}^{\top}\mw^{-1}$,
then \eqref{eq:def_Pxw} and sensitivity $c_{s}\leq4$ shows that
\[
\|\mProj_{x,w}\mw^{-1}\|_{\infty\rightarrow\infty}\leq\left(\min_{i\in[m]}w_{i}\right)^{-1}+\|\mq\|_{\infty\rightarrow\infty}\leq2m+m\max_{i\in[m]}\mq_{ii}\leq6m.
\]
Substituting into \eqref{eq:delta_c} yields
\[
\delta_{t_{1}}^{\vc}(\next{\vx},\next{\vWeight})\leq\frac{1}{2^{17}\log^{3}m}+600m^{3/2}\log m\cdot t_{1}\cdot\normFull{\frac{c-d}{\sqrt{\phi''(\next{\vx})}}}_{\infty}\leq\frac{1}{2^{16}\log^{3}m}
\]
where we used that $\|c-d\|_{\infty}\leq\|c\|_{\infty}+\|\phi'(\vx_{0})\|_{\infty}\leq2U$
(Lemma~\ref{lem:gen:phi_properties_force}), $\min_{\vy}\sqrt{\phi''(\vy)}\geq\frac{1}{U}$
(Lemma \ref{lem:gen:phi_properties_sim}) and that we have chosen
$t_{1}$ small enough.

Hence, $(\next{\vx},\next{\vWeight})$ satisfy the assumption of Theorem~\ref{thm:LPSolve}
for the original cost function $c$. Now, we only need to bound how
large $t_{2}$ should be and how small $\epsilon_{2}$ should be in
order to get $\vx$ such that $\vc^{\top}\vx\leq\text{OPT}+\epsilon$.
By Lemma~\ref{lem:weighted_path:duality_gap} and $\norm{\vWeight^{\text{(final)}}}_{1}\leq2m$,
we have
\[
\vc^{\top}\vx_{t_{2}}\leq\text{OPT}+\frac{2m}{t_{2}}.
\]
Also, Lemma \ref{lem: distance gurantee} shows that we have
\[
\normFull{\sqrt{\phi''(\vx_{t_{2}})}\left(\vx^{\text{(final)}}-\vx_{t_{2}}\right)}_{\infty}\leq8\epsilon_{2}.
\]
Using $\min_{\vy}\sqrt{\phi''(\vy)}\geq\frac{1}{U}$, we have $\normFull{\vx^{\text{(final)}}-\vx_{t_{2}}}_{\infty}\leq8\epsilon_{2}U$
and hence our choice of $t_{2}$ and $\epsilon_{2}$ yields
\[
\vc^{\top}\vx^{\text{(final)}}\leq\text{OPT}+\frac{2m}{t_{2}}+8\epsilon_{2}U^{2}\leq\text{OPT}+\epsilon.
\]
\end{proof}
\begin{thm}
\label{thm:LPSolve_detailed_dual} Let $x_{0}\in\Omega\defeq\{\ma^{\top}x=b,x\geq0\}$
for non-degenerate $\ma\in\R^{m\times n}$. There is an algorithm
that finds $y\in\R^{n}$ with $\ma y\leq c$ and $b^{\top}y\geq\max_{\ma y\leq c}b^{\top}y-\epsilon$
with constant probability in 
\[
O\left(\sqrt{n}\log^{13}m\cdot\log(\frac{mU}{\epsilon})\cdot\mathcal{T}_{w}\right)\text{ work and }O\left(\sqrt{n}\log^{13}m\cdot\log(\frac{mU}{\epsilon})\cdot\mathcal{T}_{d}\right)\text{ depth}
\]
where $U\defeq\max\{\text{diam}(\Omega),\norm{\vc}_{\infty},\norm{1/x_{0}}_{\infty}\}$,
$\text{diam}(\Omega)$ is the diameter of $\Omega$, and $\mathcal{T}_{w}$
and $\mathcal{T}_{d}$ is the work and depth needed to compute $(\ma^{\top}\md\ma)^{-1}\vq$
for input positive diagonal matrix $\md$ and vector $\vq$.
\end{thm}

\begin{proof}
Use Algorithm $\code{LPSolve}$ to solve the linear program $\min_{\ma^{\top}x=b,x\geq0}c^{\top}x$.
Following the proof of Theorem~\ref{thm:LPSolve_detailed} and using
$\phi_{i}(x_{i})=-\log x_{i}$ for all $i$, we can find $x$ and
$w$ such that $\delta_{t}(x,w)\leq\frac{1}{2}$ with $t=(\frac{n}{\epsilon})^{O(1)}$
in the time same work and depth as Theorem~\ref{thm:LPSolve_detailed}.
Further, \eqref{eq:centrality_equivalence} shows that $\eta=\left(\ma_{x}^{\top}\mWeight^{-1}\ma_{x}\right)^{-1}\ma_{x}^{\top}\frac{\grad_{x}f_{t}(\vx,\vWeight)}{\vWeight\sqrt{\vphi''(\vx)}}$
satisfies
\[
\normFull{\frac{\grad_{x}f_{t}(\vx,\vWeight)-\ma\veta}{\vWeight\sqrt{\vphi''(\vx)}}}_{\infty}\leq1.
\]
We will prove that $y=\frac{\eta}{t}$ has the desired properties.
Since $\phi''_{i}(x)=x_{i}^{-2}$, we have that
\[
\normFull{\mw^{-1}\mx\left(tc-\frac{w}{x}-\ma\veta\right)}_{\infty}\leq1.
\]
In particular, we have $(\ma y)_{i}\leq c_{i}-\frac{w_{i}}{tx_{i}}+\frac{w_{i}}{tx_{i}}\leq c_{i}$
for all $i\in[n]$. Similarly, we have that $(\ma y)_{i}\geq c_{i}-\frac{w_{i}}{tx_{i}}-\frac{w_{i}}{tx_{i}}$.
Hence, we have
\[
b^{\top}y=x^{\top}\ma y\geq c^{\top}x-\frac{2}{t}\sum_{i\in[m]}w_{i}\geq c^{\top}x-\frac{3n}{t}
\]
and picking $t=\frac{3n}{\epsilon}$ gives the result.
\end{proof}

\subsection{Minimum Cost Maximum Flow}

\label{sec:app:mincost}

Here we show how to use the interior point method of the previous
Section~\ref{sec:lp_alg} to solve the maximum flow problem and the
minimum cost flow problem and thereby prove Theorem~\ref{thm:maxflow}.
Formally, the maximum flow and minimum cost flow problems \cite{daitch2008faster}
is as follows. Let $G=(V,E)$ be a connected directed graph where
each edge $e\in E$ has capacity $c_{e}>0$. We call $x\in\R^{E}$
a \emph{$s$-$t$ flow} for $s,t\in V$ if $x_{e}\in[0,c_{e}]$ for
all $e$ in $E$ and for each vertex $v\notin\{s,t\}$ the amount
of flow entering $v$, i.e. $\sum_{e=(a,v)\in E}f_{e}$ equals the
amount of flow leaving $v$, i.e. $\sum_{e=(v,b)\in E}f_{e}$. The
value of $s$-$t$ flow is the amount of flow leaving $s$ (or equivalently,
entering $t$). The \emph{maximum flow problem} is to compute a $s$-$t$
flow of maximum value. In the \emph{minimum cost maximum flow problem}
there are costs $q_{e}\in\R$ on each edge $e\in E$ and the goal
is to compute a maximum $s$-$t$ flow of minimum cost, $\sum_{e\in E}q_{e}f_{e}=q^{\top}f$. 

Since the minimum cost flow problem includes the maximum flow problem,
we focus on this general formulation. The problem can be written as
the following linear program
\[
\min_{\vzero\leq\vx\leq\vc}\vq^{\top}\vx\text{ such that }\ma\vx=Fe_{t}
\]
where $F$ is the maximum flow value, $e_{t}\in\R^{|V\setminus\{s\}|}$
is an indicator vector of size $|V|-1$ that is non-zero at vertices
$t$ and $\ma$ is a $\left|V\backslash\{s\}\right|\times\left|E\right|$
matrix such that for each edge $e$, we have $\ma_{e_{\text{head}},e}=1$
and $\ma_{e_{\text{tail}},e}=-1$. In order words, the constraint
$\ma x=Fe_{t}$ requires the flow to satisfies the flow conversation
at all vertices except $s$ and $t$ and requires it flows $F$ unit
of flow into $t$ (and therefore $F$ out of $s$). We assume $c_{e}$
and $q_{e}$ are integer and $M$ be the maximum absolute value of
$c_{e}$ and $q_{e}$.

Note that $\rank\left(\ma\right)=|V|-1$ because the graph is connected
and hence our algorithm takes only $\tilde{O}(\sqrt{|V|}\log(U/\epsilon))$
iterations to compute an $\epsilon$-approximate solution to this
linear program. However, to solve minimum cost maximum flow with this
we need to bound $U/\epsilon$, compute $F$, and turn the approximate
solution into an exact minimum cost maximum flow. While there are
many ways to deal with this issue we consider a different linear program
formulation below related to \cite{daitch2008faster}.
\begin{lem}
\label{lem:flowreduction} Given a directed graph $G=(V,E)$ with
integral costs $q\in\Z^{E}$ and capacities $c\in\Z_{\geq0}^{E}$
with $\|q\|_{\infty}\leq M$ and $\|c\|_{\infty}\leq M$ in linear
time we can find a new integral cost vector $\tilde{q}\in\Z^{E}$
with $\|\tilde{q}\|_{\infty}\leq\tilde{M}\defeq8|E|^{2}M^{3}$ such
that the following modified linear program 
\begin{align*}
\min\qquad & \tilde{q}^{\top}\vx+\lambda(\onesVec^{\top}\vy+\onesVec^{\top}\vz)-2n\tilde{M}F\\
\text{ subject to}\qquad & \ma\vx+\vy-\vz=Fe_{t}\\
 & 0\leq x_{i}\leq c_{i},\\
 & 0\leq y_{i}\leq4|V|M,\\
 & 0\leq z_{i}\leq4|V|M,\\
 & 0\leq F\leq2|V|M
\end{align*}
with $\lambda=440|E|^{4}\tilde{M}^{2}M^{3}$ satisfies the following
conditions with constant probability:

\begin{enumerate}
\item $F=|V|M$, $\vx=\frac{c}{2}$, $\vy=2|V|M\onesVec-(\ma\frac{c}{2})^{-}+Fe_{t}$,
$z=2|V|M\onesVec+(\ma\frac{c}{2})^{+}$ is an interior point of the
linear program.
\item Given any feasible $(\vx,\vy,\vz)$ with cost value within $\frac{1}{12M}$
of the optimum. Then, one can find an exact minimum cost maximum $s$-$t$
flow for graph $G$ with costs $q$ and capacities $c$ in $O(|E|)$
work and $O(1)$ depth.
\item The linear system of the linear program can be solve in nearly linear
time, i.e. for any positive diagonal matrix $\ms$ and vector $b$,
it takes 
\[
O\left(|E|\log^{4}|V|\log(|V|/\eta)\right)\text{ work and }O\left(\log^{6}|V|\log(|V|/\eta)\right)\text{ depth}
\]
to find $x$ such that
\begin{equation}
\norm{x-\mvar L^{-1}b}_{\mvar L}\leq\eta\norm x_{\mvar L}\label{eq:max_flow_acc}
\end{equation}
where $\mvar L=[\ \ma\ |\ \mi\ |\ -\mi\ |\ -e_{t}\ ]\ms[\ \ma\ |\ \mi\ |\ -\mi\ |\ -e_{t}\ ]^{\top}.$
\end{enumerate}
\end{lem}

\begin{proof}
By Lemma~\cite[Lemma 3.13]{daitch2008faster}, if we add the cost
of every edge by a number uniformly at random from $\{\frac{1}{4|E|^{2}M^{2}},\frac{2}{4|E|^{2}M^{2}},\cdots,\frac{2|E|M}{4|E|^{2}M^{2}}\}$.
Then with probability at least $1/2$, the new problem has an unique
solution and this solution is a solution for the original problem.
Applying this reduction and scaling the problem back to integral,
we obtain the new cost vector such that the solution is unique. 

For 1) Note that $|V|M\leq y_{i}\leq3|V|M$ and $|V|M\leq z_{i}\leq3|V|M$.
So, $(x,y,z)$ is an interior point.

For 2) Let $\mathtt{Val}=\tilde{q}^{\top}x-2|V|\tilde{M}F+\lambda\left(\onesVec^{\top}\vy+\onesVec^{\top}\vz\right)$
be the objective value of $(x,y,z)$. Let $\mathtt{Opt}$ be the objective
value given by the minimum cost maximum flow. 

First, we prove the the total excess demand $\onesVec^{\top}\vy+\onesVec^{\top}\vz$
is small. Since $0\leq F\leq2|V|M$ and $|q^{\top}x|\leq|E|\tilde{M}M$,
we have that
\begin{equation}
|\tilde{q}^{\top}x-2|V|\tilde{M}F|\leq5|E|^{2}\tilde{M}M\label{eq:maxflow_val_bound}
\end{equation}
for both the algorithm and for the optimum flow. By assumption, we
know that $\mathtt{Val}\leq\mathtt{Opt}+\frac{1}{12M}$, we have that
\[
\lambda(\onesVec^{\top}\vy+\onesVec^{\top}\vz)\leq11|E|^{2}\tilde{M}M.
\]
Using $\lambda=440|E|^{4}\tilde{M}^{2}M^{3}$, we have that the total
excess demand $\onesVec^{\top}\vy+\onesVec^{\top}\vz\leq\epsilon\defeq\frac{1}{40|E|^{2}\tilde{M}M^{2}}$.

To route back the excess demand, we first scale the vector $x,y,z$
and $F$ by a $1-\epsilon$ factor. Then, we create a spanning tree
at $s$. At every vertex $v$, we route the excess demand from $v$
back to the source $s$ in the tree. To route one unit of excess demand
at $v$, we pay at most $\sum_{i\in P_{v}}|\tilde{q}_{i}|$ where
$P_{v}$ is the path from $s$ to $v$ on the tree. Since $\sum_{i\in P_{v}}|\tilde{q}_{i}|\leq|V|\tilde{M}$,
the cost we pay for routing is at most the potential decrease in the
term $\lambda\left(\onesVec^{\top}\vy+\onesVec^{\top}\vz\right)$.
So the objective value of this new flow is at most 
\[
(1-\epsilon)\mathtt{Val}\leq\mathtt{Val}+5\epsilon|E|^{2}\tilde{M}M\leq\mathtt{Val}+\frac{1}{12M}\leq\mathtt{Opt}+\frac{1}{6M}
\]
where we used $\mathtt{Val}\leq5|E|^{2}\tilde{M}M$ due to \eqref{eq:maxflow_val_bound}.
Since we scale the vectors by $1-\epsilon$ factor, the flow is feasible. 

Due to the routing above, we can assume the flow $x$ has no excess
demand with $\mathtt{Val}\leq\mathtt{Opt}+\frac{1}{6M}$. However,
the flow $x$ may not be the maximum flow. Imagine now, we send the
extra flow from $s$ to $t$ to make $x$ maximum. For every unit
we send, we decrease the objective by at least $|V|\tilde{M}$ due
to the term $-2|V|\tilde{M}F$ and the fact that the cost of that
unit of flow is at most $|V|\tilde{M}$. Since $\mathtt{Val}\leq\mathtt{Opt}+\frac{1}{6M}$,
we can send at most $\frac{1}{6M\cdot|V|\tilde{M}}$ amount of extra
flow. We call this new $x$ as $\tilde{x}$ and we let $\mathtt{\widetilde{Val}}$
be the objective value for this $\tilde{x}$. Note that the procedure
above only decrease the objective value. So, we have again $\mathtt{\widetilde{Val}}\leq\mathtt{Opt}+\frac{1}{6M}$.
Finally, we note that $\tilde{x}$ is a weighted combination of maximum
flow from $s$ to $t$. Since the minimum cost solution is unique
and the cost are integral, the combined weight contributed by non-minimum-cost
flow is at most $\frac{1}{6M}$. Since the flow is bounded by $M$,
we know $\tilde{x}$ is at most $\frac{1}{6}$ far from the minimum
cost solution for all edges. 

For the total runtime, note that both the step $\tilde{x}$ and the
step of routing excess demand cannot be omitted because it does not
change the flow for every edge by more than $1/6$. So, we can simply
round every number to nearest integer, which takes linear work and
constant depth.

For Part 3, $\mvar L$ is symmetric diagonally dominant. The result
follows from \cite[Theorem 9.2]{lee2015sparsified}.
\end{proof}
Using the reduction mentioned above, one can obtain the promised minimum
cost flow algorithm. (Further, using techniques from \cite{daitch2008faster}
this can be generalized to solving lossy flow problems.)

\maxflow*
\begin{proof}
Using the reduction (Lemma~\ref{lem:flowreduction}) and Theorem~\ref{thm:LPSolve_detailed},
we get an algorithm of minimum cost flow by solving
\[
O\left(\sqrt{|V|}\log^{13}|E|\cdot\log(|V|M)\cdot\mathcal{T}_{w}\right)\text{ work and }O\left(\sqrt{|V|}\log^{13}|E|\cdot\log(|V|M)\cdot\mathcal{T}_{d}\right)\text{ depth}
\]
where $\mathcal{T}_{w}$ and $\mathcal{T}_{d}$ are the work and the
depth of solving linear systems. It is known that for interior point
methods, we only need to solve linear system with accuracy $\eta=\frac{1}{m^{O(1)}}$
($\eta$ defined in \eqref{eq:max_flow_acc}) because each step of
interior point method only need to decrease the centrality $\delta_{t}$
by a constant factor. Hence, Lemma \ref{lem:flowreduction} shows
that each linear system takes
\[
\mathcal{T}_{w}=O\left(|E|\log^{4}|V|\log(|V|)\right)\text{ work and }\mathcal{T}_{d}=O\left(\log^{6}|V|\log(|V|)\right)\text{ depth}.
\]
Hence, we have the result.
\end{proof}

\subsection{Computable Nearly Universal Barrier}

\label{sec:barrier_algorithm}

Here we show how to combine the results of the preceding sections
to obtain our main results on a polynomial time computable nearly-universal
self-concordant barrier. We first provide and proof Theorem~\ref{thm:lewisbarriercomp},
a generalization of Theorem~\ref{thm:lewisbarriercomp}, and then
show Theorem~\ref{thm:lewisbarriercomp} as a special case. The results
of this section use following result regarding computing Lewis weights
proved in Appendix~\ref{sec:weights_full:computing}.

\begin{restatable}[Exact Weight Computation]{thm}{lewisweightexact}\label{thm:lewisexactfull}
Let $\ma\in\R^{m\times n}$ be non-degenerate matrix and let $\epsilon\in(0,1)$
and $p\in(0,\infty)$. For all $w^{(0)}\in\dWeights$ with $\normInf{w_{(0)}^{-1}(\lpweight(\ma)-w^{(0)})}\leq\frac{p}{20(p+2)}$,
the algorithm\textbf{ $\code{computeExactWeight}(\ma,p,w^{(0)},\epsilon)$}
(Algorithm~\ref{alg:exact_weight}) can be implemented to return
$w$ such that $\norm{\lpweight(\ma)^{-1}(\lpweight(\ma)-w)}_{\infty}\leq\epsilon$
in $O(mn^{\omega-1}(p+p^{-1})\log(n(1+\frac{1}{p})\epsilon^{-1}))$
work and $O((p+p^{-1})\log(m)\log(n(1+\frac{1}{p})\epsilon^{-1}))$
depth.

Without $w^{(0)}$, the algorithm $\code{computeInitialWeight}(\ma,p,\epsilon)$
(Algorithm~\ref{alg:initialweight}) can be implemented to achieve
the same guarantee with $O(mn^{\omega-(1/2)}(p+p^{-1})^{2}\log(\frac{m}{n})\log(n\epsilon^{-1}(p+p^{-1})))$
work and $O((p+p^{-1})^{2}\log(\frac{m}{n})\log(m)\log(n\epsilon^{-1}(p+p^{-1})))$
depth.

\end{restatable}
\begin{thm}
\label{thm:lewisbarriercomputefull} Let $\interior=\{x\,:\,\ma x>b\}$
denote the interior of non-empty polytope for non-degenerate $\ma\in\R^{m\times n}$.
There is an $O(n\log^{5}m)$-self concordant barrier $\psi$ defined
using $\ell_{q}$ Lewis weight with $q=\Theta(\log m)$ (See~\eqref{eq:LS_barrier})
satisfying
\[
\ma_{x}^{\top}\mw_{x}\ma_{x}^{\top}\preceq\nabla^{2}\psi(x)\preceq(q+1)\ma_{x}^{\top}\mw_{x}\ma_{x}^{\top}
\]
where $\ma_{x}=\mDiag(\ma x-b)$ and $w_{x}$ is the $\ell_{q}$ Lewis
weight of the matrix $\ma_{x}$. Furthermore, we can compute or update
the $w_{x}$, $\nabla\psi(x)$ and $\nabla^{2}\psi(x)$ as follows:
\begin{itemize}
\item Initial Weight: For any $x\in\R^{n}$, we can compute a vector $\widetilde{w}_{x}$
such that $(1-\epsilon)w_{x}\leq\widetilde{w}_{x}\leq(1+\epsilon)w_{x}$
in $O(mn^{\omega-(1/2)}\cdot\log^{3}m\cdot\log(m/\epsilon))$-work
and $O(\sqrt{n}\cdot\log^{4}m\cdot\log(m/\epsilon))$-depth.
\item Update Weight and Compute Gradient/Hessian: Given a vector $\widetilde{w}_{x}$
such that $\widetilde{w}_{x}=(1\pm\frac{1}{100})w_{x}$, for any $y$
with $\|x-y\|_{\ma_{x}^{\top}\mw_{x}\ma_{x}^{\top}}\leq\frac{c}{\log^{2}m}$
with some small constant $c>0$, we can compute $\widetilde{w}_{y}$,
$v$ and $\mh$ such that $\widetilde{w}_{y}=(1\pm\epsilon)w_{y},$
\[
\|v-\nabla\psi(x)\|_{\nabla^{2}\psi(x)^{-1}}\leq\epsilon\text{ and }(1-\epsilon)\nabla^{2}\psi(x)\preceq\mh\preceq(1+\epsilon)\nabla^{2}\psi(x)
\]
in $O(mn^{\omega-1}\cdot\log m\cdot\log(m/\epsilon))$-work and $O(\log^{2}m\cdot\log(m/\epsilon))$-depth.
\end{itemize}
\end{thm}

\begin{proof}
Theorem \ref{thm:LS_barrier_sc} with $q=\log m$ shows that there
is such a barrier function $\psi$ that is $O(n\log^{5}m)$. Lemma~~\ref{lem:p_derivatives}
shows that
\[
\grad\psi(x)=-\ma_{x}^{\top}\sigma_{x}\text{ and }\hess\psi(x)=\ma_{x}^{\top}\mSigma_{x}^{1/2}(\mi+\mn_{x})\mSigma_{x}^{1/2}\ma_{x}.
\]
Note that $\sigma_{x}=\lqweight(\ma_{x})$ and we can compute $\tilde{\sigma}$
such that $\tilde{\sigma}\in(1\pm\frac{\epsilon}{\sqrt{n}})\sigma_{x}$
in $O\left(mn^{\omega-\frac{1}{2}}\cdot\log^{3}m\cdot\log\frac{m}{\epsilon}\right)$
work and $O\left(n^{1/2}\cdot\log^{4}m\cdot\log\frac{m}{\epsilon}\right)$
depth using Theorem~\ref{thm:lewisexactfull}.

For the update version, we let $w_{x}$ and $w_{y}$ be the Lewis
weight corresponding to $x$ and $y$. Picking $c$ to be small enough
constant, Lemma \ref{lem:w_lipschitz} shows that $\|w_{y}^{-1}(w_{x}-w_{y})\|_{\infty}\leq O(c)\leq\frac{q}{20(q+2)}$.
Hence, Theorem \ref{thm:lewisexactfull} shows that we can compute
$w_{x}$ with $O(mn^{\omega-1}\cdot\log m\cdot\log(m/\epsilon))$-work
and $O(\log^{3}m\cdot\log(m/\epsilon))$-depth in this case.

For the gradient, with the approximate Lewis weight, Lemma \ref{lem:force}
shows that
\[
\|\nabla\psi(x)+\ma_{x}^{\top}\tilde{\sigma}\|_{\nabla^{2}\psi(x)^{-1}}\leq\frac{\epsilon}{\sqrt{n}}\|\nabla\psi(x)\|_{\nabla^{2}\psi(x)^{-1}}\leq\epsilon.
\]

For the Hessian, we recall that $\mn_{x}=2\mNormProjLap_{x}(\mi-(1-\frac{2}{q})\mNormProjLap_{x})^{-1}$
and $\mNormProjLap_{x}\defeq\mNormProjLap(\mSigma_{x}^{\frac{1}{2}-\frac{1}{q}}\ma_{x})$.
Following calculations in Lemma~\ref{lem:change_of_n} and Lemma~\ref{lem:stable_hess},
one can check that replacing $\sigma_{x}$ by $(1\pm\epsilon)\sigma_{x}$
in the formula of $\hess\psi(x)$ (via $\mn_{x}$ and $\mNormProjLap_{x}$)
only changes the matrix $\hess\psi(x)$ multiplicatively by $\pm\epsilon\log^{O(1)}m$.
Hence, we can compute it again in the same work and depth.
\end{proof}
Leveraging this this theorem we prove, Theorem~\ref{thm:lewisbarriercomp}.
\begin{proof}[Proof of Theorem~\ref{thm:lewisbarriercomp}]
 This theorem is a specialization of Theorem~\ref{thm:lewisbarriercomputefull}.
\end{proof}

\section{Acknowledgments}

We thank Yan Kit Chi, Michael B. Cohen, Jonathan A. Kelner, Aleksander
M\k{a}dry, Richard Peng, and Nisheeth Vishnoi for helpful conversations.
This work was partially supported by NSF awards CCF-0843915 and CCF-1111109,
NSF Graduate Research Fellowship (grant no. 1122374), Hong Kong RGC
grant 2150701, CCF-1749609, CCF-1740551, DMS-1839116, CCF-1844855,
and a Microsoft Research Faculty Fellowship. Part of this work was
done while both authors were visiting the Simons Institute for the
Theory of Computing, UC Berkeley.

\bibliographystyle{plain}
\bibliography{main}

\begin{thebibliography}{10}

\bibitem{abernethy2016faster}
Jacob Abernethy and Elad Hazan.
\newblock Faster convex optimization: Simulated annealing with an efficient
  universal barrier.
\newblock In {\em International Conference on Machine Learning}, pages
  2520--2528, 2016.

\bibitem{AdilKPS19}
Deeksha Adil, Rasmus Kyng, Richard Peng, and Sushant Sachdeva.
\newblock Iterative refinement for $\ell_p$-norm regression.
\newblock In {\em Proceedings of the Thirtieth Annual {ACM-SIAM} Symposium on
  Discrete Algorithms, {SODA} 2019, San Diego, California, USA, January 6-9,
  2019}, pages 1405--1424, 2019.

\bibitem{anstreicher96}
Kurt~M. Anstreicher.
\newblock Volumetric path following algorithms for linear programming.
\newblock {\em Math. Program.}, 76:245--263, 1996.

\bibitem{bourgain1989approximation}
Jean Bourgain, Joram Lindenstrauss, and V~Milman.
\newblock Approximation of zonoids by zonotopes.
\newblock {\em Acta mathematica}, 162(1):73--141, 1989.

\bibitem{BubeckE15}
S{\'{e}}bastien Bubeck and Ronen Eldan.
\newblock The entropic barrier: a simple and optimal universal self-concordant
  barrier.
\newblock In {\em Proceedings of The 28th Conference on Learning Theory, {COLT}
  2015, Paris, France, July 3-6, 2015}, page 279, 2015.

\bibitem{cohen2014solving}
Michael~B Cohen, Rasmus Kyng, Gary~L Miller, Jakub~W Pachocki, Richard Peng,
  Anup~B Rao, and Shen~Chen Xu.
\newblock Solving sdd linear systems in nearly m log 1/2 n time.
\newblock In {\em Proceedings of the forty-sixth annual ACM symposium on Theory
  of computing}, pages 343--352. ACM, 2014.

\bibitem{cohen2015uniform}
Michael~B Cohen, Yin~Tat Lee, Cameron Musco, Christopher Musco, Richard Peng,
  and Aaron Sidford.
\newblock Uniform sampling for matrix approximation.
\newblock In {\em Proceedings of the 2015 Conference on Innovations in
  Theoretical Computer Science}, pages 181--190. ACM, 2015.

\bibitem{CohenLS18}
Michael~B. Cohen, Yin~Tat Lee, and Zhao Song.
\newblock Solving linear programs in the current matrix multiplication time.
\newblock {\em CoRR}, abs/1810.07896, 2018.

\bibitem{CohenP15}
Michael~B. Cohen and Richard Peng.
\newblock $\ell_p$ row sampling by lewis weights.
\newblock In {\em Proceedings of the Forty-Seventh Annual {ACM} on Symposium on
  Theory of Computing, {STOC} 2015, Portland, OR, USA, June 14-17, 2015}, pages
  183--192, 2015.

\bibitem{cole1988parallel}
Richard Cole.
\newblock Parallel merge sort.
\newblock {\em SIAM Journal on Computing}, 17(4):770--785, 1988.

\bibitem{daitch2008faster}
Samuel~I Daitch and Daniel~A Spielman.
\newblock Faster approximate lossy generalized flow via interior point
  algorithms.
\newblock In {\em Proceedings of the 40th annual ACM symposium on Theory of
  computing}, pages 451--460. ACM, 2008.

\bibitem{dantzig1951maximization}
George~B Dantzig.
\newblock Maximization of a linear function of variables subject to linear
  inequalities.
\newblock {\em New York}, 1951.

\bibitem{deza2006central}
Antoine Deza, Eissa Nematollahi, Reza Peyghami, and Tam{\'a}s Terlaky.
\newblock The central path visits all the vertices of the klee--minty cube.
\newblock {\em Optimisation Methods and Software}, 21(5):851--865, 2006.

\bibitem{deza2008good}
Antoine Deza, Eissa Nematollahi, and Tam{\'a}s Terlaky.
\newblock How good are interior point methods? klee--minty cubes tighten
  iteration-complexity bounds.
\newblock {\em Mathematical Programming}, 113(1):1--14, 2008.

\bibitem{drineas2012fast}
Petros Drineas, Malik Magdon-Ismail, Michael~W Mahoney, and David~P Woodruff.
\newblock Fast approximation of matrix coherence and statistical leverage.
\newblock {\em Journal of Machine Learning Research}, 13(Dec):3475--3506, 2012.

\bibitem{even1975network}
Shimon Even and R~Endre Tarjan.
\newblock Network flow and testing graph connectivity.
\newblock {\em SIAM journal on computing}, 4(4):507--518, 1975.

\bibitem{freund_weighted}
RobertM. Freund.
\newblock Projective transformations for interior-point algorithms, and a
  superlinearly convergent algorithm for the w-center problem.
\newblock {\em Mathematical Programming}, 58(1-3):385--414, 1993.

\bibitem{GoldbergRao}
Andrew~V. Goldberg and Satish Rao.
\newblock Beyond the flow decomposition barrier.
\newblock {\em J. ACM}, 45(5):783--797, 1998.

\bibitem{gonzaga1992path}
Clovis~C Gonzaga.
\newblock Path-following methods for linear programming.
\newblock {\em SIAM review}, 34(2):167--224, 1992.

\bibitem{john1948extremum}
Fritz John.
\newblock Extremum problems with inequalities as subsidiary conditions, studies
  and essays presented to r. courant on his 60th birthday, january 8, 1948,
  187--204.
\newblock 1948.

\bibitem{karmarkar1984new}
Narendra Karmarkar.
\newblock A new polynomial-time algorithm for linear programming.
\newblock In {\em Proceedings of the sixteenth annual ACM symposium on Theory
  of computing}, pages 302--311. ACM, 1984.

\bibitem{k1973}
Alexander~V Karzanov.
\newblock On finding a maximum flow in a network with special structure and
  some applications.
\newblock {\em Matematicheskie Voprosy Upravleniya Proizvodstvom}, 5:81--94,
  1973.

\bibitem{Kelner2013}
Jonathan~A. Kelner, Lorenzo Orecchia, Aaron Sidford, and Zeyuan~Allen Zhu.
\newblock {A Simple, Combinatorial Algorithm for Solving SDD Systems in
  Nearly-Linear Time}.
\newblock January 2013.

\bibitem{khachiyan1980polynomial}
Leonid~G Khachiyan.
\newblock Polynomial algorithms in linear programming.
\newblock {\em USSR Computational Mathematics and Mathematical Physics},
  20(1):53--72, 1980.

\bibitem{khachiyan1996rounding}
Leonid~G Khachiyan.
\newblock Rounding of polytopes in the real number model of computation.
\newblock {\em Mathematics of Operations Research}, 21(2):307--320, 1996.

\bibitem{KoutisMP10}
Ioannis Koutis, Gary~L. Miller, and Richard Peng.
\newblock Approaching optimality for solving {SDD} systems.
\newblock In {\em Proceedings of the 51st Annual Symposium on Foundations of
  Computer Science}, 2010.

\bibitem{KMP11}
Ioannis Koutis, Gary~L. Miller, and Richard Peng.
\newblock A nearly-m log n time solver for sdd linear systems.
\newblock In {\em Foundations of Computer Science (FOCS), 2011 IEEE 52nd Annual
  Symposium on}, pages 590 --598, oct. 2011.

\bibitem{KyngLPSS16}
Rasmus Kyng, Yin~Tat Lee, Richard Peng, Sushant Sachdeva, and Daniel~A.
  Spielman.
\newblock Sparsified cholesky and multigrid solvers for connection laplacians.
\newblock In {\em Proceedings of the 48th Annual {ACM} {SIGACT} Symposium on
  Theory of Computing, {STOC} 2016, Cambridge, MA, USA, June 18-21, 2016},
  pages 842--850, 2016.

\bibitem{kyng2016approximate}
Rasmus Kyng and Sushant Sachdeva.
\newblock Approximate gaussian elimination for laplacians-fast, sparse, and
  simple.
\newblock In {\em 2016 IEEE 57th Annual Symposium on Foundations of Computer
  Science (FOCS)}, pages 573--582. IEEE, 2016.

\bibitem{lee2015sparsified}
Yin~Tat Lee, Richard Peng, and Daniel~A Spielman.
\newblock Sparsified cholesky solvers for sdd linear systems.
\newblock {\em arXiv preprint arXiv:1506.08204}, 2015.

\bibitem{lee2013ACDM}
Yin~Tat Lee and Aaron Sidford.
\newblock Efficient accelerated coordinate descent methods and faster
  algorithms for solving linear systems.
\newblock In {\em The 54th Annual Symposium on Foundations of Computer Science
  (FOCS)}, 2013.

\bibitem{lsInteriorPoint}
Yin~Tat Lee and Aaron Sidford.
\newblock Path finding i: Solving linear programs with $\backslash$\~{} o
  (sqrt(rank)) linear system solves.
\newblock {\em arXiv preprint arXiv:1312.6677}, 2013.

\bibitem{lsMaxflow}
Yin~Tat Lee and Aaron Sidford.
\newblock Path finding ii: An$\backslash$\~{} o (m sqrt (n)) algorithm for the
  minimum cost flow problem.
\newblock {\em arXiv preprint arXiv:1312.6713}, 2013.

\bibitem{leeS14}
Yin~Tat Lee and Aaron Sidford.
\newblock Path-finding methods for linear programming : Solving linear programs
  in \~o(sqrt(rank)) iterations and faster algorithms for maximum flow.
\newblock In {\em 55th Annual {IEEE} Symposium on Foundations of Computer
  Science, {FOCS} 2014, 18-21 October, 2014, Philadelphia, PA, {USA}}, pages
  424--433, 2014.

\bibitem{LeeS15}
Yin~Tat Lee and Aaron Sidford.
\newblock Efficient inverse maintenance and faster algorithms for linear
  programming.
\newblock In {\em {IEEE} 56th Annual Symposium on Foundations of Computer
  Science, {FOCS} 2015, Berkeley, CA, USA, 17-20 October, 2015}, pages
  230--249, 2015.

\bibitem{lee2018universal}
Yin~Tat Lee and Man-Chung Yue.
\newblock Universal barrier is $ n $-self-concordant.
\newblock {\em arXiv preprint arXiv:1809.03011}, 2018.

\bibitem{Lewis1978}
D.~Lewis.
\newblock Finite dimensional subspaces of $l_\{p\}$.
\newblock {\em Studia Mathematica}, 63(2):207--212, 1978.

\bibitem{li2012iterative}
Mu~Li, Gary~L Miller, and Richard Peng.
\newblock Iterative row sampling.
\newblock 2012.

\bibitem{madryFlow}
Aleksander Madry.
\newblock Navigating central path with electrical flows: from flows to
  matchings, and back.
\newblock In {\em Proceedings of the 54th Annual Symposium on Foundations of
  Computer Science}, 2013.

\bibitem{mahoney11survey}
Michael~W. Mahoney.
\newblock Randomized algorithms for matrices and data.
\newblock {\em Foundations and Trends in Machine Learning}, 3(2):123--224,
  2011.

\bibitem{megiddo_weighted}
Nimrod Megiddo.
\newblock Pathways to the optimal set in linear programming.
\newblock In Nimrod Megiddo, editor, {\em Progress in Mathematical
  Programming}, pages 131--158. Springer New York, 1989.

\bibitem{mut2013tight}
Murat Mut and Tam{\'a}s Terlaky.
\newblock A tight iteration-complexity upper bound for the mty
  predictor-corrector algorithm via redundant klee-minty cubes.
\newblock 2013.

\bibitem{nelson2012osnap}
Jelani Nelson and Huy~L Nguy{\^e}n.
\newblock Osnap: Faster numerical linear algebra algorithms via sparser
  subspace embeddings.
\newblock {\em arXiv preprint arXiv:1211.1002}, 2012.

\bibitem{nematollahi2008redundant}
Eissa Nematollahi and Tam{\'a}s Terlaky.
\newblock A redundant klee--minty construction with all the redundant
  constraints touching the feasible region.
\newblock {\em Operations Research Letters}, 36(4):414--418, 2008.

\bibitem{nematollahi2008simpler}
Eissa Nematollahi and Tam{\'a}s Terlaky.
\newblock A simpler and tighter redundant klee--minty construction.
\newblock {\em Optimization Letters}, 2(3):403--414, 2008.

\bibitem{Nesterov2003}
Yu~Nesterov.
\newblock {\em {Introductory Lectures on Convex Optimization: A Basic Course}},
  volume~I.
\newblock 2003.

\bibitem{nesterov1989self}
Yu~Nesterov and Arkadi Nemirovskiy.
\newblock {\em Self-concordant functions and polynomial-time methods in convex
  programming}.
\newblock USSR Academy of Sciences, Central Economic \& Mathematic Institute,
  1989.

\bibitem{nesterov1997self}
Yu~E Nesterov and Michael~J Todd.
\newblock Self-scaled barriers and interior-point methods for convex
  programming.
\newblock {\em Mathematics of Operations research}, 22(1):1--42, 1997.

\bibitem{Nesterov1994}
Yurii Nesterov and Arkadii~Semenovich Nemirovskii.
\newblock {\em Interior-point polynomial algorithms in convex programming},
  volume~13.
\newblock Society for Industrial and Applied Mathematics, 1994.

\bibitem{peng2013efficient}
Richard Peng and Daniel~A Spielman.
\newblock An efficient parallel solver for sdd linear systems.
\newblock {\em arXiv preprint arXiv:1311.3286}, 2013.

\bibitem{renegar1988polynomial}
James Renegar.
\newblock A polynomial-time algorithm, based on newton's method, for linear
  programming.
\newblock {\em Mathematical Programming}, 40(1-3):59--93, 1988.

\bibitem{schrijver2003combinatorial}
Alexander Schrijver.
\newblock {\em Combinatorial optimization: polyhedra and efficiency},
  volume~24.
\newblock Springer, 2003.

\bibitem{spielmanS08sparsRes}
Daniel~A Spielman and Nikhil Srivastava.
\newblock Graph sparsification by effective resistances.
\newblock {\em SIAM Journal on Computing}, 40(6):1913--1926, 2011.

\bibitem{spielman2004nearly}
Daniel~A Spielman and Shang-Hua Teng.
\newblock Nearly-linear time algorithms for graph partitioning, graph
  sparsification, and solving linear systems.
\newblock In {\em Proceedings of the thirty-sixth annual ACM symposium on
  Theory of computing}, pages 81--90. ACM, 2004.

\bibitem{todd1994scaling}
Michael~J Todd.
\newblock Scaling, shifting and weighting in interior-point methods.
\newblock {\em Computational Optimization and Applications}, 3(4):305--315,
  1994.

\bibitem{vaidya89convexSet}
Pravin~M. Vaidya.
\newblock A new algorithm for minimizing convex functions over convex sets
  (extended abstract).
\newblock In {\em FOCS}, pages 338--343, 1989.

\bibitem{vaidya1989speeding}
Pravin~M Vaidya.
\newblock Speeding-up linear programming using fast matrix multiplication.
\newblock In {\em Foundations of Computer Science, 1989., 30th Annual Symposium
  on}, pages 332--337. IEEE, 1989.

\bibitem{vaidya1987speeding}
Pravin~M Vaidya.
\newblock An algorithm for linear programming which requires o (((m+ n) n 2+(m+
  n) 1.5 n) l) arithmetic operations.
\newblock {\em Mathematical Programming}, 47(1-3):175--201, 1990.

\bibitem{vaidya90parallel}
Pravin~M. Vaidya.
\newblock Reducing the parallel complexity of certain linear programming
  problems (extended abstract).
\newblock In {\em FOCS}, pages 583--589, 1990.

\bibitem{vaidya1996new}
Pravin~M Vaidya.
\newblock A new algorithm for minimizing convex functions over convex sets.
\newblock {\em Mathematical Programming}, 73(3):291--341, 1996.

\bibitem{vaidya1993technique}
Pravin~M Vaidya and David~S Atkinson.
\newblock A technique for bounding the number of iterations in path following
  algorithms.
\newblock {\em Complexity in Numerical Optimization}, pages 462--489, 1993.

\bibitem{vavasis1996primal}
Stephen~A Vavasis and Yinyu Ye.
\newblock A primal-dual interior point method whose running time depends only
  on the constraint matrix.
\newblock {\em Mathematical Programming}, 74(1):79--120, 1996.

\bibitem{vempala2010recent}
Santosh~S Vempala.
\newblock Recent progress and open problems in algorithmic convex geometry.
\newblock In {\em LIPIcs-Leibniz International Proceedings in Informatics},
  volume~8. Schloss Dagstuhl-Leibniz-Zentrum fuer Informatik, 2010.

\bibitem{williams2012matrixmult}
Virginia~Vassilevska Williams.
\newblock Multiplying matrices faster than coppersmith-winograd.
\newblock In {\em Proceedings of the forty-fourth annual ACM symposium on
  Theory of computing}, pages 887--898. ACM, 2012.

\bibitem{woodruff2014sketching}
David~P Woodruff et~al.
\newblock Sketching as a tool for numerical linear algebra.
\newblock {\em Foundations and Trends{\textregistered} in Theoretical Computer
  Science}, 10(1--2):1--157, 2014.

\bibitem{ye2011interior}
Yinyu Ye.
\newblock {\em Interior point algorithms: theory and analysis}, volume~44.
\newblock John Wiley \& Sons, 2011.

\end{thebibliography}

\appendix

\section{Projection Matrices, Leverages Scores, and $\mathrm{\log\det}$\label{sec:app:lemmas}}

In this section, we prove various properties of projection matrices,
leverage scores, and the logarithm of the determinant that we use
throughout the paper. 

First we provide the following theorem which gives various properties
of projection matrices and leverage scores.
\begin{lem}[Projection Matrices]
\label{lem:tool:projection_matrices} Let $\mProj\in\Rmm$ be an
arbitrary orthogonal projection matrix and let $\mSigma=\mDiag(\mProj)$.
For all $i,j\in[m]$, $\vx,y\in\Rm$ , and $\mx=\mDiag(x)$ we have
\[
\begin{array}{lcl}
(1)\,\,\mSigma_{ii}=\sum_{j\in[m]}\mProj_{ij}^{(2)} & \enspace & (5)\,\,\norm{\mSigma^{-1}\mProj^{(2)}x}_{\infty}\leq\norm x_{\infty}\\
(2)\,\,\mZero\preceq\mProj^{(2)}\preceq\mSigma\preceq\iMatrix,(\text{in particular,}0\leq\mSigma_{ii}\leq1) & \enspace & (6)\,\,\sum_{i\in[m]}\mSigma_{ii}=\rank(\mProj)\\
(3)\,\,\mProj_{ij}^{(2)}\leq\mSigma_{ii}\mSigma_{jj} & \enspace & (7)\,\,\left|y^{\top}\mx\mProj^{(2)}y\right|\leq\norm y_{\mSigma}^{2}\cdot\norm x_{\mSigma}\\
(4)\,\,\norm{\mSigma^{-1}\mProj^{(2)}x}_{\infty}\leq\norm{\vx}_{\mSigma} & \enspace & (8)\,\,\left|y^{\top}\left(\mProj\circ\mProj\mx\mProj\right)y\right|\leq\norm y_{\mSigma}^{2}\cdot\norm x_{\mSigma}\,.
\end{array}
\]
\end{lem}

\begin{proof}
To prove (1), we simply note that by definition of a projection matrix
$\mProj=\mProj\mProj$ and therefore
\[
\mSigma_{ii}=\mProj_{ii}=\cordVec i^{\top}\mProj\cordVec i=\cordVec i^{\top}\mProj\mProj\cordVec i=\sum_{j\in[m]}\mProj_{ij}^{2}=\sum_{j\in[m]}\mProj_{ij}^{(2)}.
\]

To prove (2), we observe that since $\mProj$ is a projection matrix,
all its eigenvalues are either 0 or 1. Therefore, $\mSigma\preceq\iMatrix$
and by (1) $\mSigma-\mProj^{(2)}$ is diagonally dominant. Consequently,
$\mSigma-\mProj^{(2)}\succeq0$. Rearranging terms and using the well
known fact that the Shur product of two positive semi-definite matrices
is positive semi-definite yields (2).

To prove (3), we use $\mProj=\mProj\mProj$, Cauchy-Schwarz, and (1)
to derive
\[
\mProj_{ij}=\sum_{k\in[m]}\mProj_{ik}\mProj_{kj}\leq\sqrt{\left(\sum_{k\in[m]}\mProj_{ik}^{2}\right)\left(\sum_{k\in[m]}\mProj_{kj}^{2}\right)}=\sqrt{\mSigma_{ii}\mSigma_{jj}}\enspace.
\]
Squaring then yields (3).

To prove (4), we note that by the definition of $\mProj^{(2)}$ and
Cauchy-Schwarz, we have
\begin{equation}
\left|\cordVec i^{\top}\mProj^{(2)}\vx\right|=\left|\sum_{j\in[m]}\mProj_{ij}^{(2)}x_{j}\right|\leq\sqrt{\left(\sum_{j\in[m]}\mSigma_{jj}x_{j}^{2}\right)\cdot\sum_{j\in[m]}\frac{\mProj_{ij}^{(4)}}{\mSigma_{jj}}}.\label{eq:lem:proj1}
\end{equation}
Now, by (1) and (3), we know that
\begin{equation}
\sum_{j\in[m]}\frac{\mProj_{ij}^{4}}{\mSigma_{jj}}\leq\sum_{j\in[m]}\frac{\mProj_{ij}^{2}\mSigma_{ii}\mSigma_{jj}}{\mSigma_{jj}}=\mSigma_{ii}\sum_{j\in[m]}\mProj_{ij}^{2}=\mSigma_{ii}^{2}.\label{eq:lem:proj2}
\end{equation}
Since $\norm{\vx}_{\mSigma}\defeq\sqrt{\sum_{j\in[m]}\mSigma_{jj}\vx_{j}^{2}}$,
combining (\ref{eq:lem:proj1}) and (\ref{eq:lem:proj2}) yields $\left|\cordVec i^{\top}\mProj^{(2)}\vx\right|\leq\mSigma_{ii}\norm{\vx}_{\mSigma}$
as desired.

To prove (5), we note that
\[
\left|\cordVec i^{\top}\mProj^{(2)}\vx\right|=\left|\sum_{j\in[m]}\mProj_{ij}^{(2)}x_{j}\right|\leq\sum_{j\in[m]}\mProj_{ij}^{(2)}\left|x_{j}\right|=\mSigma_{ii}\norm x_{\infty}
\]

To prove (6), we note that all the eigenvalues of $\mProj$ are either
0 or 1 and $\sum_{i\in[m]}\Sigma_{ii}=\tr(\mProj)$.

To prove (7), we apply (4) and Cauchy Schwarz to show
\[
\left|y^{\top}\mx\mProj^{(2)}y\right|=\left|\sum_{i\in[m]}x_{i}\cdot y_{i}\indicVec i^{\top}\mProj^{(2)}y\right|\leq\sum_{i\in[m]}|x_{i}|\cdot|y_{i}|\cdot\mSigma_{ii}\cdot\norm y_{\mSigma}\leq\norm x_{\mSigma}\norm y_{\mSigma}\norm y_{\mSigma}\,.
\]

To prove (8), we note that by Cauchy Schwarz
\begin{align*}
\left|y^{\top}\left(\mProj\circ\mProj\mx\mProj\right)y\right| & =\left|\sum_{i,j\in[m]}y_{i}y_{j}\mProj_{ij}\left(\sum_{k\in[m]}\mProj_{ik}\mProj_{jk}x_{k}\right)\right|\\
 & \leq\sqrt{\left(\sum_{i,j\in[m]}|y_{i}|\cdot|y_{j}|\cdot\mProj_{ij}^{2}\right)\cdot\left(\sum_{i,j\in[m]}|y_{i}|\cdot|y_{j}|\cdot\left(\sum_{k\in[m]}\mProj_{ik}\mProj_{jk}x_{k}\right)^{2}\right)}\,.
\end{align*}
Letting $|x|$ and $|y|$ be the vectors whose entries are the absolute
values of the entries of $x$ and $y$ we respectively, see that by
(2) we have 
\[
\sum_{i,j\in[m]}|y_{i}|\cdot|y_{j}|\cdot\mProj_{ij}^{2}=\norm{|y|}_{\mProj^{(2)}}^{2}\leq\norm{|y|}_{\mSigma}^{2}=\norm y_{\mSigma}^{2}
\]
and 
\[
\sum_{i,j\in[m]}|y_{i}|\cdot|y_{j}|\cdot\left(\sum_{k\in[m]}\mProj_{ik}\mProj_{jk}x_{k}\right)^{2}=\sum_{i,j\in[m]}\left(\sum_{k\in[m]}\left[\mProj_{ik}\sqrt{|y_{i}||x_{k}|}\right]\left[\mProj_{jk}\sqrt{|y_{j}||x_{k}|}\right]\right)^{2}\,.
\]
Applying Cauchy Schwarz twice then yields that
\begin{align*}
\sum_{i,j\in[m]}|y_{i}|\cdot|y_{j}|\cdot\left(\sum_{k\in[m]}\mProj_{ik}\mProj_{jk}x_{k}\right)^{2} & \leq\left(\sum_{i,k\in[m]}|y_{i}|\mProj_{ik}^{2}|x_{k}|\right)^{2}=\left(|y|^{\top}\mProj^{(2)}|x|\right)^{2}\\
 & \leq\norm{|y|}_{\mProj^{(2)}}^{2}\norm{|x|}_{\mProj^{(2)}}^{2}\leq\norm{|y|}_{\mSigma}^{2}\norm{|x|}_{\mSigma}^{2}=\norm y_{\mSigma}^{2}\norm x_{\mSigma}^{2}\,.
\end{align*}
Combining these inequalities than yields the desired bound on $\left|y^{\top}\left(\mProj\circ\mProj\mx\mProj\right)y\right|$.
\end{proof}
Next, we derive various matrix calculus formulas relating the projection
matrix with the log determinant. We start by computing the derivative
of the \emph{volumetric barrier function}, $f(w)\defeq\log\det(\ma^{\top}\mw\ma)$.
\begin{lem}[Derivative of Volumetric Barrier]
\label{lem:deriv:log_det} For full rank matrix $\ma\in\Rnm$ let
$f:\Rpm\rightarrow\R$ be given by $f(w)\defeq\log\det(\ma^{\top}\mw\ma)$.
For any $w\in\Rpm$, we have $\nabla f(w)=\mw^{-1}\sigma(\mw^{\frac{1}{2}}\ma)$.
\end{lem}

\begin{proof}
Using the derivative of $\log\det$, we have that for all $i\in[m]$
\[
\frac{\partial f(w)}{\partial w_{i}}=\tr\left[\left(\ma^{\top}\mw\ma\right)^{-1}\frac{\partial}{\partial w_{i}}\left(\ma^{\top}\mw\ma\right)\right]=\tr\left[\left(\ma^{\top}\mw\ma\right)^{-1}\ma^{\top}\cordVec i\cordVec i^{\top}\ma\right]=w_{i}^{-1}\sigma\left(\mw^{\frac{1}{2}}\ma\right)_{i}\,.
\]
\end{proof}
Next we bound the rate of change of entries of the projection matrix.
\begin{lem}[Derivative of Projection Matrix]
\label{lem:deriv:proj} Given full rank $\ma\in\Rnm$ and $w\in\R_{>0}^{m}$
we have
\[
D_{w}\mProj(\mw\ma)[h]=\mDelta\mProj(\mw\ma)+\mProj(\mw\ma)\mDelta-2\mProj(\mw\ma)\mDelta\mProj(\mw\ma)
\]
where $\mw=\mDiag(w)$ and $\mDelta=\mDiag(h/w)$. In particular,
we have that 
\[
D_{w}\sigma(\mw\ma)[h]=2\mLambda(\mw\ma)\mw^{-1}h.
\]
\end{lem}

\begin{proof}
Note that
\[
\mProj(\mw\ma)=\mw\ma(\ma^{\top}\mw^{2}\ma)^{-1}\ma^{\top}\mw.
\]
Using the derivative of matrix inverse, we have that
\begin{align*}
D_{w}\mProj(\mw\ma)[h]= & \mh\ma(\ma^{\top}\mw^{2}\ma)^{-1}\ma^{\top}\mw+\mw\ma(\ma^{\top}\mw\ma)^{-1}\ma^{\top}\mh\\
 & -2\mw\ma(\ma^{\top}\mw^{2}\ma)^{-1}\ma^{\top}\mh\mw\ma(\ma^{\top}\mw^{2}\ma)^{-1}\ma^{\top}\mw\\
= & \mDelta\mProj+\mProj\mDelta-2\mProj\mDelta\mProj\,.
\end{align*}
Consequently, 
\begin{align*}
D_{w}\sigma(\mw\ma)_{i}[h] & =[D_{w}\mProj(\mw\ma)[h]]_{ii}=2\mDelta_{ii}\mProj_{ii}-2(\mProj\mDelta\mProj)_{ii}\\
 & =2\frac{\sigma_{i}h_{i}}{w_{i}}-2\sum_{j\in[m]}\mProj_{ij}^{2}\mDelta_{jj}=2\left[\left(\mSigma-\mProj^{(2)}\right)(h/w)\right]_{i}=2(\mLambda(h/w))_{i}\,.
\end{align*}
\end{proof}
In the following lemma we provide a general formula regarding the
derivative of a function that appears throughout the paper.
\begin{lem}[Potential Function Derivative]
\label{lem:all_volumetric_derivs} For non-degenerate $\ma\in\R^{m\times n}$
and $q>0$ with $q\neq2$ let $p(x,w)\defeq\ln\det(\ma_{x}^{\top}\mw^{1-\frac{2}{q}}\ma_{x})$
for all $x\in\R^{n}$ with $\ma x>b$ and all $w\in\R_{>0}^{m}$ where
$\ma_{x}\defeq\ms_{x}^{-1}\ma$, $\ms_{x}=\mDiag(\ma x-b)$, and $w\in\mDiag(w)$.
Then, the following hold
\begin{align*}
\nabla_{x}p(x,w) & =-2\ma_{x}^{\top}\sigma_{x,w}, & \nabla_{w}p(x,w) & =c_{q}\mw^{-1}\sigma_{x,w},\\
\nabla_{xx}^{2}p(x,w) & =\ma_{x}^{\top}(2\mSigma_{x,w}+4\mLambda_{x,w})\ma_{x}, & \nabla_{ww}^{2}p(x,w) & =-c_{q}\mw^{-1}(\mSigma_{x,w}-c_{q}\mLambda_{x,w})\mw^{-1},\text{ and}\\
\nabla_{xw}^{2}p(x,w) & =-2c_{q}\ma_{x}^{\top}\mLambda_{x,w}\mw^{-1}
\end{align*}
where, $c_{q}\defeq1-\frac{2}{q}$, $\sigma_{x,w}\defeq\sigma(\mw^{\frac{1}{2}-\frac{1}{q}}\ma_{x})$,
$\mSigma_{x,w}\defeq\mSigma(\mw^{\frac{1}{2}-\frac{1}{q}}\ma_{x})$,
and $\mLambda_{x,w}\defeq\mLambda(\mw^{\frac{1}{2}-\frac{1}{q}}\ma_{x})$.
\end{lem}

\begin{proof}
To simplify the calculations, throughout this proof we overload notation
and let $p(s,w)\defeq\ln\det(\ma^{\top}\ms^{-1}\mw^{1-\frac{2}{q}}\ms^{-1}\ma)$
where $\ms\defeq\mDiag(s)$ and let $s_{x}\defeq\ma x-b$. Since $\ms_{x}=\mDiag(s_{x})$
and $s_{x}$ is a linear transformation of $x$ this implies that
the derivatives with respect to $x$ to follow immediately from the
derivatives with respect to $s$.

For $\nabla_{x}p(x,w)$, Lemma~\ref{lem:deriv:log_det} shows that
\begin{equation}
\frac{\partial}{\partial s_{i}}p(s_{x},w)=s_{i}^{2}w_{i}^{-1+\frac{2}{q}}[\sigma_{x,w}]_{i}\cdot(-2s_{i}^{-3}w_{i}^{1-\frac{2}{q}})=-2\frac{[\sigma_{x,w}]_{i}}{s_{i}}\,.\label{eq:partials_1}
\end{equation}
Therefore $\nabla_{x}p(x,w)=-2\ma_{x}^{\top}\sigma_{x,w}$ by chain
rule.

For $\nabla_{w}p(x,w)$, Lemma~\ref{lem:deriv:log_det} and chain
rule shows that
\begin{equation}
\frac{\partial}{\partial w_{i}}p(x,w)=s_{i}^{2}w_{i}^{-1+\frac{2}{q}}[\sigma_{x,w}]_{i}\cdot\left(s_{i}^{-2}\left(1-\frac{2}{q}\right)w_{i}^{-\frac{2}{q}}\right)=c_{q}\frac{[\sigma_{x,w}]_{i}}{w_{i}}\,.\label{eq:partials_2}
\end{equation}
Therefore $\nabla_{w}p(x,w)=c_{q}\mw^{-1}\sigma_{x,w}$.

For $\nabla_{xx}p(x,w)$, the formula for $\frac{\partial}{\partial s_{i}}p(s_{x},w)$
given by (\ref{eq:partials_1}) and Lemma~\ref{lem:deriv:proj} yields
\begin{align*}
\frac{\partial^{2}}{\partial s_{i}\partial s_{j}}p(s_{x},w) & =2\frac{[\sigma_{x,w}]_{i}}{s_{i}^{2}}\indicVec{i=j}-4\frac{[\mLambda_{x,w}]_{ij}}{s_{i}}\cdot s_{j}w_{j}^{-\frac{1}{2}+\frac{1}{q}}\cdot(-s_{j}^{-2}w_{j}^{\frac{1}{2}-\frac{1}{q}})\\
 & =2\frac{[\sigma_{x,w}]_{i}}{s_{i}^{2}}\indicVec{i=j}+4\frac{[\mLambda_{x,w}]_{ij}}{s_{i}s_{j}}\,.
\end{align*}
Therefore, $\nabla_{xx}^{2}p(x,w)=\ma_{x}^{\top}(2\mSigma_{x,w}+4\mLambda_{x,w})\ma_{x}$.

For $\nabla_{ww}p(x,w)$, the formula for $\frac{\partial}{\partial w_{i}}p(x,w)$
given by (\ref{eq:partials_2}) and Lemma~\ref{lem:deriv:proj} yields
\begin{align*}
\frac{\partial^{2}}{\partial w_{i}\partial w_{j}}p(x,w) & =-c_{q}\frac{[\sigma_{x,w}]_{i}}{w_{i}^{2}}\indicVec{i=j}+2c_{q}\frac{[\mLambda_{x,w}]_{ij}}{w_{i}}\cdot s_{j}w_{j}^{-\frac{1}{2}+\frac{1}{q}}\cdot\left(\frac{1}{2}-\frac{1}{p}\right)(s_{j}^{-1}w_{j}^{-\frac{1}{2}-\frac{1}{q}})\\
 & =-c_{q}\frac{[\sigma_{x,w}]_{i}}{w_{i}^{2}}\indicVec{i=j}+c_{q}^{2}\frac{[\mLambda_{x,w}]_{ij}}{w_{i}w_{j}}\,.
\end{align*}
Therefore $\nabla_{ww}^{2}p(x,w)=-c_{q}\mw^{-1}(\mSigma_{x,w}-c_{q}\mLambda_{x,w})\mw^{-1}$.

For $\nabla_{xw}p(x,w)$, the formula for $\frac{\partial}{\partial s_{i}}p(s_{x},w)$
given by (\ref{eq:partials_1}) and Lemma~\ref{lem:deriv:proj} yield
\[
\frac{\partial^{2}}{\partial s_{i}\partial w_{j}}p(s_{x},w)=-4\frac{[\mLambda_{x,w}]_{ij}}{s_{i}}\cdot s_{j}w_{j}^{-\frac{1}{2}+\frac{1}{q}}\cdot\left(\frac{1}{2}-\frac{1}{p}\right)(s_{j}^{-1}w_{j}^{-\frac{1}{2}-\frac{1}{q}})=-2c_{q}\frac{[\mLambda_{x,w}]_{ij}}{s_{i}w_{j}}\,.
\]
Therefore, $-2c_{q}\ma_{x}^{\top}\mLambda_{x,w}\mw^{-1}$ by chain
rule. 
\end{proof}
\begin{lem}
\label{lem:newton_step}For any vector $v$, any positive vector $w$
and matrix $\ma$, we have that
\[
\argmin_{\ma^{\top}x=0}v^{\top}x+\frac{1}{2}\|x\|_{w}^{2}=x_{*}\defeq-\mw^{-1}v+\mw^{-1}\ma(\ma^{\top}\mw^{-1}\ma)^{-1}\ma^{\top}\mw^{-1}v.
\]
\end{lem}

\begin{proof}
Let $f(x)\defeq v^{\top}x+\frac{1}{2}\norm x_{w}^{2}$. Note that
$\grad f(x)=v+\mw x$ and consequently, $x\in\ker(\ma^{\top})$ is
optimal if and only if $v+\mw x\perp\ker(\ma^{\top}x)$, i.e. $v+\mw x\in\im(\ma)$,
and $\ma^{\top}x=0$. Since $\ma^{\top}x_{*}=0$ and $w+\mw x=\ma(\ma^{\top}\mw^{-1}\ma)^{-1}\ma^{\top}\mw^{-1}v\in\im(\ma^{\top})$
the result follows.
\end{proof}

\section{Lewis Weight Computation\label{sec:weights_full:computing}}

Here, we describe how to efficiently compute approximations to Lewis
weights and ultimately prove Theorem~\ref{thm:lewisweightexactapproxfull}
and Theorem~\ref{thm:lewisexactfull} (the Lewis weight computation
results claimed and used in Section~\ref{sec:algorithms}). We achieve
our results by a combination of a number of technical tools, including
projected gradient descent (for computing Lewis weights exactly in
Section~\ref{sec:lewis_exact} given a good initial weight), the
Johnson-Lindenstrauss lemma (for computing Lewis weights approximately
in Section~\ref{sec:approx_weight} given a good initial weight),
and \emph{homotopy methods} (for computing initial weights and completing
the proofs of the main theorems in Section~\ref{sec:lewis_initial}). 

Throughout the remainder of this section we let $\ma\in\R^{m\times n}$
denote an arbitrary non-degenerate matrix and $p\in(0,\infty)$ with
$p\neq2$. Further we let $\lpweight\defeq\lpweight(\ma)$ and $\mw_{p}\defeq\mDiag(\lpweight)$.

\subsection{Exact Computation}

\label{sec:lewis_exact}

Since Lewis weight can be found by the minimizer of a convex optimization
problem (Lemma \ref{lem:unique_lewis}), we can use the gradient descent
method directly to minimize $\mathcal{V}_{p}^{\ma}(w)$. Indeed, in
this section we show how applying the gradient descent method in a
carefully scaled space allows us to compute the weight to good accuracy
in $\otilde(\mathrm{poly}(p))$ iterations. This results makes two
assumptions to compute the weight: (1) we compute the gradient of
$\mathcal{V}_{p}^{\ma}(w)$ exactly and (2) we are given a weight
that is not too far from the true weight. In the remaining subsection
we show how to address these issues. 

First we state the following theorem regarding gradient descent method
we use in our analysis. This theorem shows that if we take repeated
projected gradient steps then we can achieve linear convergence up
to bounds on how much the Hessian of the function changes over the
domain of interest.
\begin{thm}[Simple Constrained Minimization for Twice Differentiable Function]
\label{thm:constrained_minimization} Let $\mh$ be a positive definite
matrix and $Q\subseteq\Rm$ be a convex set. Let $f:Q\rightarrow\R$
be a twice differentiable function. Suppose that there are constants
$0\leq\mu\leq L$ such that for all $\vx\in Q$ we have $\mu\cdot\mh\preceq\nabla^{2}f(\vx)\preceq L\cdot\mh$.
For any $\vx^{(0)}\in Q$ and any $k\geq0$ if we apply the update
rule
\[
\vx^{(k+1)}=\argmin_{\vx\in Q}f(x^{(k)})+\nabla f(\vx^{(k)})^{\top}(\vx-\vx^{(k)})+\frac{L}{2}\norm{\vx-\vx^{(k)}}_{\mh}^{2}
\]
then it follows that 
\[
\norm{\vx^{(k)}-\vx^{*}}_{\mh}^{2}\leq\left(1-\frac{\mu}{L}\right)^{k}\norm{\vx^{(0)}-\vx^{*}}_{\mh}^{2}.
\]
\end{thm}

To apply Theorem~\ref{thm:constrained_minimization} to compute Lewis
weight, we first recall from Lemma \ref{lem:unique_lewis} that the
Lewis weight $w_{p}(\ma)$ is the unique minimizer of the convex problem,
$\min_{w_{i}\geq0}\mathcal{V}_{p}^{\ma}(w)+\sum_{i\in[m]}w_{i}$.
Therefore, to apply this result we first need to show that there is
a region around the optimal point $\lpweight$ such that the Hessian
of $\mathcal{V}_{p}^{\ma}(w)$ does not change too much. 
\begin{lem}[Hessian Approximation]
\label{lem:Hessian_of_weight} If $w\in\R_{\geq0}^{m}$ satisfies
$\normInf{\mw^{-1}(\lpweight-w)}\leq\frac{p}{8(p+2)}$ for the matrix
$\mw\defeq\mDiag(w)$ then
\[
\min\left\{ \frac{1}{2},\frac{1}{p}\right\} \mw^{-1}\preceq\nabla^{2}\mathcal{V}_{p}^{\ma}(w)\preceq\max\left\{ 2,\frac{4}{p}\right\} \mw^{-1}.
\]
\end{lem}

\begin{proof}
Using that $\normInf{\mw^{-1}(\lpweight-w)}\leq\frac{p}{8(p+2)}$
and letting $\mv=\mDiag(\lpweight)$, we have
\[
\mSigma_{w}=\mSigma\left(\mw^{\frac{1}{2}-\frac{1}{p}}\ma\right)\preceq\frac{(1+\frac{p}{8(p+2)})^{\left|1-\frac{2}{p}\right|}}{(1-\frac{p}{8(p+2)})^{\left|1-\frac{2}{p}\right|}}\mSigma\left(\mv^{\frac{1}{2}-\frac{1}{p}}\ma\right)\preceq\frac{3}{2}\mv\preceq2\mw
\]
and
\[
\mSigma_{w}=\mSigma\left(\mw^{\frac{1}{2}-\frac{1}{p}}\ma\right)\succeq\frac{(1-\frac{p}{8(p+2)})^{\left|1-\frac{2}{p}\right|}}{(1+\frac{p}{8(p+2)})^{\left|1-\frac{2}{p}\right|}}\mSigma\left(\mv^{\frac{1}{2}-\frac{1}{p}}\ma\right)\succeq\frac{3}{4}\mv\succeq\frac{1}{2}\mw.
\]
The result therefore follows immediately from Lemma~\ref{lem:lewis_potential_derivatives}.
\end{proof}
Combining Theorem \ref{thm:constrained_minimization} and Lemma \ref{lem:Hessian_of_weight},
we get the following algorithm to compute the weight function using
the exact computation of the gradient of $\mathcal{V}_{p}^{\ma}$. 
\begin{lem}
\label{lem:weight_iterative} Let $w^{(0)}\in\dWeight$ such that
$\normInf{\mw_{(0)}^{-1}(\lpweight-w^{(0)})}\leq r$ where $r\defeq\frac{p}{20(p+2)}$.
For all $k\geq0$ let $L\defeq\max\{4,\frac{8}{p}\}$ and
\begin{equation}
w^{(k+1)}\defeq\code{median}\left(\left(1-r\right)w^{(0)},w^{(k)}-\frac{1}{L}\left(w^{(0)}-\frac{w^{(0)}}{w^{(k)}}\sigma\left(\mw_{(k)}^{\frac{1}{2}-\frac{1}{p}}\ma\right)\right),\left(1+r\right)w^{(0)}\right)\label{eq:explicit_formula_w}
\end{equation}
where $[\code{median}\left(\vx,\vy,\vz\right)]_{i}$ is the median
of $x_{i}$, $y_{i}$ and $z_{i}$ for all $i\in[m]$. For all $k$,
we have
\[
\norm{w^{(k)}-\lpweight}_{\mlpweight^{-1}}^{2}\leq4n\cdot\left(1-\frac{1}{16(\frac{p}{2}+\frac{2}{p})}\right)^{k}\norm{\mw_{(0)}^{-1}(\lpweight-w^{(0)})}_{\infty}^{2}\,.
\]
\end{lem}

\begin{proof}
Let $Q\defeq\{w\in\Rm~|~\normInf{\mw_{(0)}^{-1}(w-w^{(0)})}\leq r\}$
and $\mw_{(k)}\defeq\mDiag(w^{(k)})$. Applied to $\min_{w_{i}\geq0}\mathcal{V}_{p}^{\ma}(w)+\sum_{i=1}^{m}w_{i}$,
by Lemma~\ref{lem:lewis_potential_derivatives}, iterations of Theorem~\ref{thm:constrained_minimization}
are
\begin{eqnarray*}
w^{(k+1)} & = & \argmin_{w\in Q}\left\langle \onesVec-w_{(k)}^{-1}\sigma\left(\mw_{(k)}^{\frac{1}{2}-\frac{1}{p}}\ma\right),w\right\rangle +\frac{L}{2}\normFull{w-w^{(k)}}_{\mw_{(0)}^{-1}}^{2}\\
 & = & \argmin_{w\in Q}\normFull{w-w^{(k)}+\frac{1}{L}\left(w^{(0)}-\frac{w^{(0)}}{w^{(k)}}\sigma\left(\mw_{(k)}^{\frac{1}{2}-\frac{1}{p}}\ma\right)\right)}_{\mw_{(0)}^{-1}}^{2}.
\end{eqnarray*}
Since the objective function and the constraints are axis-aligned,
we can compute $w^{(k+1)}$ coordinate-wise and we see that this is
the same as in the statement of this lemma. 

To apply Theorem \ref{thm:constrained_minimization}, we note that
$\normInf{\mw_{(0)}^{-1}(\lpweight-w^{(0)})}\leq\frac{p}{20(p+2)}$
implies that any $w\in Q$ satisfies $\normInf{\mw^{-1}(\lpweight-w)}\leq\frac{p}{8(p+2)}$
and hence Lemma~\ref{lem:Hessian_of_weight} shows that
\begin{equation}
\min\left\{ \frac{1}{4},\frac{1}{2p}\right\} \mw_{(0)}^{-1}\preceq\min\left\{ \frac{1}{2},\frac{1}{p}\right\} \mw^{-1}\preceq\nabla^{2}\mathcal{V}(w)\preceq\max\left\{ 2,\frac{4}{p}\right\} \mw^{-1}\preceq\max\left\{ 4,\frac{8}{p}\right\} \mw_{(0)}^{-1}.\label{eq:lip_grad_vw}
\end{equation}
Hence, Theorem \ref{thm:constrained_minimization} and inequality
\eqref{eq:lip_grad_vw} shows that
\[
\norm{w^{(k)}-\lpweight}_{\mw_{(0)}^{-1}}^{2}\leq\left(1-\frac{\min\{\frac{1}{4},\frac{1}{2p}\}}{\max(4,\frac{8}{p})}\right)^{k}\norm{w^{(0)}-\lpweight}_{\mw_{(0)}^{-1}}^{2}\leq\left(1-\frac{1}{16(\frac{p}{2}+\frac{2}{p})}\right)^{k}\norm{w^{(0)}-\lpweight}_{\mw_{(0)}^{-1}}^{2}.
\]
The result follows as $w^{(0)}$ is close to $\lpweight$ multiplicatively
and therefore
\[
\norm{w^{(0)}-\lpweight}_{\mlpweight^{-1}}^{2}\leq\frac{3}{2}\sum_{i\in[n]}w_{i}^{(0)}\cdot\norm{\mw_{(0)}^{-1}(\lpweight-w^{(0)})}_{\infty}^{2}\leq2n\norm{\mw_{(0)}^{-1}(\lpweight-w^{(0)})}_{\infty}^{2}.
\]
\end{proof}
Note that the lemma does not shows that $w^{(k)}$ is a multiplicative
approximation of $\lpweight$. The following lemma shows that we can
use $w^{(k)}$ to get a multiplicative approximation.
\begin{lem}
\label{lem:fix_small_weight} Given $w$ such that $\|\mlpweight^{-1}(\lpweight-w)\|_{\infty}\leq\frac{p}{8(p+2)}$
and that $\|w-\lpweight\|_{\mlpweight^{-1}}\leq\frac{1}{4(1+\frac{2}{p})^{2}\sqrt{n}}$.
Let $\widehat{w}=(\diag(\ma(\ma^{\top}\mw^{1-\frac{2}{p}}\ma)^{-1}\ma^{\top}))^{\frac{2}{p}}$.
Then, we have that 
\[
\normFull{\mlpweight^{-1}(\widehat{w}-\lpweight)}_{\infty}\leq4\left(1+\frac{2}{p}\right)^{2}\sqrt{n}\cdot\normFull{w-\lpweight}_{\mlpweight^{-1}}.
\]
\end{lem}

\begin{proof}
The definition of $\widehat{w}$ is motivated from the equality $\lpweight=(\diag(\ma(\ma^{\top}\mlpweight^{1-\frac{2}{p}}\ma)^{-1}\ma^{\top}))^{\frac{2}{p}}$.
To show $\widehat{w}$ is multiplicative close to $\lpweight$, it
therefore suffices to prove that $\ma^{\top}\mlpweight^{1-\frac{2}{p}}\ma$
is multiplicatively close to $\ma^{\top}\mw^{1-\frac{2}{p}}\ma$.
Note that $(1-\alpha)\ma^{\top}\mlpweight^{1-\frac{2}{p}}\ma\preceq\ma^{\top}\mw^{1-\frac{2}{p}}\ma\preceq(1+\alpha)\ma^{\top}\mlpweight^{1-\frac{2}{p}}\ma$
with
\[
\alpha\leq\tr\left[(\ma^{\top}\mlpweight^{1-\frac{2}{p}}\ma)^{-1}(\ma^{\top}|\mw^{1-\frac{2}{p}}-\mlpweight^{1-\frac{2}{p}}|\ma)\right]=\sum_{i\in[n]}\frac{\sigma_{i}(\mlpweight^{\frac{1}{2}-\frac{2}{p}}\ma)}{[\lpweight]_{i}^{1-2/p}}\left|w_{i}^{1-2/p}-[\lpweight]_{i}^{1-2/p}\right|\,.
\]
Since $\normInf{\mlpweight^{-1}(\lpweight-w)}\leq\frac{p}{8(p+2)}$
we have that for all $i\in[n]$
\[
\left|w_{i}^{1-2/p}-[\lpweight]_{i}^{1-2/p}\right|\leq2\left|1-\frac{2}{p}\right|\left|\frac{w_{i}-[\lpweight]_{i}}{[\lpweight]_{i}^{2/p}}\right|
\]
and therefore, by Cauchy Schwarz and that $\sum_{i\in[n]}\frac{\sigma_{i}(\mlpweight^{\frac{1}{2}-\frac{2}{p}}\ma)^{2}}{[\lpweight]_{i}}=\sum_{i\in[n]}\sigma_{i}(\mlpweight^{\frac{1}{2}-\frac{2}{p}}\ma)=n$
we have
\begin{align*}
\alpha & \leq2\left|1-\frac{2}{p}\right|\sum_{i\in[n]}\frac{\sigma_{i}(\mlpweight^{\frac{1}{2}-\frac{2}{p}}\ma)}{[\lpweight]_{i}^{1-2/p}}\left|\frac{w_{i}-[\lpweight]_{i}}{[\lpweight]_{i}^{2/p}}\right|\\
 & \leq\left|1-\frac{2}{p}\right|\sqrt{\sum_{i\in[n]}\frac{\sigma_{i}(\mlpweight^{\frac{1}{2}-\frac{2}{p}}\ma)^{2}}{[\lpweight]_{i}}}\sqrt{\sum_{i\in[n]}\frac{(w_{i}-[\lpweight]_{i})^{2}}{[\lpweight]_{i}}}=2\left|1-\frac{2}{p}\right|\sqrt{n}\cdot\delta\,.
\end{align*}
The result follows from $\lpweight=(\diag(\ma(\ma^{\top}\mlpweight^{1-\frac{2}{p}}\ma)^{-1}\ma^{\top}))^{\frac{2}{p}}$,
that $(1-2|1-(2/p)|\delta\sqrt{n})^{-2/p}\geq1-4(1+\frac{2}{p})^{2}\sqrt{n}\delta$,
and that $(1+2|1-(2/p)|\delta\sqrt{n})^{-2/p}\leq1+4(1+\frac{2}{p})^{2}\sqrt{n}\delta$.
\end{proof}
Combining Lemma~\ref{lem:weight_iterative} and Lemma~\ref{lem:fix_small_weight}
yields the following Theorem~\ref{thm:lewisweightexact}, the main
result of this section on weight computation.

\begin{algorithm2e}[H]

\label{alg:exact_weight}

\caption{$\ensuremath{\ensuremath{w=\code{computeExactWeight}(\ma,p,w^{(0)},\varepsilon)}}$}

\SetAlgoLined

Let $T=\left\lceil 32(\frac{p}{2}+\frac{2}{p})\log(8n(1+\frac{2}{p})\epsilon^{-1})\right\rceil $,
$r=\frac{p}{20(p+2)}$, and $L=\max\{4,\frac{8}{p}\}$

\For{$k=1,\cdots,T-1$}{

$w^{(k+1)}=\code{median}\left(\left(1-r\right)w^{(0)},w^{(k)}-\frac{1}{L}\left(w^{(0)}-\frac{w^{(0)}}{w^{(k)}}\sigma\left(\mw_{(k)}^{\frac{1}{2}-\frac{1}{p}}\ma\right)\right),\left(1+r\right)w^{(0)}\right)$

}

\textbf{Output:} $(\diag(\ma(\ma^{\top}\mw_{(T)}^{1-\frac{2}{p}}\ma)^{-1}\ma^{\top}))^{\frac{2}{p}}$
for $\mWeight_{(T)}=\mDiag(w^{(T)})$.

\end{algorithm2e}

\begin{restatable}[Exact Weight Updates]{thm}{lewisweightexact}\label{thm:lewisweightexact}
For all $\epsilon\in(0,1)$ and $w^{(0)}\in\dWeight$ with $\normInf{w_{(0)}^{-1}(\lpweight(\ma)-w^{(0)})}\leq\frac{p}{20(p+2)}$
the algorithm $\code{computeExactWeight}(\ma,p,w^{(0)},\epsilon)$
(Algorithm~\ref{alg:exact_weight}) outputs $w\in\dWeight$ with
$\norm{\lpweight(\ma)^{-1}(\lpweight(\ma)-w^{(0)}}_{\infty}\leq\epsilon$
in $O((p+\frac{1}{p})\log(n(1+\frac{1}{p})\epsilon^{-1}))$ iterations,
where each iteration involves computing $\sigma\left(\mv\ma\right)$
for diagonal matrix $\mv$ and extra linear time work and $O(1)$
depth.

\end{restatable}
\begin{proof}
This result follows immediately from Lemma~\ref{lem:weight_iterative}
and Lemma~\ref{lem:fix_small_weight}.
\end{proof}

\subsection{Approximate Computation\label{sec:approx_weight}}

Here we show how to modify the algorithm and analysis of the previous
subsection to use approximate leverage scores instead of exact leverage
score in computing gradient. Further, we show how to use the Johnson-Lindenstrauss
lemma to compute approximate leverage scores efficiently using a linear
system solver. Together, these results give us efficient algorithms
for improving the approximation quality of Lewis weights.

To analyze our algorithm, $\code{computeApxWeight}$ (Algorithm~\ref{alg:approx_weight})
given below, we first give a lemma showing that the optimality condition
$\sigma(\mw^{\frac{1}{2}-\frac{1}{p}}\ma)_{i}/w_{i}$ is stable under
changes to $w$.

\begin{algorithm2e}[H]

\label{alg:approx_weight}

\caption{$\ensuremath{\ensuremath{w=\code{computeApxWeight}(\ma,p,w^{(0)},\varepsilon)}}$}

\SetAlgoLined

$L=\max\{4,\frac{8}{p}\}$, $r=\frac{p^{2}(4-p)}{2^{20}}$ and $\delta=\frac{(4-p)\epsilon}{256}$.

Let the number of iterations $T=\left\lceil 80(\frac{p}{2}+\frac{2}{p})\log\left(\frac{pn}{32\epsilon}\right)\right\rceil $.

\For{$j=1,\cdots,T-1$}{

Compute $\vsigma^{(j)}\in\R^{n}$ such that $e^{-\delta}\sigma(\mw_{(j)}^{\frac{1}{2}-\frac{1}{p}}\ma)_{i}\leq\vsigma_{i}^{(j)}\leq e^{\delta}\sigma(\mw_{(j)}^{\frac{1}{2}-\frac{1}{p}}\ma)_{i}$
for all $i\in[n]$.

$w^{(j+1)}=\code{median}\left(\left(1-r\right)w^{(0)},w^{(j)}-\frac{1}{L}\left(w^{(0)}-\frac{w^{(0)}}{w^{(j)}}\vsigma^{(j)}\right),\left(1+r\right)w^{(0)}\right)$.

}

\textbf{Output:} $(\diag(\ma(\ma^{\top}\mw_{(T)}^{1-\frac{2}{p}}\ma)^{-1}\ma^{\top}))^{\frac{2}{p}}$.

\end{algorithm2e}
\begin{lem}
\label{lem:stable_grad} Let $w,v\in\dWeight$ with $w_{i}=e^{\delta_{i}}v_{i}$
for $|\delta_{i}|\leq\delta$ for all $i\in[n]$. Then, for all $i\in[n]$
\[
e^{\frac{2}{p}\delta_{i}-|1-\frac{2}{p}|\delta}\cdot\frac{\sigma(\mw^{\frac{1}{2}-\frac{1}{p}}\ma)_{i}}{w_{i}}\leq\frac{\sigma(\mv^{\frac{1}{2}-\frac{1}{p}}\ma)_{i}}{v_{i}}\leq e^{\frac{2}{p}\delta_{i}+|1-\frac{2}{p}|\delta}\cdot\frac{\sigma(\mw^{\frac{1}{2}-\frac{1}{p}}\ma)_{i}}{w_{i}}.
\]
\end{lem}

\begin{proof}
Note that $v_{i}^{-1}\sigma(\mv^{\frac{1}{2}-\frac{1}{p}}\ma)_{i}=v_{i}^{-\frac{2}{p}}a_{i}^{\top}(\ma^{\top}\mv^{1-\frac{2}{p}}\ma)^{-1}a_{i}$
where $a_{i}$ is the $i$-th row of $\ma$. By the assumptions on
$w,v\in\dWeight$ we have
\[
v_{i}^{-1}\sigma(\mv^{\frac{1}{2}-\frac{1}{p}}\ma)_{i}=e^{\frac{2}{p}\delta_{i}}w_{i}^{-\frac{2}{p}}a_{i}^{\top}(\ma^{\top}\mv^{1-\frac{2}{p}}\ma)^{-1}a_{i}\leq e^{\frac{2}{p}\delta_{i}+|1-\frac{2}{p}|\delta}w_{i}^{-\frac{2}{p}}a_{i}^{\top}(\ma^{\top}\mw^{1-\frac{2}{p}}\ma)^{-1}a_{i}.
\]
and the lower bound on $\sigma(\mv^{\frac{1}{2}-\frac{1}{p}}\ma)_{i}$
follows similarly.
\end{proof}
\begin{thm}[Approximate Weight Computation]
\label{thm:weights_full:approximate_weight} If $p\in(0,4)$ and
$w^{(0)}\in\dWeights$ satisfies $\normInf{w_{(0)}^{-1}(\lpweight(\ma)-w^{(0)})}\leq r$
where $r=\frac{p^{2}(4-p)}{2^{20}}$. For $0<\epsilon<\frac{2}{p}-|1-\frac{2}{p}|$,
the algorithm\textbf{ $\code{computeApxWeight}(\vx,w^{(0)},\varepsilon)$}
returns $w$ such that $\norm{\lpweight(\ma)^{-1}(\lpweight(\ma)-w)}_{\infty}\leq\epsilon$
in $O(p^{-1}\log(np^{-1}\epsilon^{-1}))$ steps. Each step involves
computing $\sigma$ up to $\pm\Theta((4-p)\cdot\epsilon)$ multiplicative
error with some extra linear time work.
\end{thm}

\begin{proof}
Consider an execution of \textbf{$\code{computeApxWeight}(x,w^{(0)},\varepsilon)$}
where there is no error in computing leverages scores, i.e. $\vsigma^{(j)}=\sigma(\mw_{(j)}^{\frac{1}{2}-\frac{1}{p}}\ma)$,
and let $\vv^{(j)}$ denote the $w$ computed during this idealized
execution of $\code{computeApxWeight}$. We will show that $w^{(j)}$
and $v^{(j)}$ are multiplicatively close.

Suppose that $w_{i}^{(j)}=e^{\delta_{i}^{(j)}}v_{i}^{(j)}$ with $|\delta_{i}^{(j)}|\leq\delta^{(j)}$
for some $\delta^{(j)}\geq0$. Define $\overline{v}^{(j+1)},\overline{w}^{(j+1)}\in\dWeight$
to be $v^{(j+1)}$ and $w_{i}^{(j+1)}$ before taking the median,
i.e. 
\[
\overline{v}_{i}^{(j+1)}=v^{(j)}-\frac{1}{L}\left(w^{(0)}-\frac{w^{(0)}}{v^{(j)}}\sigma(\mv_{(j)}^{\frac{1}{2}-\frac{1}{p}}\ma)\right)\text{ and }\overline{w}_{i}^{(j+1)}=w^{(j)}-\frac{1}{L}\left(w^{(0)}-\frac{w^{(0)}}{w^{(j)}}\sigma^{(j)}\right)
\]
Using $\pm\delta$ to denote a real value with magnitude at most $\delta$
and applying Lemma~\ref{lem:stable_grad} with $v=v^{(j)}$ and $w=w^{(j)}$,
we have
\begin{align}
\overline{w}_{i}^{(j+1)}-\overline{v}_{i}^{(j+1)} & =w^{(j)}-v^{(j)}+\frac{w^{(0)}}{L}\left(\frac{e^{\pm\delta}\sigma(\mw_{(j)}^{\frac{1}{2}-\frac{1}{p}}\ma)}{w^{(j)}}-\frac{\sigma(\mv_{(j)}^{\frac{1}{2}-\frac{1}{p}}\ma)}{v^{(j)}}\right)\nonumber \\
 & =(e^{\delta_{i}^{(j)}}-1)v_{i}^{(j)}+\frac{w^{(0)}}{L}\left(e^{-\frac{2}{p}\delta_{i}^{(j)}\pm|1-\frac{2}{p}|\delta^{(j)}\pm\delta}-1\right)\cdot\frac{\sigma(\mv_{(j)}^{\frac{1}{2}-\frac{1}{p}}\ma)}{v^{(j)}}.\label{eq:w_v_error}
\end{align}
Since $\|\mw_{(0)}^{-1}(w^{(0)}-v^{(j)})\|_{\infty}\leq r$ and that
$\|\mw_{(0)}^{-1}(w^{(0)}-\lpweight(\ma))\|_{\infty}\leq r$, we have
that $\lpweight(\ma)=e^{\pm3r}v^{(j)}$. Lemma \ref{lem:stable_grad}
shows that for $w=w_{p}(\ma)$ we have
\begin{equation}
e^{-3(\frac{2}{p}+|1-\frac{2}{p}|)r}\cdot w_{i}^{-1}\sigma(\mw^{\frac{1}{2}-\frac{1}{p}}\ma)_{i}\leq(v_{i}^{(j)})^{-1}\sigma(\mv_{(j)}^{\frac{1}{2}-\frac{1}{p}}\ma)_{i}\leq e^{3(\frac{2}{p}+|1-\frac{2}{p}|)r}\cdot w_{i}^{-1}\sigma(\mw^{\frac{1}{2}-\frac{1}{p}}\ma)_{i}.\label{eq:w_j_sigma}
\end{equation}
where Using that $w_{i}=\sigma(\mw^{\frac{1}{2}-\frac{1}{p}}\ma)_{i}$,
\eqref{eq:w_j_sigma}, and \eqref{eq:w_v_error}, we have that
\begin{align*}
\overline{w}_{i}^{(j+1)}-\overline{v}_{i}^{(j+1)} & =(e^{\delta_{i}^{(j)}}-1)v_{i}^{(j)}+\frac{w^{(0)}}{L}\left(e^{-\frac{2}{p}\delta_{i}^{(j)}\pm|1-\frac{2}{p}|\delta^{(j)}\pm\delta}-1\right)e^{\pm3(1+\frac{4}{p})r}
\end{align*}
Since $w^{(j+1)}$ and $v^{(j+1)}$ are just truncation of $\overline{w}^{(j+1)}$
and $\overline{v}^{(j+1)}$, we have the same bound for $w_{i}^{(j+1)}-v_{i}^{(j+1)}$.
Using $w_{i}^{(j+1)}=e^{\delta_{i}^{(j+1)}}v_{i}^{(j+1)}$, we get
that
\[
(e^{\delta_{i}^{(j+1)}}-1)v_{i}^{(j+1)}=(e^{\delta_{i}^{(j)}}-1)v_{i}^{(j)}+\frac{w^{(0)}}{L}\left(e^{-\frac{2}{p}\delta_{i}^{(j)}\pm|1-\frac{2}{p}|\delta^{(j)}\pm\delta}-1\right)e^{\pm3(1+\frac{4}{p})r}.
\]
Finally, we note that $v^{(j+1)}=e^{\pm2r}w^{(0)}$ and hence
\[
e^{\delta_{i}^{(j+1)}}-1=e^{\pm4r}(e^{\delta_{i}^{(j)}}-1)+\frac{1}{L}\left(e^{-\frac{2}{p}\delta_{i}^{(j)}\pm|1-\frac{2}{p}|\delta^{(j)}\pm\delta}-1\right)e^{\pm3(2+\frac{4}{p})r}.
\]

Using that $L=\max\{4,\frac{8}{p}\}$, $r=\frac{p^{2}(4-p)}{2^{20}}$,
$\left|\delta_{i}^{(j)}\right|\leq\delta^{(j)}\leq2r$, $\delta\leq\frac{1}{32}$
and $p\leq4$ we obtain
\begin{align}
 & e^{\delta_{i}^{(j+1)}}-1\nonumber \\
= & (e^{\delta_{i}^{(j)}}-1)\pm8r\delta_{i}^{(j)}+\frac{1}{L}\left(e^{-\frac{2}{p}\delta_{i}^{(j)}\pm|1-\frac{2}{p}|\delta^{(j)}\pm\delta}-1\right)\pm\frac{6}{L}\left(2+\frac{4}{p}\right)r\left(\left(1+\frac{4}{p}\right)\delta^{(j)}+\delta\right)\nonumber \\
= & \delta_{i}^{(j)}\pm(\delta^{(j)})^{2}\pm8r\delta_{i}^{(j)}+\frac{1}{L}\left(-\frac{2}{p}\delta_{i}^{(j)}\pm\left|1-\frac{2}{p}\right|\delta^{(j)}\pm2\delta\pm\left(\left(1+\frac{4}{p}\right)\delta^{(j)}\right)^{2}\right)\nonumber \\
 & \enspace\enspace\enspace\pm\frac{6}{L}\left(2+\frac{4}{p}\right)r\left(\left(1+\frac{4}{p}\right)\delta^{(j)}+\delta\right)\nonumber \\
= & \left(1-\frac{2}{pL}\right)\delta_{i}^{(j)}\pm\frac{1}{L}\left|1-\frac{2}{p}\right|\delta^{(j)}\pm\frac{3\delta}{L}\pm2(\delta^{(j)})^{2}\pm8r\delta_{i}^{(j)}\pm\frac{12}{L}\left(1+\frac{4}{p}\right)^{2}\delta^{(j)}r\nonumber \\
= & \left(1-\frac{2}{pL}\right)\delta_{i}^{(j)}\pm\frac{1}{L}\left|1-\frac{2}{p}\right|\delta^{(j)}\pm\frac{3\delta}{L}\pm40\left(1+\frac{4}{p}\right)\delta^{(j)}r\label{eq:exp_delta_complicated}
\end{align}
where we used $e^{x}=1\pm\frac{4}{3}x$ for $|x|\leq\frac{1}{2}$
in the first equality, we used $e^{x}=1+x\pm x^{2}$ for $|x|\leq\frac{1}{2}$
in the second equality.

For the first two terms, we have that
\[
\left|\left(1-\frac{2}{pL}\right)\delta_{i}^{(j)}\pm\frac{1}{L}\left|1-\frac{2}{p}\right|\delta^{(j)}\right|\leq\left(1-\frac{2}{pL}+\frac{1}{L}\left|1-\frac{2}{p}\right|\right)\delta^{(j)}.
\]
Using this, $1+x\leq e^{x}$ and \eqref{eq:exp_delta_complicated}
and $\left|\delta_{i}^{(j)}\right|\leq\delta^{(j)}\leq2r$, we have
\[
\delta^{(j+1)}\leq\left(1-\frac{2}{pL}+\frac{1}{L}\left|1-\frac{2}{p}\right|+40\left(1+\frac{4}{p}\right)r\right)\delta^{(j)}+\frac{3\delta}{L}.
\]
Using our choice of $r$, we have $40(1+\frac{4}{p})r\leq\frac{1}{2L}(\frac{2}{p}-|1-\frac{2}{p}|)$
and hence
\[
\delta^{(j+1)}\leq\left(1-\frac{1}{2L}\left(\frac{2}{p}-\left|1-\frac{2}{p}\right|\right)\right)\delta^{(j)}+\frac{3\delta}{L}
\]
and hence for all $j\in[m]$
\[
\delta^{(j)}\leq\frac{1}{\frac{1}{2L}(\frac{2}{p}-|1-\frac{2}{p}|)}\cdot\frac{3\delta}{L}\leq\frac{8\delta}{\frac{2}{p}-|1-\frac{2}{p}|}\leq\frac{\epsilon}{4}\,.
\]

Applying Lemma~\ref{lem:weight_iterative} and Lemma \ref{lem:fix_small_weight},
and recalling that $k=\left\lceil 80(\frac{p}{2}+\frac{2}{p})\log\left(\frac{pn}{32\epsilon}\right)\right\rceil $,
we have
\begin{align*}
\normFull{\mlpweight^{-1}(\lpweight-w^{(k)})}_{\infty} & \leq\normFull{\mlpweight^{-1}(\lpweight-\vv^{(k)})}_{\infty}+\normFull{\mlpweight^{-1}\left(\vv^{(k)}-w^{(k)}\right)}_{\infty}\\
 & \leq4\left(1+\frac{2}{p}\right)^{2}\sqrt{n}\cdot\normFull{w-\lpweight}_{\mlpweight^{-1}}+2\delta^{(k)}\\
 & \leq4\left(1+\frac{2}{p}\right)^{2}\sqrt{n}\cdot2\sqrt{n}\cdot\left(1-\frac{1}{16(\frac{p}{2}+\frac{2}{p})}\right)^{\frac{k}{2}}\cdot\frac{p}{160}+2\delta^{(k)}\leq\epsilon.
\end{align*}
\end{proof}
Unfortunately, we cannot use the previous lemma directly as computing
$\vsigma$ exactly is too expensive for our purposes. However, in
\cite{spielmanS08sparsRes,drineas2012fast} they showed that we can
compute leverage scores, $\vsigma$, approximately by solving only
polylogarithmically many regression problems (See \cite{mahoney11survey,li2012iterative,woodruff2014sketching,cohen2015uniform}
for more details). These results use the fact that the leverage scores
of the the $i^{th}$ constraint, i.e. $\sigma(\ma)_{i}$ is the $\ellTwo$
length of vector $\ma(\ma^{\top}\ma)^{-1}\ma^{\top}\indicVec i$ and
that by the Johnson-Lindenstrauss Lemma these lengths are persevered
up to multiplicative error if we project these vectors onto certain
random low dimensional subspace. Consequently, to approximate the
$\vsigma$ we first compute the projected vectors and then use it
to approximate $\vsigma$ and hence only need to solve $\otilde(1)$
regression problems. For completeness, we provide an algorithm and
theorem statement below most closely resembling the one from \cite{spielmanS08sparsRes}.
here:

\begin{algorithm2e}[H]

\caption{\textbf{$\vsigma^{\mathrm{(apx)}}=\code{computeLeverageScores}(\ma,\epsilon)$}}

\SetAlgoLined

Let $q^{(j)}$ be $k$ random $\pm1/\sqrt{k}$ vectors of length $m$
with $k=O(\log(m)/\epsilon^{2})$. 

Compute $\vl^{(j)}=(\ma^{\top}\ma)^{-1}\ma^{\top}q^{(j)}$ and $\vp^{(j)}=\ma\vl^{(i)}$.

\textbf{Output:} $\sum_{j=1}^{k}\left(\vp_{i}^{(j)}\right)^{2}$.

\end{algorithm2e}
\begin{lem}
\label{thm:weights_full:leverage-score_sampling} For $\epsilon\in(0,1)$
with probability at least $1-\frac{1}{m^{O(1)}}$\footnote{This is the only place our algorithm uses randomness for general linear
programs. Since we can verify the centrality of central path by computing
leverage score exactly (instead of using this theorem) every $m^{O(1)}$
iterations of interior point method, $1-\frac{1}{m^{O(1)}}$ probability
is high enough even for the case $\epsilon$ is doubly exponentially
small.} the algorithm\textbf{ $\code{computeLeverageScores}$} returns $\vsigma^{\mathrm{(apx)}}$
such that for all $i\in[m]$ , $\left(1-\epsilon\right)\sigma(\ma)_{i}\leq\vsigma_{i}^{\mathrm{(apx)}}\leq\left(1+\epsilon\right)\sigma(\ma)_{i}$,
by solving only $O(\epsilon^{-2}\cdot\log m)$ linear systems.
\end{lem}

In the next section we show how to combine these results to obtain
our main result on approximate weight computation, Theorem~\ref{thm:lewisweightexactapproxfull}.

\subsection{Initial Weight and Final Theorems}

\label{sec:lewis_initial}

Here, we show how to compute an initial weight without having an approximate
weight to help the computation. While we can use the results of the
previous section during the iterations of our linear programming algorithms
(as we have shown that the Lewis weights do not change too quickly)
we still need to design a routine to compute the initial weights.
Here we show that the algorithm $\code{computeInitialWeight}$ (Algorithm~\ref{alg:initialweight})
simply calls the weight computation algorithms of the previous sections
$\tilde{O}(\sqrt{n})$ times by first computing lewis weights for
$p=2$, i.e. leverage scores, and then gradually decreasing $p$ can
achieve this goal.

\begin{algorithm2e}[H]

\label{alg:initialweight}

\caption{\textbf{$\ensuremath{w=\code{computeInitialWeight}(\ma,p_{\text{target}},\epsilon)}$}}

\SetAlgoLined

$p=2$

\While{$p\neq p_{\text{target}}$}{

Let $r$ be defined as in $\code{computeApxWeight}$ or $\code{computeExactWeight}$

$h=\frac{\min\{2,p\}}{\sqrt{n}\log\frac{me^{2}}{n}}\cdot r$

$\next p=\code{median}(p-h,p_{\text{target}},p+h)$.

$w=\code{computeApxWeight}(\next p,w^{\frac{\next p}{p}},\frac{r}{4})$
or $\code{computeExactWeight}(\next p,w^{\frac{\next p}{p}},\frac{r}{4})$

$p=\next p$.

}

\textbf{Output:} $\code{computeApxWeight}(p_{\text{target}},w,\epsilon)$.

\end{algorithm2e}

The correctness of the above algorithm directly follows from the following
lemma:
\begin{lem}
\label{lem:LS_weights_change_bound}For all $q>0$ let $\widetilde{w}_{q}\in\dWeight$
denote the vector with $[\widetilde{w}_{q}]_{i}=[w_{p}(\ma)]_{i}^{q/p}$
for all $i\in[m]$. If $|p-q|\leq\frac{\min\{2,p\}}{\sqrt{n}\log(me^{2}/n)}$
then 
\begin{equation}
\normFull{\log\left(\frac{\lqweight(\ma)}{\widetilde{\lqweight}}\right)}_{\infty}\leq\max\left\{ \frac{1}{2},\frac{1}{p}\right\} \sqrt{n}\log\left(\frac{me^{2}}{n}\right)\cdot|p-q|\,.\label{eq:lewis_q_change}
\end{equation}
\end{lem}

\begin{proof}
For notational convenience let $w=\lpweight(\ma)$, $\mw\defeq\mDiag(w)$
and $\mLambda\defeq\mLambda(\mw_{p}^{\frac{1}{2}-\frac{1}{p}}\ma)$.
Taking derivative with respect to $p$ on both sides and using Lemma~\ref{lem:deriv:proj}
yields
\[
\frac{d\lpweight(\ma)}{dp}=2\mLambda\left[\frac{(\frac{1}{2}-\frac{1}{p})w^{-\frac{1}{2}-\frac{1}{p}}\frac{dw}{dp}+w^{\frac{1}{2}-\frac{1}{p}}\frac{1}{p^{2}}\log w}{w^{\frac{1}{2}-\frac{1}{p}}}\right]=\mLambda\left[\left(1-\frac{2}{p}\right)\mw^{-1}\frac{dw}{dp}+\frac{2}{p^{2}}\ln w\right].
\]
Hence, we have that
\begin{equation}
\frac{d\lpweight(\ma)}{dp}=2\mw\left(\mw-\left(1-\frac{2}{p}\right)\mLambda\right)^{-1}\mLambda\cdot\frac{\ln w}{p^{2}}.\label{eq:dwdp}
\end{equation}
Lemma \ref{lem:LS_weights_de} and \ref{lem:LS_weights_stable} shows
that for all $h\in\R^{m}$
\[
\normFull{2\left(\mw-(1-\frac{2}{p})\mLambda\right)^{-1}\mLambda h-ph}_{\infty}\leq p\cdot\max\left\{ \frac{p}{2},1\right\} \cdot\norm h_{\mw}.
\]
Setting $h=p^{-2}\log w$ and using \eqref{eq:dwdp} we have
\[
\normFull{\mw^{-1}\frac{d\lpweight(\ma)}{dp}-\frac{\ln\lpweight}{p}}_{\infty}\leq p\cdot\max\left\{ \frac{p}{2},1\right\} \cdot\|p^{-2}\log w\|_{\mw}\leq\max\left\{ \frac{1}{2},\frac{1}{p}\right\} \normFull{\log w}_{\mw}.
\]
Finally, we note that
\[
\norm{\ln w}_{\mw}^{2}=\sum_{i\in[m]}w_{i}\log^{2}w_{i}\leq\sum_{w_{i}\leq\frac{1}{e}}w_{i}\log^{2}w_{i}+\sum_{w_{i}\in(\frac{1}{e},1]}w_{i}\leq n\log^{2}\frac{m}{n}+n\leq n\log^{2}\frac{me}{n}
\]
where we used that $w\log^{2}w$ is concave on $[0,\frac{1}{e}]$
and $\sum_{i\in[n]}w_{i}\leq n$. 

Combining these bounds yields that for all $q$, $\lqweight\defeq\lqweight(\ma)$,
and $\mw_{q}\defeq\mDiag(w_{q})$
\begin{align*}
\normFull{\frac{d}{dq}\ln(\lqweight/\widetilde{\lqweight})}_{\infty} & =\normFull{\mw_{q}^{-1}\frac{d}{dq}\lqweight-\log(\widetilde{\lqweight})\frac{1}{p}}_{\infty}\leq\max\left\{ \frac{1}{2},\frac{1}{p}\right\} \sqrt{n}\log\left(\frac{me}{n}\right)+\frac{1}{p}\norm{\ln(\lqweight/\widetilde{\lqweight})}_{\infty}\,.
\end{align*}
Now let $\delta$ be the largest number for which $q$ satisfying
$|p-q|\leq\delta$ implies that $\norm{\ln(\lqweight/\widetilde{\lqweight})}_{\infty}\leq1$.
Since for all such $q$ we have
\[
\normFull{\frac{d}{dq}\ln(\lqweight/\widetilde{\lqweight})}_{\infty}\leq\max\left\{ \frac{1}{2},\frac{1}{p}\right\} \sqrt{n}\log\left(\frac{me^{2}}{n}\right)
\]
and $\ln(w_{p}/w_{p})=\vzero$, integration yields that \eqref{eq:lewis_q_change}
holds for all such $q$. Therefore, it must be the case that $\delta\leq\left[\max\left\{ \frac{1}{2},\frac{1}{p}\right\} \sqrt{n}\log\left(\frac{me^{2}}{n}\right)\right]^{-1}$
and the result follows.
\end{proof}

We now have everything we need to prove our main theorems regarding
exact and approximate Lewis weight computation. First we prove the
result on exact weight computation (Theorem~\ref{thm:lewisexactfull})
and then we prove the result on approximate weight computation (Theorem~\ref{thm:lewisweightexactapproxfull}).
\begin{proof}[Proof of Theorem~\ref{thm:lewisexactfull}]
From Lemma \ref{lem:LS_weights_change_bound}, we see that each step
of $p$, we lies within the requirement of Theorem~\ref{thm:lewisweightexact}.
Furthermore, Lemma \ref{lem:LS_weights_change_bound} shows that it
takes $O(\sqrt{n}\cdot(p+\frac{1}{p})\cdot\log\frac{m}{n})$ steps
in the $\code{computeInitialWeight}$. Each call of $\code{computeExactWeight}$
involves $O((p+\frac{1}{p})\log(n\epsilon^{-1}(1+\frac{1}{p}))$ iterations
and each iteration involves computing leverage score, which takes
$O(mn^{\omega-1})$ work and $O(\log m)$ depth.
\end{proof}
\begin{proof}[Proof of Theorem~\ref{thm:lewisweightexactapproxfull}]
From Lemma \ref{lem:LS_weights_change_bound}, we see that each step
of $p$, we lies within the requirement of Theorem~\ref{thm:weights_full:approximate_weight}.
Furthermore, Lemma \ref{lem:LS_weights_change_bound} shows that it
takes $O(\sqrt{n}\cdot((4-p)^{-1}+p^{-2})\cdot\log\frac{m}{n})$ steps
in the $\code{computeInitialWeight}$. Each call of $\code{computeApxWeight}$
involves $O(p^{-1}\log(n/(p\epsilon)))$ iterations and each iteration
involves computing leverage score up to accuracy $\frac{\epsilon}{32(\frac{2}{p}-|1-\frac{2}{p}|)}=\Theta((4-p)\cdot\epsilon)$.
Finally, \ref{thm:weights_full:leverage-score_sampling} shows this
involves solving solving $O((4-p)^{-2}\epsilon^{-2}\log m)$ many
linear systems.
\end{proof}

\section{Chasing Game\label{sec:approx_path:chasing_zero}}

\global\long\def\trSetCurr{U^{(k)}}%

The goal of this section is to prove the following theorem:

\zerogame*

This theorem says that taking ``projected gradient steps'' using
the potential function $\Phi_{\mu}(\vx)$, suffices to maintain a
point $x^{(k)}$ sufficiently close to $y^{(k)}$ with respect to
$\ell_{\infty}$ provided that $y^{(k)}$ are updated by a direction
in $U^{(k)}$, noisy $z^{(k)}$ measurements to the $y^{(k)}$ are
available, and slightly large movements to the $x^{(k)}$ (i.e. by
$(1+\epsilon)U^{(k)}$) are allowed. Formally, this theorem analyze
the strategy of updating $x^{(k)}$ by setting the change, $\vDelta^{(k)}$,
to be the vector in $(1+\epsilon)\trSetCurr$ that best minimizes
the potential function of the observed position difference, i.e. $\Phi_{\mu}(x^{(k)}-z^{(k)})$
for careful choice of $\mu$. 

To prove Theorem~\ref{thm:zerogame}, we first show the following
properties of the potential function $\Phi_{\mu}$.
\begin{lem}
\label{lem:smoothing:helper} For all $\vx\in\Rm$ and $\mu>0$, we
have
\begin{equation}
e^{\mu\|\vx\|_{\infty}}\leq\Phi_{\mu}(\vx)\leq2me^{\mu\|\vx\|_{\infty}}\enspace\text{ and }\enspace\mu\Phi_{\mu}(\vx)-2\mu m\leq\normOne{\grad\Phi_{\mu}(\vx)}\label{eq:smoothing:pot_prop_1}
\end{equation}
Furthermore, for any symmetric convex set $U\subseteq\Rm$ and any
$\vx\in\Rm$, let $\vx^{\flat}\defeq\argmax_{\vy\in U}\left\langle x,y\right\rangle $
and $\norm{\vx}_{U}\defeq\max_{\vy\in U}\left\langle x,y\right\rangle $.
Then for all $\vx,\vy\in\Rm$ with $\normInf{\vx-\vy}\leq\delta\leq\frac{1}{5\mu}$
we have
\begin{equation}
e^{-\mu\delta}\norm{\grad\Phi_{\mu}(\vy)}_{U}-\mu\normOne{\grad\Phi_{\mu}(\vy)^{\flat}}\leq\left\langle \grad\Phi_{\mu}(\vx),\grad\Phi_{\mu}(\vy)^{\flat}\right\rangle \leq e^{\mu\delta}\norm{\grad\Phi_{\mu}(\vy)}_{U}+\mu e^{\mu\delta}\normOne{\grad\Phi_{\mu}(\vy)^{\flat}}.\label{eq:smoothing:pot_prop_2}
\end{equation}
If additionally $U$ is contained in a $\ellInf$ ball of radius $R$
then
\begin{equation}
e^{-\mu\delta}\norm{\grad\Phi_{\mu}(\vy)}_{U}-\mu mR\leq\norm{\grad\Phi_{\mu}(\vx)}_{U}\leq e^{\mu\delta}\norm{\grad\Phi_{\mu}(\vy)}_{U}+\mu e^{\mu\delta}mR.\label{eq:smoothing:pot_prop_3}
\end{equation}
\end{lem}

\begin{proof}
For, notational convenience let $p_{u}(x)\defeq e^{\mu x}+e^{-\mu x}$
for all $x\in\R$ so that $\Phi_{\mu}(x)=\sum_{i\in[m]}p_{u}(x_{i})$.
Equation (\ref{eq:smoothing:pot_prop_1}) follows from the fact that
for all $x\in\R$,
\[
e^{\mu|x|}\leq p_{\mu}(x)\leq2e^{\mu|x|}\enspace\text{and}\enspace p'_{\mu}(x)=\mu\sign(x)\left(e^{\mu|x|}-e^{-\mu|x|}\right)\,.
\]

Next, let $x,y\in\R$ with $|x-y|\leq\delta$. Note that $\left|p'_{\mu}(x)\right|=p'_{\mu}(\left|x\right|)=\mu(e^{\mu|x|}-e^{-\mu|x|})$
and $\left|x-y\right|\leq\delta$ implies that $|x|=|y|+z$ for some
$z\in[-\delta,\delta]$. Using that $p'(|x|)$ is monotonic in $|x|$
we then have
\begin{align}
|p'_{\mu}(x)| & =p'_{\mu}(|x|)=p'_{\mu}(|y|+z)\leq p'_{\mu}(|y|+\delta)=\mu\left(e^{\mu|y|+\mu\delta}-e^{-\mu|y|-\mu\delta}\right)\nonumber \\
 & =e^{\mu\delta}p'(|y|)+\mu\left(e^{\mu\delta-\mu|y|}-e^{-\mu|y|-\mu\delta}\right)\leq e^{\mu\delta}\left|p'(y)\right|+\mu e^{\mu\delta}.\label{eq:smoothingp2}
\end{align}
By symmetry (i.e. replacing $x$ and $y$) this implies that
\begin{equation}
|p'_{\mu}(x)|\geq e^{-\mu\delta}|p'(y)|-\mu\label{eq:smoothing:p3}
\end{equation}

Since $U$ is symmetric this implies that for all $i\in[m]$ we have
$\sign(\grad\Phi_{\mu}(\vy)^{\flat})_{i}=\sign(\grad\Phi_{\mu}(\vy)_{i})=\sign(y_{i})$.
Therefore, if for all $i\in[n]$ we have $\sign(x_{i})=\sign(y_{i})$,
by (\ref{eq:smoothingp2}), we see that 
\begin{eqnarray*}
\left\langle \grad\Phi_{\mu}(\vx),\grad\Phi_{\mu}(\vy)^{\flat}\right\rangle  & = & \sum_{i\in[m]}p'_{\mu}(x_{i})\grad\Phi_{\mu}(\vy)_{i}^{\flat}\leq\sum_{i\in[m]}\left(e^{\mu\delta}p'_{\mu}(y_{i})+\mu e^{\mu\delta}\right)\grad\Phi_{\mu}(\vy)_{i}^{\flat}\\
 & \leq & e^{\mu\delta}\left\langle \grad\Phi_{\mu}(\vy),\grad\Phi_{\mu}(\vy)^{\flat}\right\rangle +\mu e^{\mu\delta}\norm{\grad\Phi_{\mu}(\vy)^{\flat}}_{1}\\
 & = & e^{\mu\delta}\norm{\grad\Phi_{\mu}(\vy)}_{U}+\mu e^{\mu\delta}\normOne{\grad\Phi_{\mu}(\vy)^{\flat}}.
\end{eqnarray*}
Similarly, using (\ref{eq:smoothing:p3}), we have $e^{-\mu\delta}\norm{\grad\Phi_{\mu}(\vy)}_{U}-\mu\normOne{\grad\Phi_{\mu}(\vy)^{\flat}}\leq\left\langle \grad\Phi_{\mu}(\vx),\grad\Phi_{\mu}(\vy)^{\flat}\right\rangle $
and hence (\ref{eq:smoothing:pot_prop_2}) holds. On the other hand
if $\sign(x_{i})\neq\sign(y_{i})$ then we know that $|x_{i}|\leq\delta$
and consequently $|p'_{\mu}(x_{i})|\leq\mu(e^{\mu\delta}-e^{-\mu\delta})\leq\frac{\mu}{2}$
since $\delta\leq\frac{1}{5\mu}$. Thus, we have
\[
e^{-\mu\delta}\left|p'_{\mu}(y_{i})\right|-\mu\leq-\frac{\mu}{2}\leq\sign\left(y_{i}\right)p'_{\mu}(x_{i})\leq0\leq e^{\mu\delta}\left|p'_{\mu}(y_{i})\right|+\mu e^{\mu\delta}.
\]
Taking inner product on both sides with $\grad\Phi_{\mu}(\vy)_{i}^{\flat}$
and using definition of $\norm{\cdot}_{U}$ and $\cdot^{\flat}$,
we get (\ref{eq:smoothing:pot_prop_2}). Thus, (\ref{eq:smoothing:pot_prop_2})
holds in general.

Finally we note that since $U$ is contained in a $\ellInf$ ball
of radius $R$, we have $\normOne{\vy^{\flat}}\leq mR$ for all $\vy$.
Using this fact, (\ref{eq:smoothing:pot_prop_2}), and the definition
of $\norm{\cdot}_{U}$, we obtain
\[
e^{-\mu\delta}\norm{\grad\Phi_{\mu}(\vy)}_{U}-\mu mR\leq\left\langle \grad\Phi_{\mu}(\vx),\grad\Phi_{\mu}(\vy)^{\flat}\right\rangle \leq\norm{\grad\Phi_{\mu}(\vx)}_{U}
\]
where the last inequality additionally uses $\grad\Phi_{\mu}(\vy)^{\flat}\in U$.
By symmetry (\ref{eq:smoothing:pot_prop_3}) follows.
\end{proof}
Using Lemma~\ref{lem:smoothing:helper} we prove Theorem~\ref{thm:zerogame}.
\begin{proof}[Proof of Theorem~\ref{thm:zerogame}]
For the remainder of the proof, let $\norm{\vx}_{U^{(k)}}=\max_{\vy\in U^{(k)}}\left\langle \vx,\vy\right\rangle $
and $\vx^{\flat_{(k)}}=\argmax_{\vy\in U^{(k)}}\left\langle \vx,\vy\right\rangle $.
Since $U^{(k)}$ is symmetric, we know that $\vDelta^{(k)}=-(1+\epsilon)\left(\grad\Phi_{\mu}(x^{(k-1)}-\vz^{(k)})\right)^{\flat_{(k)}}$
and therefore by applying the mean value theorem twice we have that
\begin{eqnarray*}
\Phi_{\mu}(x^{(k)}-y^{(k)}) & = & \Phi_{\mu}(x^{(k-1)}-y^{(k)})+\left\langle \grad\Phi_{\mu}(\zeta_{1}),x^{(k)}-x^{(k-1)}\right\rangle \\
 & = & \Phi_{\mu}(x^{(k-1)}-y^{(k-1)})+\left\langle \grad\Phi_{\mu}(\zeta_{2}),y^{(k)}-y^{(k-1)}\right\rangle +\left\langle \grad\Phi_{\mu}(\zeta_{1}),x^{(k)}-x^{(k-1)}\right\rangle 
\end{eqnarray*}
for some $\zeta_{1}$ between $x^{(k)}-y^{(k)}$ and $x^{(k-1)}-y^{(k)}$
and some $\zeta_{2}$ between $x^{(k-1)}-y^{(k)}$ and $x^{(k-1)}-y^{(k-1)}$.
Now, using that $y^{(k)}-y^{(k-1)}\in U^{(k)}$ and that $x^{(k)}-x^{(k-1)}=\vDelta^{(k)}$
we have
\begin{equation}
\Phi_{\mu}(x^{(k)}-y^{(k)})\leq\Phi_{\mu}(x^{(k-1)}-y^{(k-1)})+\norm{\nabla\Phi_{\mu}(\zeta_{2})}_{\trSetCurr}-\left(1+\epsilon\right)\left\langle \nabla\Phi_{\mu}(\zeta_{1}),\left(\nabla\Phi_{\mu}(x^{(k-1)}-z^{(k)})\right)^{\flat_{(k)}}\right\rangle .\label{eq:Phi_est_1}
\end{equation}
Since $U^{k}$ is contained within the $\ellInf$ ball of radius $R_{k}$,
Lemma~\ref{lem:smoothing:helper} shows that
\begin{equation}
\norm{\nabla\Phi_{\mu}(\zeta_{2})}_{\trSetCurr}\leq e^{\mu R_{k}}\norm{\nabla\Phi_{\mu}(x^{(k-1)}-y^{(k-1)})}_{\trSetCurr}+m\mu R_{k}e^{\mu R_{k}}.\label{eq:Phi_est_2}
\end{equation}
Furthermore, since $\epsilon<\frac{1}{5}$ and $R_{k}\leq R$, by
triangle inequality we have $\normInf{\zeta_{1}-(x^{(k-1)}-z^{(k)})}\leq(1+\epsilon)R_{k}+R\leq3R$
and $\normInf{z^{(k)}-y^{(k-1)}}\leq2R$. Therefore, applying Lemma~\ref{lem:smoothing:helper}
twice yields that
\begin{align}
\left\langle \nabla\Phi_{\mu}(\zeta_{1}),\left(\nabla\Phi_{\mu}(x^{(k-1)}-z^{(k)})\right)^{\flat_{(k)}}\right\rangle  & \geq e^{-3\mu R}\norm{\nabla\Phi_{\mu}(x^{(k-1)}-z^{(k)})}_{\trSetCurr}-\mu mR_{k}\nonumber \\
 & \geq e^{-5\mu R}\norm{\nabla\Phi_{\mu}(x^{(k-1)}-y^{(k-1)})}_{\trSetCurr}-2\mu mR_{k}.\label{eq:Phi_est_3}
\end{align}
Combining (\ref{eq:Phi_est_1}), (\ref{eq:Phi_est_2}), and (\ref{eq:Phi_est_3})
then yields that
\begin{align*}
\Phi_{\mu}(x^{(k)}-y^{(k)})\leq & \Phi_{\mu}(x^{(k-1)}-y^{(k-1)})-\left((1+\epsilon)e^{-5\mu R}-e^{\mu R}\right)\norm{\nabla\Phi_{\mu}(x^{(k-1)}-y^{(k-1)})}_{\trSetCurr}\\
 & +m\mu R_{k}e^{\mu R}+2(1+\epsilon)m\mu R_{k}.
\end{align*}
Since we chose $\mu=\frac{\epsilon}{12R}$ and $\epsilon\in(0,1/5)$
we have $(1+\epsilon)e^{-5\mu R}-e^{\mu R}\geq\frac{2\epsilon}{5}$
and 
\[
m\mu R_{k}e^{\mu R}+2(1+\epsilon)m\mu R_{k}\leq\epsilon m\frac{7R_{k}}{24R}.
\]
Thus, we have
\[
\Phi_{\mu}(x^{(k)}-y^{(k)})\leq\Phi_{\mu}(x^{(k-1)}-y^{(k-1)})-\frac{2\epsilon}{5}\norm{\nabla\Phi_{\mu}(x^{(k-1)}-y^{(k-1)})}_{\trSetCurr}+\epsilon m\frac{7R_{k}}{24R}.
\]
Using Lemma~\ref{lem:smoothing:helper} and the fact that $U_{k}$
contains a $\ellInf$ ball of radius $r_{k}$, we have
\[
\norm{\nabla\Phi_{\mu}(x^{(k-1)}-y^{(k-1)})}_{\trSetCurr}\geq r_{k}\norm{\nabla\Phi_{\mu}(x^{(k-1)}-y^{(k-1)})}_{1}\geq\frac{\epsilon r_{k}}{12R}\left(\Phi_{\mu}(x^{(k-1)}-y^{(k-1)})-2m\right).
\]
Therefore, we have that
\begin{eqnarray*}
\Phi_{\mu}(x^{(k)}-y^{(k)}) & \leq & \left(1-\frac{\epsilon^{2}r_{k}}{30R}\right)\Phi_{\mu}(x^{(k-1)}-y^{(k-1)})+\frac{\epsilon^{2}r_{k}}{15R}m+\epsilon m\frac{7R_{k}}{24R}\\
 & \leq & \left(1-\frac{\epsilon^{2}r_{k}}{30R}\right)\Phi_{\mu}(x^{(k-1)}-y^{(k-1)})+\epsilon m\frac{R_{k}}{3R}.
\end{eqnarray*}
Hence, if $\Phi_{\mu}(x^{(k-1)}-y^{(k-1)})\leq\frac{12m\tau}{\epsilon}$,
we have $\Phi_{\mu}(x^{(k)}-y^{(k)})\leq\frac{12m\tau}{\epsilon}$.
Since $\Phi_{\mu}(x^{(0)}-y^{(0)})\leq\frac{12m\tau}{\epsilon}$ by
assumption we have by induction that $\Phi_{\mu}(x^{(k)}-y^{(k)})\leq\frac{12m\tau}{\epsilon}$
for all $k$. The necessary bound on $\normInf{x^{(k)}-y^{(k)}}$
then follows immediately from Lemma~\ref{lem:smoothing:helper}.
\end{proof}

\section{Appendix: Projection on Mixed Norm Ball\label{sec:app:Project_ball_box}}

Here we give an algorithm to solve the following problem
\begin{equation}
\max_{\norm{\vx}_{2}+\norm{\vl^{-1}\vx}_{\infty}\leq1}\left\langle \va,\vx\right\rangle \label{eq:project_our_form}
\end{equation}
for some given vector $l$ and $a$. This is used in Section \ref{sec:master_thm}
to compute weights. Note that 
\begin{eqnarray}
\max_{\norm{\vx}_{2}+\norm{\vl^{-1}\vx}_{\infty}\leq1}\left\langle \va,\vx\right\rangle  & = & \max_{0\leq t\leq1}\left[\max_{\norm{\vx}_{2}\leq1-t\text{ and }-tl_{i}\leq x_{i}\leq tl_{i}}\left\langle \va,\vx\right\rangle \right]\nonumber \\
 & = & \max_{0\leq t\leq1}(1-t)\left[\max_{\norm{\vx}_{2}\leq1\text{ and }-\frac{t}{1-t}l_{i}\leq x_{i}\leq\frac{t}{1-t}l_{i}}\left\langle \va,\vx\right\rangle \right]\,.\nonumber \\
 & = & \max_{0\leq t\leq1}(1-t)f(t)\text{ where }f(t)\defeq\max_{\norm{\vx}_{2}\leq1,-\frac{t}{1-t}l_{i}\leq x_{i}\leq\frac{t}{1-t}l_{i}}\left\langle \va,\vx\right\rangle \,.\label{eq:projection_mixed_norm_ball_eq}
\end{eqnarray}
After sorting the coordinates so that $|a_{i}|/l_{i}$ monotonically
decrease with $i\in[n]$, and considering the maximization problem
in $f(t)$ with only the $\norm x_{2}$ or $-\frac{t}{1-t}l_{i}\leq x_{i}\leq\frac{t}{1-t}l_{i}$
constraints, it can be shown that the maximizing $x$ in the definition
of $f$ is $\vx^{i_{t}}$ where for all $j\in[n]$
\begin{equation}
\vx_{j}^{i_{t}}=\begin{cases}
\frac{t}{1-t}\text{sign}(a_{j})l_{j} & \text{if }j\in[i_{t}]\\
\sqrt{\frac{1-\left(\frac{t}{1-t}\right)^{2}\sum_{k\in[i_{t}]}l_{k}^{2}}{\norm a_{2}^{2}-\sum_{k\in[i_{t}]}a_{k}^{2}}}\va_{j} & \text{otherwise}
\end{cases}.\label{eq:project_x_form}
\end{equation}
and $i_{t}$ is the first coordinate $i\in[n]$ such that
\[
\frac{1-\left(\frac{t}{1-t}\right)^{2}\sum_{k\in[i]}l_{k}^{2}}{\norm a_{2}^{2}-\sum_{k\in[i]}a_{k}^{2}}\leq\frac{\left(\frac{t}{1-t}\right)^{2}l_{i}^{2}}{a_{i}^{2}}.
\]
Note that $i_{t}\geq i_{s}$ if $t\leq s$. Therefore, the set of
$t$ such that $i_{t}=j$ is simply an interval given by\footnote{There are some boundary cases we ignored for simplicity.}
\begin{equation}
\frac{\left|a_{j}\right|}{\sqrt{l_{j}^{2}\left(\norm a_{2}^{2}-\sum_{k\in[j]}a_{k}^{2}\right)+a_{j}^{2}\sum_{k\in[j]}l_{k}^{2}}}\leq\frac{t}{1-t}<\frac{\left|a_{j-1}\right|}{\sqrt{l_{j-1}^{2}\left(\norm a_{2}^{2}-\sum_{k\in[j-1]}a_{k}^{2}\right)+a_{j-1}^{2}\sum_{k\in[j-1]}l_{k}^{2}}}.\label{eq:proj_t_interval}
\end{equation}
Therefore, we know that
\[
f(t)=\left\langle \va,\vx^{(i_{t})}\right\rangle =\frac{t}{1-t}\sum_{j\in[i_{t}]}\left|a_{j}\right|\left|l_{j}\right|+\sqrt{1-\left(\frac{t}{1-t}\right)^{2}\sum_{k\in[i_{t}]}l_{k}^{2}}\sqrt{\norm a_{2}^{2}-\sum_{k\in[i_{t}]}a_{k}^{2}}\,.
\]
Substituting this into \ref{eq:projection_mixed_norm_ball_eq}, we
have that
\begin{eqnarray*}
\max_{\norm{\vx}_{2}+\norm{\vl^{-1}\vx}_{\infty}\leq1}\left\langle \va,\vx\right\rangle  & = & \max_{0\leq t\leq1}g(t)\defeq t\sum_{j\in[i_{t}]}\left|a_{j}\right|\left|l_{j}\right|+\sqrt{(1-t)^{2}-t^{2}\sum_{k\in[i_{t}]}l_{k}^{2}}\sqrt{\norm a_{2}^{2}-\sum_{k\in[i_{t}]}a_{k}^{2}}.
\end{eqnarray*}

Note that 
\begin{align*}
g'(t) & =\sum_{j\in[i_{t}]}\left|a_{j}\right|\left|l_{j}\right|+\frac{\left((1-\sum_{k\in[i_{t}]}l_{k}^{2})t-1\right)\sqrt{\norm a_{2}^{2}-\sum_{k\in[i_{t}]}a_{k}^{2}}}{\sqrt{(1-t)^{2}-t^{2}\sum_{k\in[i_{t}]}l_{k}^{2}}},\\
g''(t) & =-\frac{\left(\sum_{k\in[i_{t}]}l_{k}^{2}\right)\cdot\sqrt{\norm a_{2}^{2}-\sum_{k\in[i_{t}]}a_{k}^{2}}}{\left((1-t)^{2}-t^{2}\sum_{k\in[i_{t}]}l_{k}^{2}\right)^{3/2}}.
\end{align*}
Hence, $g(t)$ is concave and its maximizer has a closed form via
the quadratic formula. Therefore, one can compute the maximum value
for each interval of $t$ (\ref{eq:proj_t_interval}) and find which
is the best. This yields the following algorithm.

\begin{algorithm2e}[H]

\caption{$\vx=\code{projectMixedBall}(\va,\vl)$}

\SetAlgoLined

Sort the coordinate such that $\left|a_{i}\right|/l_{i}$ is in descending
order.

Precompute $\sum_{k=0}^{i}l_{k}^{2}$, $\sum_{k=0}^{i}a_{k}^{2}$
and $\sum_{j=1}^{i}\left|a_{j}\right|\left|l_{j}\right|$ for all
$i$. 

Let $g_{i}(t)=t\sum_{j\in[i]}\left|a_{j}\right|\left|l_{j}\right|+\sqrt{(1-t)^{2}-t^{2}\sum_{k=0}^{i}l_{k}^{2}}\sqrt{\norm a_{2}^{2}-\sum_{k=0}^{i}a_{k}^{2}}.$

For each $j\in\{1,\cdots,n\}$, Find $t_{j}=\argmax_{i_{t}=j}g_{j}(t)$
using (\ref{eq:proj_t_interval}).

Find $i=\argmax_{i}g_{i}(t_{i}).$

\textbf{Output:} $(1-t_{i})\vx^{(i)}$ defined by (\ref{eq:project_x_form}).

\end{algorithm2e}

The discussion above leads to the following theorem.
\begin{thm}
\label{thm:project_mixed_norm} For any $a\in\Rn$ and $l\in\R_{>0}^{n}$,
the algorithm \textup{$\code{projectMixedBall}(a,l)$ outputs a solution
to }(\ref{eq:project_our_form})\textup{ in total work $O(n\log n)$
and depth $O(\log n)$ (in EREW model).}
\end{thm}

\begin{proof}
The correctness follows from the discussion above. For the runtime,
it is known that sorting can be done in $O(n\log n)$ work and $O(\log n)$
depth in EREW model \cite{cole1988parallel} and that prefix sum can
be done in $O(n)$ work and $O(\log n)$ depth in EREW model. The
rest is easy.
\end{proof}

\section{Extreme Lewis Weights and Barrier}

\label{sec:extreme_lewis_weights}

In this section we discuss the limits of Lewis weights and the Lewis
weight barrier when $p\rightarrow0$ and $p\rightarrow\infty$. In
Section~\ref{sec:pzero} we show that as $p\rightarrow0$ Lewis weights
converge to the uniform distribution over rows of a matrix under mild
assumptions. This shows that under mild assumptions on the structure
of a polytope, the Lewis weight barrier considered in Section~\ref{sec:self-concordance}
converges to the standard logarithmic barrier. In Section~\ref{sec:pinf}
we consider the opposite extreme when $p\rightarrow\infty$. In this
case we show that $\ell_{\infty}$ Lewis weights of the matrix $\ma$
are precisely the weights that induce a John ellipse of the polytope
$\{x\in\R^{n}:\norm{\ma x}_{\infty}\leq1\}$. This justifies the intuition
given in the introduction regarding our barrier and path finding scheme
as following a path induced by regularized John ellipses.

\subsection{$p\rightarrow0$}

\label{sec:pzero}

Here we show that $\ell_{p}$ Lewis weights for a matrix $\ma\in\R^{m\times n}$
in general position, i.e. any $n$ rows are linearly independent,
converge to uniform as $p\rightarrow0$. Note that the assumption
of general position is stronger than that of non-degeneracy and required
for the statement to be true. For example, if there is a row that
is perpendicular to all other rows, then it is not difficult to show
that this row must have Lewis weight $1$ for any $p>0$. 
\begin{lem}
\label{lem:l0} Given a matrix $\ma\in\Rmn$ in general position,
i.e. any $n$ rows of $\ma$ are linearly independent, then
\[
\lim_{p\rightarrow0^{+}}\lpweight(\ma)_{i}=\frac{n}{m}\text{ for all }i.
\]
\end{lem}

\begin{proof}
For $p\in(0,2)$, Lemma~\ref{lem:unique_lewis} shows that the Lewis
weight is given by 
\[
\lpweight(\ma)=\argmin_{w\in\R_{\geq0}^{m},\sum_{i\in[m]}w_{i}=n}\det(\ma^{\top}\mw^{1-\frac{2}{p}}\ma).
\]
Considering $w\in\R^{m}$ with $w_{i}=\frac{n}{m}$ for all $i\in[m]$
we see that 
\begin{align}
\min_{w_{i}\geq0,\sum_{i\in[m]}w_{i}=n}\det(\ma^{\top}\mw^{1-\frac{2}{p}}\ma) & \leq\left(\frac{n}{m}\right)^{n(1-\frac{2}{p})}\det(\ma^{\top}\ma).\label{eq:l0_lewis}
\end{align}
On the other hand the Cauchy--Binet formula shows that
\[
\det(\ma^{\top}\mw^{1-\frac{2}{p}}\ma)=\sum_{S\in{[m] \choose n}}\det(\ma_{S}){}^{2}\det(\mw_{S}^{1-\frac{2}{p}})
\]
where $\ma_{S}\subset\Rnn$ are the rows of $\ma$ at indices from
$S$, $\mw_{S}\subset\Rnn$ is diagonal with the diagonals of $\mw$
at indices from $S$ and the summation is over all subsets of size
$n$. Since $\ma$ is in general position, we have that $\det\ma_{S}\neq0$
for all $S$. Therefore, for all subsets $S\subseteq[m]$ of size
$n$ and all $\mw\succeq0$.
\begin{equation}
\det(\mw_{S}^{1-\frac{2}{p}})\leq\frac{\det(\ma^{\top}\mw^{1-\frac{2}{p}}\ma)}{\min_{S\subset\binom{[m]}{n}}\det(\ma_{S})^{2}}\,.\label{eq:l0_lewis2}
\end{equation}

Now, let $\mw_{p}=\mDiag(\lpweight(\ma))$ be the diagonal matrix
formed by the $\ell_{p}$ Lewis weight of $\ma$. Combining (\ref{eq:l0_lewis})
and (\ref{eq:l0_lewis2}), we have that
\[
\det([\mw_{p}^{1-\frac{2}{p}}]_{S})\leq c\left(\frac{n}{m}\right)^{n(1-\frac{2}{p})}\text{ where }c=\frac{\det\ma^{\top}\ma}{\min_{S\subset\binom{m}{n}}(\det\ma_{S})^{2}}.
\]
Hence, we have that $\det([\mw_{p}]_{S})\geq c^{\frac{1}{1-\frac{2}{p}}}\cdot(\frac{n}{m})^{n}.$
Let $\mw^{*}=\liminf_{p\rightarrow0^{+}}\mw_{p}$. Taking limit $p\rightarrow0^{+}$
on both sides, we have that $\det(\mw_{S}^{*})\geq(n/m)^{n}$ for
all subsets $S$ of size $n$. Since this holds for all subsets and
since $\sum_{i\in[m]}w_{i}^{*}=n$, we have that $w_{i}^{*}=\frac{n}{m}$
for all $i$. Since 
\[
\limsup_{p\rightarrow0^{+}}\sum_{i\in[m]}\lpweight(\ma)_{i}=n=\liminf_{p\rightarrow0^{+}}\sum_{i\in[m]}\lpweight(\ma)_{i}
\]
this shows that $\lim w_{p}$ exists and it converges to $\frac{n}{m}$.
\end{proof}

\subsection{$p\rightarrow\infty$}

\label{sec:pinf}

Here we show that as $p\rightarrow\infty$ the ellipse $E=\{x\in\Rn:\ x^{\top}\mm x\leq1\}$
for $\mm\defeq\lim_{p\rightarrow+\infty}\ma^{\top}\mw_{p}\ma$ where
$\mw_{p}=\lpweight(\ma)$ is the John ellipse of the polytope $K=\{x\in\Rn:\ \norm{\ma x}_{\infty}\leq1\}$,
i.e. the ellipsoid of maximum volume contained inside $K$. To prove
this we use the following lemma proved in \cite{khachiyan1996rounding}
characterizing the John Ellipse.
\begin{lem}
\label{lem:John_ellipsoid_formula} Given a polytope $\Omega=\{x\in\Rn:\ \norm{\ma x}_{\infty}\leq1\}$
with $\ma\in\R^{m\times n}$. Let $E$ be the John ellipsoid of $\Omega$,
namely, $E$ is the maximum volume ellipsoid contained inside $\Omega$.
Then, we have that $E=\{x^{\top}\ma^{\top}\mw\ma x\leq1\}$ with the
diagonal matrix $\mw$ given by the vector maximizing
\[
\min_{w_{i}\geq0,\sum_{i\in[m[}w_{i}=n}\log\det\ma^{\top}\mw\ma.
\]
\end{lem}

Using this we prove our desired result regarding the limits of Lewis
weights as $p\rightarrow\infty$.
\begin{lem}
For non-degenerate $\ma\in\Rmn$ let $\mm\defeq\lim_{p\rightarrow+\infty}\ma^{\top}\mw_{p}\ma$
where $\mw_{p}=\mDiag(\lpweight(\ma))$. Then $E=\{x\in\Rn:\ x^{\top}\mm x\leq1\}$
is the John ellipsoid of $K=\{x\in\Rn:\ \norm{\ma x}_{\infty}\leq1\}$.
\end{lem}

\begin{proof}
Let $c_{p,m}$ be the constant defined in Lemma~\ref{lem:lewis_rounding}.
Further, for all $p>2$ let $\mm_{p}\defeq c_{p,m}^{2}\ma^{\top}\mw_{p}\ma$
and $E_{p}\defeq\{x\in\Rn:\ x^{\top}\mm_{p}x\leq1\}$. Lemma~\ref{lem:lewis_rounding}
shows that $E_{p}\subseteq K$. Further, letting $s_{n}$ is the volume
of the unit sphere we have that
\begin{align*}
\vol(E_{p}) & =s_{n}\left(c_{p,m}^{2n}\det(\ma^{\top}\mw_{p}\ma)\right)^{-\frac{1}{2}}\geq s_{n}\left(c_{p,m}^{2n}\det(\ma^{\top}\mw_{p}^{1-\frac{2}{p}}\ma)\right)^{-\frac{1}{2}}\\
 & =s_{n}\left(c_{p,m}^{2n}\cdot\min_{w_{i}\geq0,\sum_{i\in[m]}w_{i}=n}\log\det\ma^{\top}\mw^{1-\frac{2}{p}}\ma\right)^{-\frac{1}{2}}
\end{align*}
where in the last step we used Lemma~\ref{lem:unique_lewis}. Note
that $w^{1-\frac{2}{p}}\rightarrow w$ as $p\rightarrow\infty$ for
all $w>0$ and $c_{p,m}\rightarrow1$ as $p\rightarrow\infty$. Hence,
\[
\liminf_{p\rightarrow+\infty}\vol(E_{p})\geq s_{n}\cdot\left(\min_{w_{i}\geq0,\sum_{i\in[m]}w_{i}=n}\log\det\ma^{\top}\mw\ma\right)^{-\frac{1}{2}}
\]
On the other hand, it is known that the John ellipsoid $E^{*}$ of
$K$ is unique and its volume is given by the right hand side (Lemma~\ref{lem:John_ellipsoid_formula}).
This implies that $E_{p}$ converges to the John ellipsoid of $K$.
\end{proof}

\section{Linear System Properties\label{subsec:well_conditioned}}

Often the the running time for solving linear system solvers depends
on the condition number of the matrix and/or how fast the linear systems
change from iteration to iteration. Here we show that our interior
point method enjoys properties frequently exploited in other interior
point methds and therefore is amenable to techniques for improving
iteration costs.

There are two key lemmas we prove in this section. First, in Lemma~\ref{lem:slacks_bound}
we provide a general technical lemma on the structure of weighted
minimizers of self-concordant barriers. This allows us to reason about
how close the weighted central path can go to the boundary of the
polytope and allows us to reason about how ill-conditioned the linear
system we need to solver become over the course of the algorithm (see
Corollary~\ref{cor:distance_bdy} and its proof). Second, in Lemma~\ref{lem:sequence_system_changes}
we bound how much the linear systems can change over the course of
our algorithm.

Since Lemma~\ref{lem:slacks_bound} is of independent interest, we
prove a slightly more general version than what we need here (which
in turn is a generalization of \cite[Lemma 16]{vavasis1996primal}).
We consider the case of minimizing weighted combinations of arbitrary
self-concordant functions subject to a linear constraint and bound
under changes to the weights upper bound how well the line between
minimizers. To prove Lemma~\ref{lem:slacks_bound} we use the following
known equivalent characterization of self-concordance and properties
of self-concordant functions.
\begin{lem}
[{\cite[Theorem 4.1.6, 4.2.4]{Nesterov2003}}]\label{lem:general_self_concordant}
We call convex function $\phi$ a \emph{$\nu$-self-concordant barrier}
for open convex set $\Omega\subset\Rn$ if $\phi(\vx)\rightarrow+\infty$
as $\vx\rightarrow\partial\Omega$ and for all $x\in\Omega$ and $h\in\R^{n}$,
$\psi(t)\defeq\phi(\vx+th)$ satisfies $\psi'''(0)\leq2(\psi''(0))^{3/2}$
and $\psi'(0)\leq(\nu\cdot\psi''(0))^{1/2}$. For such $\phi$ the
following hold.
\begin{itemize}
\item For all $s\in\Omega$ and $t\in\Rn$ such that $\|t-s\|_{\nabla^{2}\phi(s)}<1$
then $t\in\dom(\phi)$.
\item For all $x,y\in\Omega$ we have $\nabla\phi(x)\cdot(y-x)\leq\nu$.
\end{itemize}
\end{lem}

\begin{lem}
\label{lem:slacks_bound} For all $i\in[k]$ let $\phi_{i}$ be a
$\nu_{i}$-self-concordant barriers on $\dom(\phi_{i})$, a open convex
subset of $\R^{m}$. Let $\Omega\defeq\{\vx:\ma^{\top}\vx=\vb\}\cap(\cap_{i\in[m]}\dom(\phi_{i}))$
for arbitrary $\ma\in\R^{m\times n}$ and $b\in\R^{n}$. Further,
$c\in\R^{m}$ and for all $w\in\R_{>0}^{m}$ let 
\[
\vx_{\vWeight}\defeq\arg\min_{x\in\Omega}\vc^{\top}\vx+\sum_{i\in[m]}w_{i}\phi_{i}(\vx_{i}).
\]
Then if $\vWeight^{(0)},\vWeight^{(1)}\in\dWeight$ are such that
either $w^{(0)}\geq w^{(1)}$ or $w^{(1)}\geq w^{(0)}$ entrywise
then $p(t)\defeq\vx_{\vWeight^{(0)}}+t(\vx_{\vWeight^{(1)}}-\vx_{\vWeight^{(0)}})\in\Omega$
for all $t\in(-\theta,1+\theta)$ where
\begin{equation}
\theta\defeq\frac{\min\left\{ \min_{j\in[k]}w_{j}^{(0)},\min_{j\in[k]}w_{j}^{(1)}\right\} }{\sum_{j\in[k]}\nu_{j}\left|w_{j}^{(0)}-w_{j}^{(1)}\right|}\,.\label{eq:slack_simple}
\end{equation}
Further, for any $\vWeight^{(0)},\vWeight^{(1)}\in\dWeight$ (regardless
of their entrywise relation), we have $p(t)\in\Omega$ for all $t\in(-\gamma,1+\gamma)$
where $\gamma=\frac{\theta^{2}}{1+2\theta}$ for $\theta$ as defined
above.
\end{lem}

\begin{proof}
First, we prove the case where either $w^{(0)}\geq w^{(1)}$ or $w^{(1)}\geq w^{(0)}$
entrywise. Note that $p(t)$ is a straight line intersecting $p(0)=\vx_{\vWeight^{(0)}}$
and $p(1)=x_{w^{(1)}}$. Let $\theta$ denote the smallest value of
$\theta$ for which either $p(-\theta)\notin\Omega$ or $p(1+\theta)\notin\Omega$,
i.e. the least amount the line segment between $\vx_{\vWeight^{(0)}}$
and $x_{w^{(1)}}$ needs to be extended in either direction to leave
$\Omega$. By convexity $p(t)\in\Omega$ for all $t\in[0,1]$ and
therefore $\theta\geq0$. Further, we assume that $\theta$ is finite,
i.e. the straight line passing through $x_{w^{(0)}}$ and $x_{w^{(1)}}$
leaves $\Omega$, as otherwise the lemma trivially holds.

Note that either $p(1+\theta)\in\partial\Omega$, the boundary of
$\Omega$, or $p(-\theta)\in\partial\Omega$. By symmetry, we assume
without loss of generality that $p(1+\theta)\in\partial\Omega$ (as
applying the lemma under this assumption with $w^{(0)}$ and $w^{(1)}$
swapped would yield the other case). Consequently, for some $i\in[m]$
we have $p(1+\theta)\in\partial\dom(\phi_{i})$, the boundary of $\phi_{i}$,
and we fix such a $i\in[m]$ throughout. For notational convenience,
we define $\psi_{j}(t)\defeq\phi_{j}(p(t))$ for all $j\in[m]$ and
let $\dom(\psi_{i})$ denote the set of values of $u$ for which $p(u)\in\dom(\phi_{i})$.
We will leverage that each $\psi_{j}$ is $\nu_{j}$-self-concordant
on $\dom(\phi_{i})$, as the restriction of a $\text{\ensuremath{\nu}}$-self-concordant
function to a line is $\nu$-self-concordant (see Definition~\ref{lem:general_self_concordant}).

Now, note that the first bullet of Lemma~\ref{lem:general_self_concordant}
implies that if for some $t\in(-\theta,1+\theta)$ and $u\geq t$
we have $\sqrt{\psi''_{i}(t)}(u-t)<1$ then $u\in\dom(\psi_{i})$.
However, we know that $u\notin\dom(\psi_{i})$ for $u=1+\theta$ and
thus, $\psi''_{i}(t)\geq(1+\theta-t)^{-2}$. Integrating, yields that
\[
\psi'_{i}(1)=\psi_{i}'(0)+\int_{0}^{1}\psi''_{i}(t)dt\geq\psi_{i}'(0)+\int_{0}^{1}\frac{1}{(1+\theta-t)^{2}}dt=\psi_{i}'(0)+\frac{1}{\theta(1+\theta)}\,.
\]
Further, since each $\psi_{j}$ is convex we have that $\psi_{j}'(1)\geq\psi_{j}'(0)$
and combining yields that 
\begin{equation}
\sum_{j\in[k]}w_{j}^{(1)}\psi_{j}'(1)\geq\sum_{j\in[k]}w_{j}^{(1)}\psi_{j}'(0)+\frac{w_{i}^{(1)}}{\theta(1+\theta)}\text{ and }\sum_{j\in[k]}w_{j}^{(0)}\psi_{j}'(1)\geq\sum_{j\in[k]}w_{j}^{(0)}\psi_{j}'(0)+\frac{w_{i}^{(0)}}{\theta(1+\theta)}\,.\label{eq:stability_lower_bound}
\end{equation}

Next, note that the optimality conditions of $x_{w^{(0)}}$ and $x_{w^{(1)}}$
imply that for some $\lambda^{(0)},\lambda^{(1)}\in\R^{n}$ 
\[
\vc+\sum_{j\in[k]}w_{j}^{(0)}\nabla\phi_{j}(\vx_{\vWeight^{(0)}})=\ma\lambda^{(0)}\text{ and }\vc+\sum_{j\in[k]}w_{j}^{(1)}\nabla\phi_{j}(\vx_{\vWeight^{(1)}})=\ma\lambda^{(1)}\,.
\]
Now, since $\ma^{\top}x_{w^{(0)}}=b=\ma^{\top}x_{w^{(1)}}$, this
implies
\begin{align*}
\sum_{j\in[k]}w_{j}^{(1)}\psi_{j}'(1) & =\sum_{j\in[k]}w_{j}^{(1)}\left[\nabla\phi_{j}(\vx_{\vWeight^{(1)}})^{\top}(x_{w^{(1)}}-x_{w^{(0)}})\right]=\left[c-\ma\lambda^{(1)}\right]^{\top}(x_{w^{(1)}}-x_{w^{(0)}})\\
 & =c^{\top}(x_{w^{(1)}}-x_{w^{(0)}})=\left[c-\ma\lambda^{(0)}\right]^{\top}(x_{w^{(1)}}-x_{w^{(0)}})\\
 & =\sum_{j\in[k]}w_{j}^{(0)}\left[\nabla\phi_{j}(\vx_{\vWeight^{(0)}})^{\top}(x_{w^{(1)}}-x_{w^{(0)}})\right]=\sum_{j\in[k]}w_{j}^{(0)}\psi_{j}'(0)\,.
\end{align*}
Combining with (\ref{eq:stability_lower_bound}) yields that 
\[
\frac{w_{i}^{(1)}}{\theta(1+\theta)}\leq\sum_{j\in[k]}\left(w_{j}^{(0)}-w_{j}^{(1)}\right)\psi_{j}'(0)\text{ and}\frac{w_{i}^{(0)}}{\theta(1+\theta)}\leq\sum_{j\in[k]}\left(w_{j}^{(0)}-w_{j}^{(1)}\right)\psi_{j}'(1)\,.
\]
Further, the definition of $\theta$ implies that $t\in\dom(\psi_{j})$
for all $t\in(-\theta,1+\theta)$ and $j\in[k].$ Therefore, the second
bullet of Lemma~\ref{lem:general_self_concordant} implies that 
\[
\psi_{j}'(0)\cdot(1+\theta)=\psi_{j}'(0)\cdot((1+\theta)-0)\leq\nu_{j}\text{ and }-\psi_{j}'(1)\cdot(1+\theta)=\psi_{j}'(1)\cdot((-\theta)-1)\leq\nu_{j}\,.
\]
Consequently, if $w_{j}^{(1)}\leq w_{j}^{(0)}$ for all $j$ we have
\[
\min\left\{ \min_{j\in[k]}w_{j}^{(0)},\min_{j\in[k]}w_{j}^{(1)}\right\} \leq w_{i}^{(1)}\leq\theta\sum_{j\in[k]}\left|w_{j}^{(0)}-w_{j}^{(1)}\right|\cdot\nu_{j}
\]
 and if $w_{j}^{(1)}\geq w_{j}^{(0)}$ for all $j$ we have 
\[
\min\left\{ \min_{j\in[k]}w_{j}^{(0)},\min_{j\in[k]}w_{j}^{(1)}\right\} \leq w_{i}^{(0)}\leq\theta\sum_{j\in[k]}\left|w_{j}^{(0)}-w_{j}^{(1)}\right|\cdot\nu_{j}\,.
\]
Therefore, the result in (\ref{eq:slack_simple}) holds in either
case.

Finally, we consider the case of arbitrary $\vWeight^{(0)},\vWeight^{(1)}\in\dWeight$
(i.e. where it is not necessarily the case that $w^{(0)}\geq w^{(1)}$
or $w^{(1)}\geq w^{(0)}$). In this case, we let $v_{j}=\max\{w_{j}^{(0)},w_{j}^{(1)}\}$
where max is applied entrywise. Note that $w^{(0)}\leq v$ and $v\geq w^{(1)}$
entrywise and consequently we can apply the previous result, i.e.
(\ref{eq:slack_simple}), to the pairs $(w^{(0)},v)$ and $(v,w^{(1)})$
to show that 
\[
v_{t}^{(0)}\defeq w^{(0)}+t(v-w^{(0)})\in\Omega\text{ and }v_{t}^{(1)}\defeq v+t(w^{(1)}-v)\in\Omega
\]
for all $t\in(-\theta,1+\theta)$ where (as in the previous case)
\[
\theta\geq\frac{\min\left\{ \min_{j\in[k]}w_{j}^{(0)},\min_{j\in[k]}w_{j}^{(1)}\right\} }{\sum_{j\in[k]}\nu_{j}\left|w_{j}^{(0)}-w_{j}^{(1)}\right|}\,.
\]
Consequently, since $\Omega$ is convex, considering $t=1+\gamma$
for $\gamma\in[0,\theta)$ we have
\begin{align*}
\Omega & \ni\left(\frac{\gamma}{1+2\gamma}\right)v_{t}^{(0)}+\left(\frac{1+\gamma}{1+2\gamma}\right)v_{t}^{(1)}=\left(\frac{1}{1+2\gamma}\right)\left[-\gamma^{2}w^{(0)}+(1+\gamma)^{2}w^{(1)}\right]\\
 & =w^{(0)}+\frac{(1+\gamma)^{2}}{(1+2\gamma)}\left[w^{(1)}-w^{(0)}\right]=w^{(0)}+\left[1+\frac{\gamma^{2}}{1+2\gamma}\right]\cdot\left(w^{(1)}-w^{(0)}\right)\,.
\end{align*}
Further, considering $t=-\gamma$ for $\gamma\in[0,\theta]$ we haves
\begin{align*}
\Omega & \ni\left(\frac{1+\gamma}{1+2\gamma}\right)v_{t}^{(0)}+\left(\frac{\gamma}{1+2\gamma}\right)v_{t}^{(1)}=\left(\frac{1}{1+2\gamma}\right)\left[(1+\gamma)^{2}w^{(0)}-\gamma^{2}w^{(1)}\right]\\
 & =w^{(0)}+\frac{\gamma^{2}}{(1+2\gamma)}\left[w^{(1)}-w^{(0)}\right]\,.
\end{align*}
Consequently $w^{(0)}+t(w^{(1)}-w^{(0)})\in\Omega$ for all $t\in(-\gamma,1+\gamma)$
with $\gamma=\frac{\theta^{2}}{1+2\theta}$ as desired.
\end{proof}
In the applications we consider in Section~\ref{sec:lp_alg} we have
$\arg\min_{x}f_{t}\left(\vx,\vWeight\right)=\arg\min f_{1}\left(\vx,\frac{\vWeight}{t}\right)$.
Further, since $\vWeight$ is polynomial bounded by above and below
by $m$, the ratio between old weights $w^{(1)}/t^{(2)}$ and the
new weights $w^{(t)}/t^{(2)}$ is bounded polynomially by the ratio
of $t^{(1)}$ and $t^{(2)}$ and $m$. Further, since our initial
point starts away from the boundary of the polytope Lemma~\ref{lem:slacks_bound}
implies that the distance from $x$ to the boundary can always be
bounded.
\begin{cor}
\label{cor:distance_bdy} Using the notation and assumptions in either
Theorem~\ref{thm:LPSolve_detailed} or Theorem~\ref{thm:LPSolve_detailed_dual},
we have $\phi_{i}''(\vx)\leq O(\poly(mU/\varepsilon))$ throughout
the algorithm.
\end{cor}

\begin{proof}
We prove the claim for Theorem~\ref{thm:LPSolve_detailed} as Theorem~\ref{thm:LPSolve_detailed_dual}
applies the same algorithm. Consider, Algorithm~\ref{alg:lp_solve}.
By the assumptions of Theorem~\ref{thm:LPSolve_detailed}, the initial
point $x_{0}$ has distance at least $1/U$ to each of the $\ell_{i}\leq x_{i}\leq u_{i}$
constraints. Further, by Lemma~\ref{lem:slacks_bound}, $\frac{1}{2m}\leq w\leq2$
and that $\frac{1}{2m}\leq\next w\leq2$, $x_{t_{1}}\defeq\arg\min f_{t}(\vx,\next w)$
has distance at least $\frac{1}{\poly(mU)}$ to any of the $\ell_{i}\leq x_{i}\leq u_{i}$.
By Lemma~\ref{lem:gen:phi_properties_sim}, $\phi_{i}''(\vx_{t_{1}})\leq O(\poly(mU/\varepsilon))$
and by Lemma~\ref{thm:LPSolve}, we have that $\delta_{t_{1}}(x^{(\text{new})},\next{\vWeight})\leq\frac{1}{2^{16}\log^{3}m}$.
Hence, Lemma~\ref{lem: distance gurantee} and Lemma~\ref{lem:gen:phi_properties_sim}
shows that $\phi_{i}''(\next x)\leq O(\poly(mU/\varepsilon))$. By
the same argument, we also have $\phi_{i}''(x)\leq O(\poly(mU/\varepsilon))$
for $x=x^{(\text{final})}$ and for all intermediate steps $x$.
\end{proof}
\begin{lem}
\label{lem:sequence_system_changes} Using the notation and assumptions
in either Theorem~\ref{thm:LPSolve_detailed} or Theorem~\ref{thm:LPSolve_detailed_dual}
let $\ma^{\top}\md_{k}\ma$ be the $k^{th}$ linear system that is
used in the algorithm $\code{LPSolve}$. For all $k\geq1$, we have
the following:

\begin{enumerate}
\item The condition number of $\ma^{\top}\md_{k}\ma$ relative to $\ma^{\top}\ma$
is bounded by $\poly(mU/\varepsilon)$, i.e., 
\[
\poly(\varepsilon/(mU))\ma^{\top}\ma\preceq\ma^{\top}\md_{k}\ma\preceq\poly(mU/\varepsilon)\ma^{\top}\ma
\]
\item $\norm{\log(\md_{k+1})-\log(\md_{k})}_{\infty}\leq1/10$.
\item $\norm{\log(\md_{k+1})-\log(\md_{k})}_{\weight_{p}(\md_{k}^{1/2}\ma)}\leq1/10$.
\end{enumerate}
\end{lem}

\begin{proof}
We prove the claim for Theorem~\ref{thm:LPSolve_detailed} as Theorem~\ref{thm:LPSolve_detailed_dual}
applies the same algorithm.

During the algorithm, the matrix we need to solve is of the form $\ma^{\top}\md\ma$
where $\md=\mWeight^{-1}\mPhi''(\vx)^{-1}$. Lemma~\ref{lem:lewis_rounding}
shows that $\ma^{\top}\md\ma\approx_{\poly(m)}\ma^{\top}\mPhi''(\vx)^{-1}\ma$.
Lemma \ref{lem:gen:phi_properties_sim} shows that $\phi_{i}''(\vx)\geq\frac{1}{U^{2}}$.
Also, Lemma \ref{cor:distance_bdy} shows that $\phi_{i}''(\vx)$
is upper bounded by $\poly(mU/\varepsilon)$. Thus, the condition
number of $\ma^{\top}\md_{k}\ma$ relative to $\ma^{\top}\ma$ is
bounded by $\poly(mU/\varepsilon)$. 

Now, we bound the changes of $\md$ by bound the changes of $\mPhi''(\vx)$
and the changes of $\mWeight$ separately. For the changes of $\mPhi''(\vx)$,
(\ref{eq:centrality_equivalence}) shows that $\mixedNorm{\sqrt{\vphi''(\vx)}\vh_{t}(\vx,\vWeight)}{\vWeight}\leq\mixedNorm{\mProj_{\vx,\vWeight}}{\vWeight}\delta_{t}.$
Since $\mixedNorm{\mProj_{\vx,\vWeight}}{\vWeight}\leq2$ and $\delta_{t}\leq1/80$,
we have 
\begin{eqnarray*}
\mixedNorm{\sqrt{\vphi''(\vx)}(\next{\vx}-\vx)}{\vWeight} & = & \mixedNorm{\sqrt{\vphi''(\vx)}\vh_{t}(\vx,\vWeight)}{\vWeight}\leq1/40.
\end{eqnarray*}
Applying this with Lemma~\ref{lem:gen:phi_properties_sim}, we have
\begin{eqnarray*}
\normFull{\log\left(\vphi''(\next{\vx})\right)-\log\left(\vphi''(\vx)\right)}_{\vWeight+\infty} & \leq & \left(1-\mixedNorm{\sqrt{\vphi''(\vx)}(\next{\vx}-\vx)}{\vWeight}\right)^{-1}-1\\
 & \leq & 1/36.
\end{eqnarray*}
Since $\vWeight_{i}\geq\frac{1}{2}\weight_{p}(\md^{1/2}\ma)_{i}$
for all $i$, we have
\begin{equation}
\normFull{\log\left(\vphi''(\next{\vx})\right)-\log\left(\vphi''(\vx)\right)}_{\weight_{p}(\md^{1/2}\ma)+\infty}\leq1/20.\label{eq:changes_phi}
\end{equation}

For the changes of $\mWeight$, we look at the description of $\code{centeringInexact}$.
The algorithm ensures the changes of $\log(\vWeight)$ is in $(1+\epsilon)U$
where $U=\{\vx\in\Rm~|~\mixedNorm{\vx}{\vWeight}\leq\left(1-\frac{7}{8c_{k}}\right)\delta_{t}\}$.
Since $\delta_{t}\leq1/80$ and $\vWeight_{i}\geq\frac{1}{2}\weight_{p}(\md^{1/2}\ma)_{i}$
for all $i$, we get that
\begin{equation}
\normFull{\log\left(\next{\vWeight}\right)-\log\left(\vWeight\right)}_{\weight_{p}(\md^{1/2}\ma)+\infty}\leq1/20.\label{eq:changes_of_w}
\end{equation}
The assertion (2) and (3) follows from (\ref{eq:changes_phi}) and
(\ref{eq:changes_of_w}).
\end{proof}

\end{document}